\documentclass[3p]{elsarticle}

\usepackage{amsmath, amssymb, xcolor, amsthm}
\usepackage{enumerate}
\usepackage{subcaption}
\usepackage{algorithm}
\usepackage{array,multirow}

\newcommand{\mb}[1]{\mathbf{#1}}
\newcommand{\mbb}[1]{\mathbb{#1}}
\newcommand{\mcl}[1]{\mathcal{#1}}

\newcommand{\EE}{\mathbb{E\:}}
\newcommand{\VV}{\mathbf{Var\:}}
\newcommand{\reals}{\mathbb{R}}
\newcommand{\complex}{\mathbb{C}}
\newcommand{\integers}{\mathbb{N}}

\newcommand{\dd}{\: {\rm d}}

\newcommand{\dequal}{\stackrel{d}{=}}

\newcommand{\Law}{{\rm Law}}
\newcommand{\Exp}{{\rm Exp}}
\newcommand{\Bern}{{\rm Bern}}
\newcommand{\Beta}{{\rm Beta}}
\newcommand{\Pois}{{\rm Pois}}
\newcommand{\Gamm}{{\rm Gamm}}
\newcommand{\Lap}{{\rm Lap}}

 \newtheorem{remark}{Remark}[section] 

\newtheorem{example}{Example}

 \newtheorem{theorem}{Theorem}[section]
 \newtheorem{definition}{Definition}[section]
 \newtheorem{corollary}{Corollary}[section]

\usepackage{color}
\definecolor{myorange}{rgb}{.6,.0,.0}
\definecolor{myblue}{rgb}{.0, .0,.6}
\newcommand{\myhl}[1]{{ #1}}

\renewcommand*{\appendixname}{}

\usepackage{etoolbox}

\AtEndEnvironment{example}{\null\hfill$\Diamond$}
 \title{Two Metropolis-Hastings algorithms for posterior measures with non-Gaussian
 priors in infinite dimensions}

 \author[CalTech]{Bamdad Hosseini}
    \ead{bamdadh@caltech.edu}

\address[CalTech]{Department of Computing \& Mathematical Sciences, California 
Institute of Technology, Pasadena, CA, 91125} 

\begin{document}










\begin{abstract}
We introduce two  classes of Metropolis-Hastings algorithms for sampling target measures 
that are absolutely continuous with respect to  non-Gaussian prior measures
on infinite-dimensional Hilbert spaces. In  particular, we focus on certain
classes of prior measures for which  prior-reversible 
 proposal kernels of the autoregressive type can be designed. We then use these proposal kernels to design  
 algorithms that satisfy detailed balance
 with respect to the target measures. Afterwards, we introduce a 
 new class of prior measures,
 called the Bessel-K priors, as a generalization of the gamma distribution to
 measures in infinite dimensions. The Bessel-K priors 
interpolate between  well-known  priors such as the gamma distribution and 
Besov priors and can model  sparse or compressible parameters. We present concrete instances of our
algorithms for  the Bessel-K priors in the context of
numerical examples  in density estimation,  finite-dimensional denoising  and  deconvolution on the
circle.
\end{abstract}


\maketitle





\pagestyle{myheadings}
\thispagestyle{plain}
\section{Introduction}\label{sec:introduction}


In this article we introduce two new classes of Metropolis-Hastings (MH) algorithms for sampling a
target probability measure 
that is absolutely continuous with respect to a  non-Gaussian  measure. 
We are 
particularly interested in the case where the prior  is a 
 Laplace or a generalization of the gamma distribution.
Our exposition is motivated by inverse problems on function spaces.

Markov Chain Monte Carlo (MCMC) methods are perhaps the most widely used algorithms for sampling 
complex probability measures.
However, their performance  often deteriorates as the dimension of the parameter space grows larger. 
This is a serious shortcoming when the parameter belongs to a function space.
In such cases we discretize the problem by considering a sequence of finite-dimensional measures that approximate the 
infinite-dimensional measure in a consistent manner and sample the finite-dimensional approximations instead.
 Thus, we  need algorithms that 
perform well in the limit of fine discretizations. 
To achieve this goal we pursue an algorithm that is well-defined in the infinite-dimensional limit and
 discretize it to obtain a practical finite-dimensional
algorithm.
 Examples of such  algorithms can be found in recent works: 
Cotter et al. \cite{stuart-mcmc} introduced a class of algorithms  
based on discretizations of the Langevin Stochastic Partial Differential Equation (SPDE).  
Ottobre et al. \cite{stuart-HMC} propose an infinite-dimensional version of the Hamiltonian Markov Chain (HMC) algorithm while Cui et al. \cite{marzouk-dili} present an infinite-dimensional 
 algorithm that  boosts performance by identifying subspaces of the parameter space 
that are informed by the data. More recently, Beskos et al. \cite{beskos-geometric-mcmc} studied the class of Geometric MCMC algorithms and showed that well-known algorithms such as HMC or Metropolis-Adjusted-Langevin (MALA) can be studied in a unified framework.
 
All of the above mentioned algorithms implicitly assume that the underlying prior  is absolutely continuous with respect to a Gaussian measure.  Function space
 MCMC algorithms that drop this Gaussian assumption are scarce. Wang et al. \cite{marzouk-randomize-optimize} proposed a generalization of the randomize-then-optimize strategy for inverse problems with 
 Laplace  or Besov priors \myhl{where a map is constructed to
   transform the Laplace prior to a standard Gaussian. Sampling is then done on this
   standard Gaussian and the prior-to-Gaussian map is accounted for in the likelihood. Chen et al. 
   \cite{chen2018robust} use a somewhat similar approach
   and extend the preconditioned Crank-Nicholson (pCN) algorithm of \cite{stuart-mcmc}
   to certain non-Gaussian priors such as $\ell_p$ and stable
   priors using a differentiable nonlinear map that transforms
   the prior  to a Gaussian.}  Lucka \cite{lucka-fast-Gibbs, lucka-fast-sparse}
 takes a different approach to these works and proposes fast Gibbs samplers for
inverse problems with $\ell_p$ priors in high dimensions. 
The primary contribution of this article is to introduce novel MH algorithms that
use prior-reversible proposals and satisfy detailed balance in the infinite-dimensional limit and are tailored to certain non-Gaussian priors. In contrast to previous works we directly design our algorithms
for non-Gaussian priors and do not use any mappings of the prior to a reference 
measure; hence leaving the likelihood and the forward map unchanged.
We draw  inspiration from the pCN algorithm and 
the autoregressive proposals of \cite{fox-tuning} to design  algorithms  that can  sample
a target measure when the underlying prior measure  coincides with
the limit distribution of \myhl{an} autoregressive (AR) or  random coefficient autoregressive (RCAR) process.

Let $X$ be a separable Hilbert space and suppose that $\mu$ is a  probability measure on $X$.
Throughout this article our goal is to generate samples from a target probability measure  $\nu$  on $X$ that is absolutely continuous with respect to another probability measure  $\mu$:
\begin{equation}\label{generic-target-measure}
\frac{\dd \nu}{ \dd \mu}(u) = \frac{1}{Z} \exp( - \Psi(u)),
\qquad Z = \int_X \exp( -\Psi(u)) \dd \mu(u).
\end{equation}
Here, the function $\Psi$ is assumed to be known and denotes the negative log-density of $\nu$ with respect to $\mu$. The constant $Z$ is simply a normalizing constant that ensures 
$\nu$ is a probability measure. Throughout the paper we refer to $\mu$ and $\Psi$
as the {\it prior} and the {\it likelihood} respectively,
in analogy with Bayesian inverse problems. We highlight this connection below.


Consider the problem of estimating a parameter $u \in X$ from a set of noisy 
measurements $y \in \reals^M$ given by the model 
\begin{equation}\label{additive-inverse-problem}
y = \mcl{G}(u) + \epsilon, \qquad \epsilon \sim \mcl{N}(0, \pmb{\Sigma}).
\end{equation}
Here $\mcl{G}: X \mapsto \reals^M$ is a deterministic {\it forward map}  and
$\pmb{\Sigma}\in \reals^{M\times M}$ is a positive-definite \myhl{symmetric}
matrix. The additive Gaussian noise model above is widely used in practice \cite{ calvetti, somersalo,
  stuart-acta-numerica} and it is \myhl{the} primary  model   in
this article  (see \cite{hosseini-convex, somersalo, stuart-acta-numerica} for examples 
with other noise models). 
Using \eqref{additive-inverse-problem} we can readily identify $\mu^u(y)$,
the conditional probability measure of the data $y$ given $u$, with  Lebesgue density 
$$
\frac{\dd \mu ^{u}}{\dd \Lambda}(y) = \frac{1}{ \sqrt{(2 \pi)^M |\pmb{\Sigma}|}} \exp\left(
  -\frac{1}{2} \left\| \mcl{G}( u )- y \right\|_{\pmb{\Sigma}}^2 \right).
$$
Here $\Lambda$ denotes the Lebesgue measure and we used the familiar notation $\| \cdot \|_{\pmb{\Sigma}} := \| \pmb{\Sigma}^{-1/2} \cdot \|_2$.
Now
define the {\it likelihood potential} 
\begin{equation}\label{quadratic-likelihood}
\Phi(u;y) := \frac{1}{2} \left\| \myhl{\mcl{G}(u)} - y
\right\|_{\pmb{\Sigma}}^2,
\end{equation}
and consider the infinite-dimensional version of Bayes' rule \cite{stuart-acta-numerica}  in the sense of
the Radon-Nikodym theorem \cite[Thm.~3.2.2]{bogachev1}:
\begin{equation}\label{bayes-rule}
\frac{\dd \mu^{y}}{\dd \mu_0} (u) = \frac{1}{Z(y)} \exp\left( - \Phi(u;y)
\right), \qquad \qquad Z(y) := \int_X \exp( -\Phi(u;y)) \dd \mu_0(u).
\end{equation}
Here $ \mu_0$ is the {\it prior probability measure} on $X$ that
reflects our prior knowledge regarding the parameter $u$ and $ \mu^{y}$ is the 
{\it posterior probability measure}. Note that \eqref{bayes-rule} has the
same form as \eqref{generic-target-measure}.
The posterior measure $\mu^{y}$ is considered to be the solution to the Bayesian inverse problem. 

 Generating independent samples from the
 posterior measure $\mu^{y}$ is an effective method for inference. 
 The samples can be used to approximate different statistics such as the mean, the median and standard deviations 
 that can then be used as pointwise approximations to the parameter $u$  or measures of uncertainty
 as well as computing other quantities of interest.

The secondary contribution of 
this article  is to introduce a new class of non-Gaussian prior measures called the {\it Bessel-K priors}  that generalize the Laplace and gamma distributions to infinite dimensions and appear
to be good models for sparse or compressible parameters.
 A similar class of prior measures to the Bessel-K priors 
 were introduced in \cite{hosseini-sparse} in connection to $\ell_p$-regularization in the
 variational approach to inverse 
problems when $p \in (0,1]$. The Bessel-K priors serve as a concrete example of a non-Gaussian prior 
in a well-posed inverse problem that can be sampled efficiently using our algorithms. 

The rest of this article is structured as follows: 
In 
Section~\ref{sec:MH-algorithms} we develop  two MH algorithms for simulation of 
\eqref{generic-target-measure} with  non-Gaussian prior $\mu$.
The abstract versions of our algorithms (called the RCAR and SARSD)  are introduced in
Subsections~\ref{subsec:RCAR-alg} and \ref{sec:SARSD-alg}. 
We formally introduce the Bessel-K priors in Section~\ref{sec:bessel-k}.
Lifted versions of RCAR and SARSD for the Bessel-K priors are presented in
Subsection~\ref{sec:BK-1D-algorithms} in 1D while generalizations of the
Bessel-K priors and the RCAR and SARSD algorithms to infinite dimensions are
outlined in Subsections~\ref{subsec:bessel-K-generalization-to-inf-dim}
and Subsections~\ref{subsec:sampling-with-bessel-K-Hilbert-space}.
In Section~\ref{sec:well-posedness} we briefly discuss
well-posedness  of inverse problems with Bessel-K priors
and dedicate Section~\ref{sec:numerical-experiments} to numerical experiments that
demonstrate the performance of our algorithms and properties of the Bessel-K priors.
In Subsection~\ref{subsec:density-estimation} we present a two dimensional
density estimation problem where we use our algorithms to
estimate a target density with a Bessel-K prior and give visual evidence of the
fact that our algorithms sample the correct target measures. In Subsection~\ref{sec:example-2-finite-dim-denoising} we tackle a finite-dimensional denoising problem where 
 the dimension of the parameter space and the data are increased simultaneously and the
performance of our algorithms is studied. Finally, in Subsection~\ref{sec:deconv}
we study the deconvolution problem on a function space and take a close look at how  the
the RCAR algorithm performs in the high-dimensions. In the same example we
also study the effect of hyperparameters in definition of the Bessel-K prior.

\subsection{Some notation and definitions}\label{sec:notation}
Throughout the article we assume the 
parameter space $X$ is a separable Hilbert space with inner product 
$\langle \cdot, \cdot \rangle$  and norm $\| \cdot \|_X$.
We let $P(X)$
denote \myhl{the space of  Radon {(i.e. inner regular, outer regular and locally finite)}
probability measures on $X$}. 
We use $\Lambda$ to denote the Lebesgue measure.  

Given a probability measure $\mu \in P(X)$ we define its characteristic function $\hat{\mu}:X^\ast \mapsto \complex$ via 
$$
\hat{\mu}(\varrho) := \int_X \exp(i \langle u, \varrho \rangle ) \dd \mu(u), \qquad \text{for } \varrho \in X^\ast,
$$
 where, with some abuse of notation, we have used $\langle \cdot, \cdot \rangle$ to denote the duality pairing between $X$ and $X^\ast$ following the Riesz representation theorem.

Given two probability measures $\mu,\nu \in P(X)$ we use $\nu \ll \mu$ to 
denote that $\nu$ is absolutely continuous with respect to $\mu$ (i.e. $\text{supp} \nu 
\subseteq \text{supp} \mu$) and if $\mu \ll \nu$ as well then we say $\nu$ and $\mu$
are mutually absolutely continuous or equivalent and write $\nu \sim \mu$.  We
further overload the `$\sim$' operator and for a random variable 
$u \in X$ we use $u \sim \mu$ to denote that $\mu = \text{Law}\{u\}$,
that is, $\mu$ is the law of $u$.
 We use 
the notation $u \dequal v$ whenever the random variables $u$ and $v$ have the same law up to 
sets of measure zero.

We use $\mcl{B}(X)$ to denote the Borel $\sigma$-algebra on $X$ and 
define a probability kernel $\mcl Q(\cdot, \cdot): X \times \mcl{B}(X) \mapsto [0,1]$ 
as a mapping with the following properties: 
\begin{enumerate}[(i)]
\item For every set $A \in \mcl{B}(X)$, $\mcl Q(\cdot, A): X \to [0,1]$ is measurable. 
 \item For every point $u \in X$, $\mcl Q(u, \cdot) \in P(X)$.
 \end{enumerate}
 \myhl{Note that the kernel $\mathcal Q$ above is $\sigma$-finite by definition
 since $u \mapsto \mcl Q(u, A)$ is \myhl{finite} for all sets $ A \in \mcl{B}(X)$. }

We often use the shorthand notation $\{ u_k \}$ to
denote  a sequence $\{ u_k \}_{k=1}^\infty$ of elements $u_k$ in $\mbb R$ or a Hilbert space.
We also gather the definition of some standard random variables on the real line 
that are used throughout this article in Appendix~\ref{app:standard-rvs}.

\section{Metropolis-Hastings algorithms with prior reversible
proposals}\label{sec:MH-algorithms}
Here we discuss the general theory of MH algorithms
with prior reversible proposals.
 Our
approach is based on the framework of Tierney \cite{tierney} and 
is inspired by the pCN algorithm
of \cite{stuart-mcmc}.
Recall  \eqref{generic-target-measure} defining the target measure $\nu$ that has a density with respect to a prior
$\mu$. We particularly focus on the case where 
$\mu$ is the limit distribution of an AR or RCAR process
of order one (denoted as AR(1) or RCAR(1) respectively). The key idea is to use  the transition kernel of the
underlying process to construct  $\mu$-reversible MH proposals. This property will
translate into an MH probability kernel that satisfies the detailed balance
condition with respect to $\nu$ \cite{tierney} which in turn implies the reversibility of the
overall algorithm. 

\subsection{Prior-reversible proposal kernels}\label{sec:prior-preserving-proposal-kernels}
Let us start by recalling the classic MH algorithm in general state spaces. 
Consider the space $X \times X$ with 
$\sigma$-algebra $\mcl{B}(X) \otimes \mcl{B}(X)$. For every set $A \in 
\mcl{B}(X) \otimes \mcl{B}(X)$ define the sets 
$$A^\bot:=\{ (u,v) : (v,u) \in
A\} \in \mcl{B}(X) \otimes \mcl{B}(X)  \qquad \text{and} \qquad A_u :=
\{ v : (u,v) \in A\} \in \mcl{B}(X).$$
Given a  probability kernel $ \mcl Q$ 
and a target measure $\nu$, we  define the measures $\tau$ and $\tau^\bot$ by 

\begin{equation}\label{tau-meas-def}
 \tau(A) := \int_X \mcl Q(u, A_u) \dd \nu(u), \qquad \tau^\bot(A)  =
\int_{X}  \mcl Q(u, A_u^\bot) \dd \nu(u), \qquad \forall A \in \mcl{B}(X)
\otimes \mcl{B}(X).
\end{equation}

By \cite[Prop.~1]{tierney} there exists a symmetric set $R \in \mcl{B}(X)
\otimes \mcl{B}(X)$   on which $\tau \sim \tau^\bot$. {Recall that a set $R$ in a
  vector space is
said to be symmetric  if $R = -R :=\{ -u : u \in R\}$.} The set $R$ is unique up to sets of measure
zero
for both $\tau$ and $\tau^\bot$. Furthermore, $\tau$ and $\tau^\bot$
are mutually singular on the complement of $R$. Intuitively,  $R$ consists of state pairs $(u,v)$ 
for which transition from $u$ to $v$ and vice versa is possible under
the  kernel $\mcl Q$.
Define the restrictions of $\tau$ and $\tau^\bot$ 
to the set $R$ by $\tau_R$ and $\tau_R^\bot$. Then there exists a
density $r:X\times X \mapsto \reals$ that satisfies \cite[Prop.~1]{tierney}
$$
0 < r := \frac{\dd \tau_R^\bot}{\dd \tau_R} <\infty \qquad \text{and} \qquad 
r(u,v) = ({r(v,u)})^{-1}, \qquad \forall (u,v) \in R.
$$

{We are now ready to introduce  an abstract version of the MH algorithm for sampling the measure $\nu$,
outlined in Algorithm~\ref{generic-MH-alg} where,
 following
\cite[Thm.~2]{tierney}, we choose the acceptance probability} 
\begin{equation}\label{MH-acceptance-ratio}
a(u,v) = \left\{
\begin{aligned}
&   \min\left\{ 1, r(u,v)\right\},   \qquad &&(u,v) \in R,  \\ 
& 0,  \qquad &&(u,v) \not\in R.
\end{aligned}
\right.
\end{equation}
Tierney \cite{tierney} showed that  Algorithm~\ref{generic-MH-alg}
satisfies detailed balance with respect to $\nu$.
The  absolute continuity of the
measures $\tau_R$ and $\tau^\bot_R$ is the key to constructing a reversible 
algorithm on $X$. If $\tau$ and $\tau^\bot$ were mutually singular (i.e. $R$ had measure zero) then 
the acceptance probability would be zero almost surely (a.s.).

\begin{algorithm}
\caption{Generic Metropolis Hastings (MH) algorithm with proposal kernel $\mcl Q$.}
\label{generic-MH-alg}
{\begin{enumerate}[1.]
\item Set $j = 0$ and choose $u^{(0)} \in X$. 
\item At iteration $j$ propose $v^{(j+1)} \sim \mcl  Q(u^{(j)}, dv)$. 
\item Set $u^{(j+1)} = v^{(j+1)}$ with probability $a( u^{(j)}, v^{(j+1)})$ given by \eqref{MH-acceptance-ratio}.
\item Otherwise set $u^{(j+1)} = u^{(j)}$.
\item set $j \leftarrow j +1$ and return to step 2.
\end{enumerate}}
\end{algorithm}

Now suppose the kernel $\mcl Q$ preserves the  measure $\mu$, i.e., 
\begin{equation}\label{prior-preserving-kernel}
\int_X \mcl Q(u, K) \dd \mu(u)  = \mu(K), \qquad \forall K \in \mcl{B}(X).
\end{equation}
Using \eqref{generic-target-measure} we can write the measures $\tau$ and $\tau^\bot$ as 
\begin{equation}\label{prior-preserving-kernel-product-measures}
{\tau(A) =  \frac{1}{Z}\int_X \mcl Q(u, A_u) \exp(-\Psi(u)) \dd \mu(u), \quad \tau^\bot(A)  =
\frac{1}{Z}\int_{X}  \mcl Q(u, A_u^\bot) \exp(-\Psi(u)) \dd \mu(u)},
\end{equation} 
and obtain  Theorem~\ref{prior-preserving-kernel-reversibility} below regarding their absolute continuity. The proof can be found in \cite[Thm.~22]{stuart-bayesian-lecture-notes}
and is therefore omitted. 

\begin{theorem}\label{prior-preserving-kernel-reversibility}
{ Suppose $\Psi(u)$  is continuous and locally bounded. Let $\mcl Q$ be a probability kernel
    that is reversible with respect to $\mu$, i.e.,
    \begin{equation}
      \label{prior-reversibility}
      \mcl Q(u, \dd v) \dd \mu(u) = \mcl Q(v, \dd u) \dd \mu(v),
    \end{equation}
    where the equivalence is understood in the sense of measures on $X \times X$.
 Then 
\begin{equation}\label{prior-preserving-accept-ratio}
  \frac{\dd \tau^\bot}{\dd \tau}(u,v) =
      \exp( \Psi(u)  - \Psi(v)) \qquad  \text{for} \quad (u,v) \in X.
\end{equation}}
\end{theorem}

Note that \eqref{prior-reversibility} automatically implies \eqref{prior-preserving-kernel}.
It follows from  Theorem~\ref{prior-preserving-kernel-reversibility} and \eqref{MH-acceptance-ratio} that whenever $\mcl Q$ satisfies detailed balance with respect to 
$\mu$, the acceptance probability of the MH algorithm takes the following simple form: 
\begin{equation}\label{RWSD-acceptance-ratio}
a(u,v) = \left\{
\begin{aligned}
&\min\{ 1, \: \exp( \Psi(u) - \Psi(v) ) \} \qquad && u,v \in \text{supp } \mu,\\
&0 \qquad &&u,v \not\in \text{supp } \mu.
\end{aligned}
\right.
\end{equation}
We now discuss two specific constructions of the kernel $\mcl Q$ that satisfy the properties of Theorem~\ref{prior-preserving-kernel-reversibility} for non-Gaussian priors $\mu$.

\subsection{The RCAR algorithm}\label{subsec:RCAR-alg}
Consider a random coefficient autoregressive process of the form
\begin{equation}\label{RCAR-process}
\left\{
  \begin{aligned}
    u^{(n)} &= z^{(n)} + \myhl{w^{(n)}}, \qquad u^{(0)} \sim \mu, \\
    z^{(n)} &\sim \mcl{T}_\beta(u^{(n-1)}, \dd z),
    \qquad  w^{(n)} \sim \mu_{\beta}.
\end{aligned}\right.
\end{equation}
Here we take $\beta$ to be a deterministic parameter that
parameterizes the family of probability kernels $\mcl T_\beta$ and measures $\mu_\beta \in P(X)$.
In this work we take $\beta \in (0,1)$ although more general parameterizations
are possible.

We assume  $\mcl T_\beta$ is a probability kernel defined by 
\begin{equation}
  \label{T-beta-definition}
  \mcl T_\beta(u, \dd z) = \Law \{ z=  T_\beta u\},
\end{equation}
where $T_\beta: X \mapsto X$ is a random linear operator for fixed $\beta \in (0,1)$
(see \cite{thang1998random} for
an overview of random mappings on Hilbert spaces).
We think of $T_\beta$ as the infinite-dimensional analog of a random coefficient matrix
and the  process \eqref{RCAR-process} 
as a generalization of RCAR(1) processes to Hilbert spaces \cite{nicholls2012random}.
 In light of this analogy we refer to the
measure $\mu_\beta$
as the {\it innovation}.

 Let us now define the probability kernel $\mcl Q_\beta$ via
\begin{equation}\label{RCAR-kernel-generic-form}
  \mcl Q_\beta(u, dv) = \Law \left\{ v = z + w, \quad \text{ where } \quad  z \sim \mcl T_\beta(u, \dd z),  \quad w \sim
  \mu_\beta \right\}. 
\end{equation}
If $\mcl Q_\beta$ satisfies detailed balance with respect to 
$\mu$  we can then use 
$\mcl Q_\beta$ within Algorithm~\ref{generic-MH-alg} and obtain a
well-defined  algorithm for sampling the target  $\nu$ given by \eqref{generic-target-measure}.
We refer to such an algorithm as the RCAR  algorithm
to highlight the fact that the proposal at each step coincides with the transition
kernel of an RCAR process. The generic RCAR algorithm is presented below 
in  Algorithm~\ref{RCAR-algorithm}. 

\begin{algorithm}
\caption{Random coefficient AR proposal (RCAR)}
\label{RCAR-algorithm}
Given a fixed value $\beta \in (0,1)$, probability kernel $\mathcal T_\beta$
as in \eqref{T-beta-definition},
and innovation $\mu_\beta \in P(X)$:
{\begin{enumerate}[1.]
  \item Set $j = 0$ and  draw
    $u^{(0)} \sim \mu$. 
  \item At iteration $j$ propose
    $v^{(j+1)} =  z^{(j+1)} + \myhl{w^{(j)}}$ where $z^{(j+1)} \sim
  \mcl T_\beta (u^{(j)}, \dd z^{(j+1)}),  \myhl{w^{(j)}} \sim \mu_\beta$. 
\item Set $u^{(j+1)} = v^{(j+1)}$ with probability $a( u^{(j)}, v^{(j+1)}) =
\min\{ 1, \: \exp( \Psi(u^{(j)}) - \Psi(v^{(j+1)} ) \}$.
\item Otherwise set $u^{(j+1)} = u^{(j)}$.
\item Set $j \leftarrow j +1$ and return to step 2.
\end{enumerate}}
\end{algorithm}

Of course, the assumption that $\mcl Q_\beta$ satisfies detailed balance
with respect to $\mu$ is quite strong and depends on the choice of $\mcl T_\beta, \mu_\beta$
and $\mu$. Nonetheless, these conditions hold for interesting  non-Gaussian measures
such as the gamma distribution $\Gamm(p, \sigma)$ and some of its generalizations.

\begin{example}\label{gamma-beta-invariance}
  Let $\mu = \Gamm(p, 1)$ and for $\beta \in (0,1)$ define the probability kernel 
  \begin{equation*}
    \mcl T_\beta( u, \myhl{\dd z})
    = \Law \{ \myhl{z= \zeta u} , \qquad \text{where} \qquad \myhl{\zeta} \sim \Beta (p\beta, p(1 -\beta) \},
  \end{equation*}
  and take $\mu_\beta = \Gamm(p(1 - \beta), 1)$. Then following  Appendix~\ref{sec:thinned-gamma-process} the resulting
  kernel $\mcl Q_\beta$ of the form \eqref{RCAR-kernel-generic-form} satisfies  detailed balance  with respect to $\mu$. 
\end{example}

If the measure $\mu$ is Gaussian  we can
take  $\mcl T_\beta$ to be a  deterministic kernel and the RCAR algorithm coincides
with the pCN algorithm of~\cite{stuart-mcmc}.

\begin{example}\label{ex:gaussian-RCAR}
  Let $\mu = \mcl{N}(0, \sigma^2)$ and for $\beta \in (0,1)$ define the probability kernel 
  \begin{equation*}
    \mcl T_\beta( u, \myhl{\dd z})
    = \delta_{\beta u}(\myhl{z}),
  \end{equation*}
  with  $\delta_{\beta u}$ denoting
  the \myhl{point mass} at $\beta u$.
  If $\mu_\beta = \mcl{N}(0, (1- \beta^2)\sigma^2)$  then  $\mcl Q_\beta$ of the form \eqref{RCAR-kernel-generic-form} satisfies detailed balance with respect to $\mu$
  (see \cite[Ex.~7]{stuart-bayesian-lecture-notes} for a proof). 
\end{example}

We present more concrete examples of the RCAR algorithm   in
Section~\ref{sec:bessel-k} where 
$\mu$ is taken to be a generalization of the gamma distribution that we refer to as the
Bessel-K distribution.

\subsection{The SARSD algorithm}\label{sec:SARSD-alg}
We now discuss a second strategy for constructing prior reversible proposal kernels.
Let $\tilde{\mcl Q}$ be a probability kernel with a unique fixed point
 $\mu$, i.e.,
\begin{equation}\label{prior-preserving-forward-kernel}
  \int_X \tilde{\mcl{Q}}(u, K) \dd \mu(u) = \mu(K), \qquad \forall K \in \mcl{B}(X),  
\end{equation}
but $\tilde{ \mcl{Q}}$ is not necessarily reversible
with respect to $\mu$. Let us now denote by $\tilde{ \mcl Q}^\ast$ the
time-reversal of the kernel $\tilde{ \mcl Q}$, i.e., 
\begin{equation*}
  \int_X f(u) \int_X g(v) \tilde{\mcl{Q}}(u, \dd v) \dd \mu(u)
  = \int_X g(v) \int_X f(u)\myhl{ \tilde{\mcl{Q}}^\ast(v, \dd u)} \dd \mu(v),
\end{equation*}
for all bounded and measurable functions $f$ and $g$. Assuming that  $\tilde{\mcl{Q}}^\ast$
exists we can  construct a $\mu$-reversible probability kernel $\mcl Q$ simply by symmetrizing
$\tilde{\mcl Q}$ via
\begin{equation}\label{symmetrized-kernel}
  \mcl{Q}(u, \dd v) = \Law \{ v = t z
  + \myhl{(1-t)z^\ast}\}, 
\end{equation}
where
\begin{equation*}
  z \sim \tilde{\mcl{Q}}(u,
  \dd z), \qquad \myhl{z^\ast} \sim
  \tilde{\mcl{Q}}^\ast(u, \myhl{\dd z^\ast}),
  \qquad t \sim \Bern(1/2).   
\end{equation*}
The above random variables are assumed to be independent.
We have now reduced the problem of finding a prior reversible kernel $\mcl{Q}$ to
that of finding a prior preserving kernel
$\tilde{\mcl{Q}}$ and its time-reversal $\tilde{\mcl{Q}}^\ast$.

The kernel $\tilde{\mcl Q}$ can be identified for a large class of
non-Gaussian measures $\mu$. For example, if $\mu$ is self-decomposable (SD) it follows from \eqref{SD-characteristic} that for every choice of $
\beta \in (0,1)$ there exists $\mu_\beta \in P(X)$
such that
\begin{equation*}
  \mu  = \Law \left\{ z = \beta u + \myhl{w}, \qquad \text{ where } u \sim \mu,\quad \myhl{w} \sim \mu_\beta\right\}.
\end{equation*}
Often times $\mu_\beta$ can be identified by its characteristic function.
This relationship immediately suggests that the AR transition kernel
\begin{equation}
  \label{ARSD-forward-proposal}
  \tilde{\mcl Q}_\beta(u, \dd z)
  = \Law \left\{ z = \beta u + \myhl{w},
    \qquad \myhl{w} \sim \mu_\beta \right\},
\end{equation}
satisfies \eqref{prior-preserving-forward-kernel}. Unfortunately the reverse kernel $\tilde{\mcl Q}_\beta^\ast$
in this case does not always have a closed form and must be
identified on a case by case basis. Following Example~\ref{ex:gaussian-RCAR}
we see that in the case where $\mu=\mcl{N}(0, \sigma^2)$ then $\tilde{\mcl Q}_\beta = \tilde{\mcl Q}_\beta^\ast$
and this relationship holds for Gaussian measures on Hilbert spaces as well (see \cite[Ex.~7]{stuart-bayesian-lecture-notes}). However, Weiss \cite{weiss1975time} showed that this property is
unique to Gaussian measures and  for non-Gaussian SD measures $\mu$, $\tilde{\mcl Q}_\beta \neq \tilde{\mcl Q}_\beta^\ast$. Nonetheless, we can workout the $\tilde{\mcl Q}_\beta^\ast$ kernel for
specific non-Gaussian $\mu$.

\begin{example}\label{ex:exp-symm-kernel}
  Suppose $\mu = \Exp(1)$ and fix $\beta \in (0,1)$. Following
   Appendix~\ref{app:exponential-dist} we can take
  \begin{equation}
    \label{exp-forward--kernel}
    \tilde{\mcl Q}_\beta(u, \dd z) = \Law \{ z =  \beta u + \myhl{\zeta w}, \qquad \myhl{w} \sim  \Exp(1), \quad
   \myhl{ \zeta} \sim \Bern(1 -\beta) \},
  \end{equation}
 as the forward kernel that preserves $\mu$. The reverse kernel can then be identified as
  \begin{equation}
    \label{exp-reverse--kernel}
    \tilde{\mcl Q}^\ast_\beta(u, \dd z^\ast)
    = \Law \{ z^\ast =  \min \{ u/\beta, \myhl{w}/(1-\beta)\} , \qquad \myhl{w} \sim  \Exp(1) \}.
  \end{equation}
\end{example}

To this end, for $\mu$ an SD measure with innovation $\mu_\beta$ we define the SARSD algorithm
(symmetrized autoregressive proposal for SD priors) outlined in Algorithm~\ref{SARSD-algorithm} below
that is well-defined and reversible with respect to the target measure
$\nu$ whenever the reverse kernel $\tilde{ \mcl Q}_\beta^\ast$ exists. In Section~\ref{sec:bessel-k}
we present instances of this algorithm for certain generalizations of the exponential and gamma distributions
to Hilbert spaces.

\begin{algorithm}
\caption{Symmetrized AR proposal for SD priors (SARSD)} 
\label{SARSD-algorithm}

Choose $\beta \in (0,1)$. Suppose $\mu$ is SD and let $\mu_\beta$ be its innovation.
{\begin{enumerate}[1.]
\item Set $j = 0$ and draw $u^{(0)} \sim \mu$. 
\item At iteration $j$  draw $t \sim \Bern(1/2)$.
  \item If $t = 1$ propose forward,
    $v^{(j+1)} =  \beta u^{(j)} + \myhl{w}^{(j)}$
    where $\myhl{w}^{(j)} \sim \mu_\beta$.
\item If $t = 0$ propose backward, $v^{(j+1)} \sim \tilde{\mcl Q}_\beta^\ast( u^{(j)}, \dd v^{(j+1)})$.  
\item Set $u^{(j+1)} = v^{(j+1)}$ with probability $a( u^{(j)}, v^{(j+1)}) =
\min\{ 1, \: \exp( \Psi(u^{(j)}) - \Psi(v^{(j+1)} ) \}$.
\item Otherwise set $u^{(j+1)} = u^{(j)}$.
\item Set $j \leftarrow j +1$ and return to step 2.
\end{enumerate}}
\end{algorithm}

\section{The Bessel-K prior}\label{sec:bessel-k}
In this section we introduce the Bessel-K priors as a concrete example of 
a non-Gaussian prior measure  giving rise to
target measures $\nu$ that can be efficiently sampled using the RCAR and SARSD  algorithms. 
The Bessel-K priors are  interesting by themselves for 
modelling sparse or compressible parameters. We demonstrate this feature with an example. 

\begin{example}\label{example-1}
  Suppose $\mb{G} \in \reals^{M\times N}$, $u \in \reals^N$ and consider the measurement model
$$
y = \mb{G} u + \eta, \qquad {\eta} \sim \mcl{N}(0, \mb{I}_M),
$$
where $\mb{I}_M \in \reals^{M\times M}$ is the identity matrix.
Our goal is to estimate $u$ given a realization of $y$.
Now take the prior measure 
$\mu_0$ to have Lebesgue density
\begin{equation}\label{bessel-density-proto}
\frac{\dd \mu_0}{\dd \Lambda} (u) = 
\left({\frac{{1}}{ \sqrt{\pi} \Gamma(p)  2^{p - 1/2}}} \right)^{N}
\prod_{j=1}^N \left|{u_j}\right| ^{p - 1/2}  K_{p - 1/2} \left( \left|{ u_j} \right| \right),
\end{equation}
where  $p \in (0,1]$ is a constant and $K_{\alpha}(t)$ for $\alpha \in \reals$ is the modified Bessel function of the second kind {(see Figure~\ref{fig:Gpq-density}(a))}.
Here, $u_j$ are the components of $u$ and $ \Lambda$ is the
Lebesgue measure in $\reals^N$. 

Then Bayes' rule \eqref{bayes-rule} gives the posterior measure
\begin{equation}\label{example-1-posterior}
\frac{\dd \mu^y}{\dd \Lambda} \propto \exp\left(  -\frac{1}{2} \| \myhl{\mb{G}} u - y \|_2^2 -
\left[ \sum_{j =1}^N \ln\left( | u_j|^{p-1/2} K_{p - 1/2} \left(\left| u_j \right| \right) \right)   \right]   \right).
\end{equation}
Formally, the maximizer of this density coincides with the minimizer
of the  functional
$$
\mcl{J}(z) :=\frac{1}{2} \| \myhl{\mb{G}}z - y \|_2^2+
 \sum_{j =1}^N \ln\left( | z_j|^{p-1/2} K_{p - 1/2} \left( \left| z_j \right| \right) \right).
$$
In the special case when $p=1$, we have $K_{p -1/2}(t)  = \sqrt{\frac{\pi}{2t}} \exp(-t)  $ that gives
$$
\mcl{J}(z) = \frac{1}{2} \| \myhl{\mb{G}} z - y \|_2^2 + \sqrt{\frac{\pi}{2}}
\|z\|_1.
$$
Then for $p=1$ the functional $\mcl{J}$  is precisely the $\ell_1$-regularized least-squares
functional. For $p \in (0,1)$ the term inside the logarithm is no longer  bounded 
from below at zero and so the minimizer is not well-defined. But we can consider a 
perturbed version of the $\mcl{J}$ functional
$$
\mcl{J}_\epsilon(z) :=\frac{1}{2} \| \myhl{\mb{G}}z - y \|_2^2 
+  \sum_{j =1}^N \ln\left( (| z_j | + \epsilon)^{p-1/2} K_{p - 1/2} \left( |z_j| + \epsilon \right) \right),
$$
with a small parameter $\epsilon >0$.  Since $K_{p-1/2}(t)$ has a logarithmic singularity at the origin  we conclude that the log term will
heavily penalize the modes of $z$ that are on a larger scale than $\epsilon$ and so the term involving the Bessel function is viewed as a
penalization term that enhances the sparsity of the minimizer.

 In Figure \ref{fig:generic-densities-2D} we present a
 prototypical example of the densities that arise in a
 2D version of the inverse problem at hand.
 Note that the resulting posterior can be
multimodal and concentrates around the axes which depicts the expected compressible behavior.
\end{example}

\begin{figure}[htp]
  \centering
\raisebox{.23\textwidth}{a)}
  \includegraphics[width=0.28\textwidth]{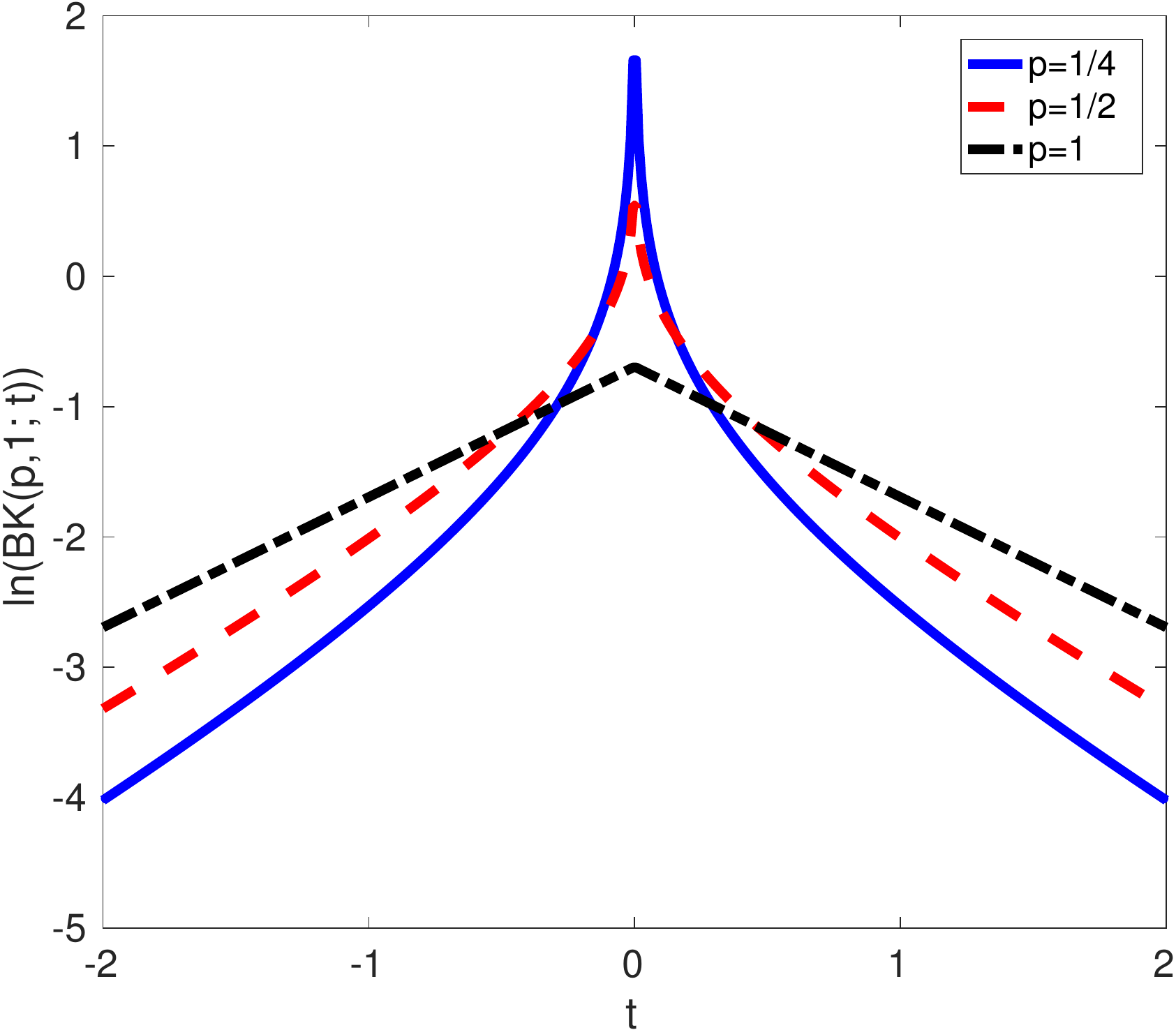}
\raisebox{.23\textwidth}{b)}
  \includegraphics[width=0.28\textwidth]{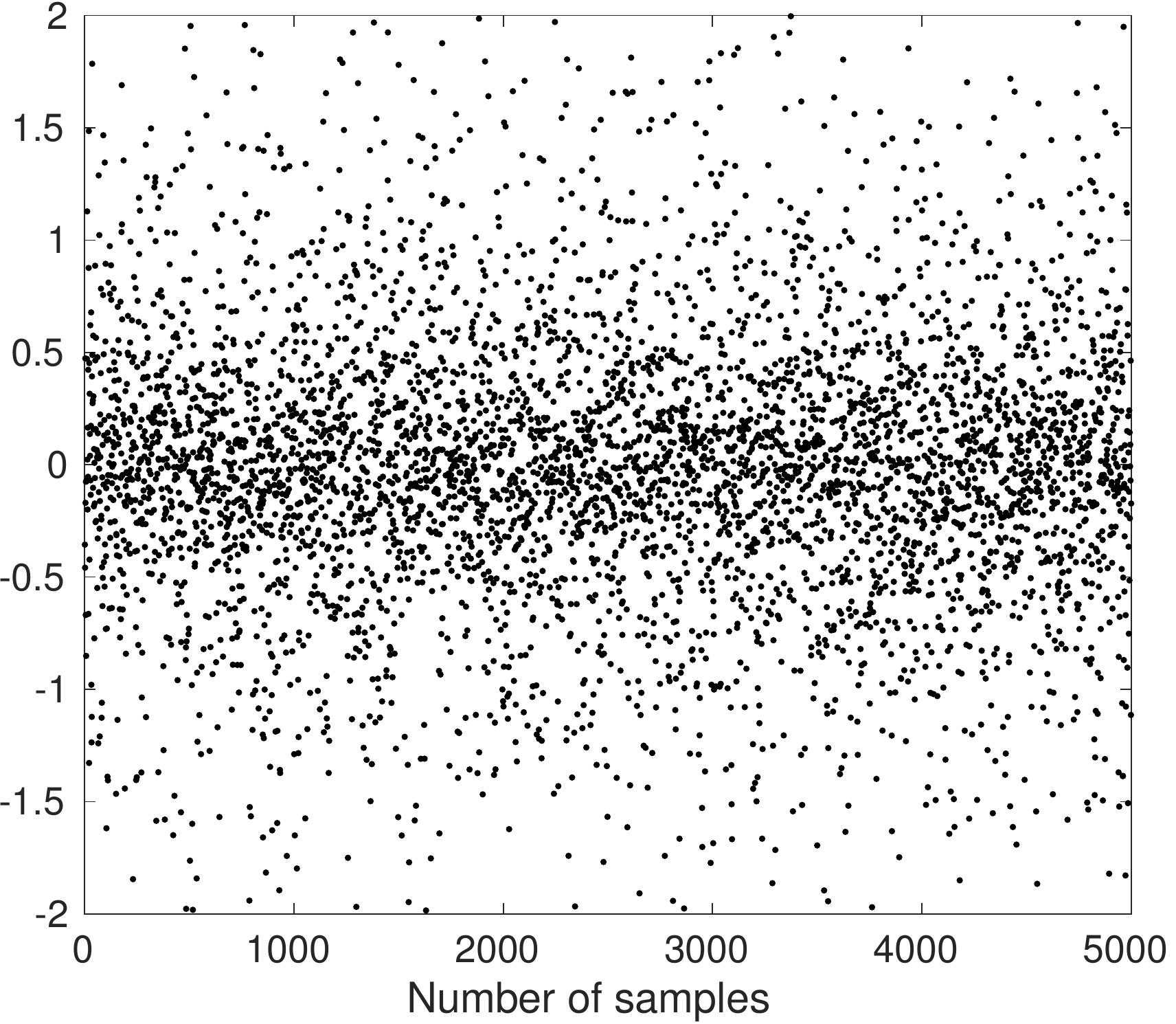}
\raisebox{.23\textwidth}{c)}
  \includegraphics[width=0.28\textwidth]{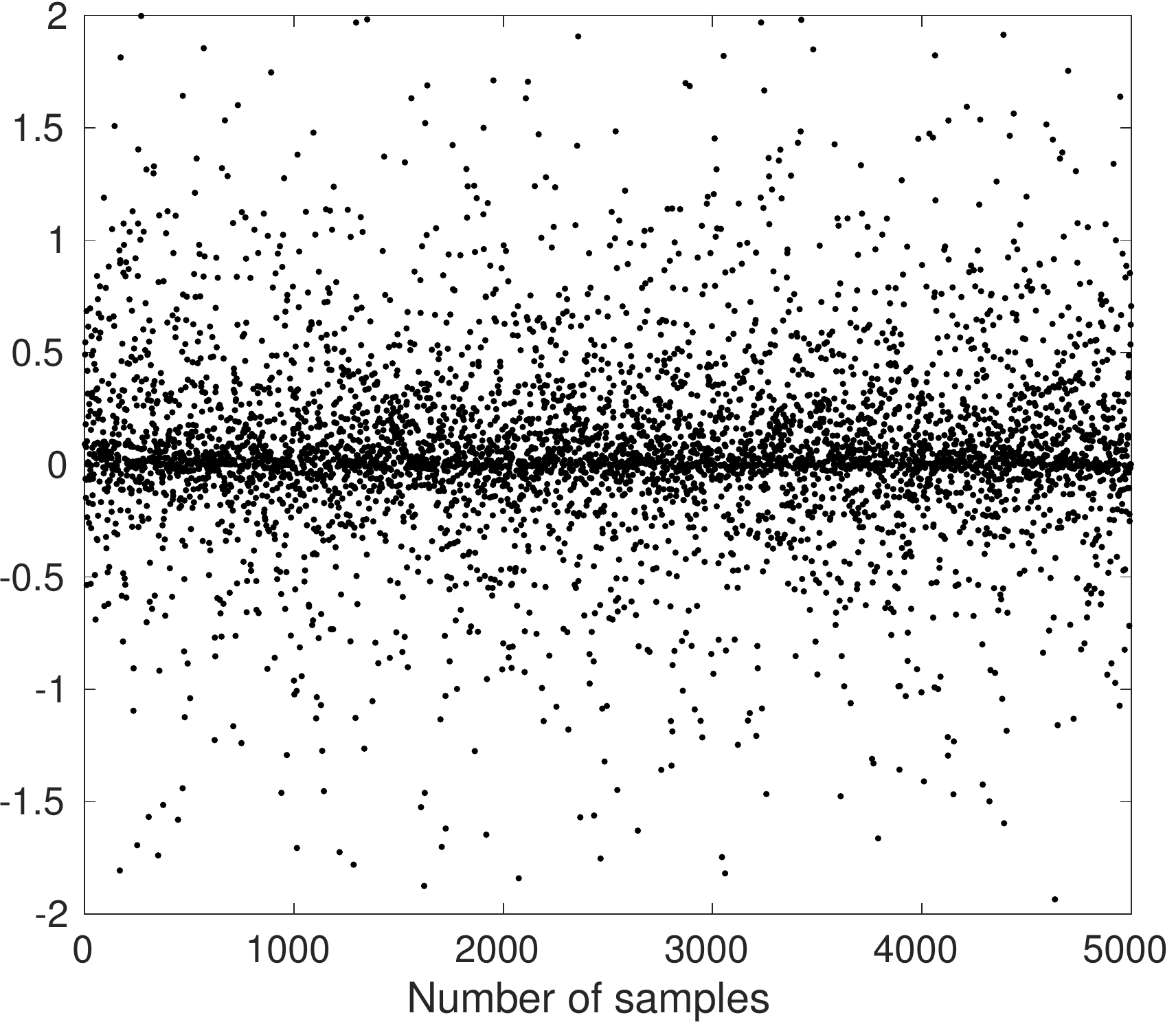}
  \caption{(a) Plot of the log of the prior density in \eqref{bessel-density-proto} for different values of
    $0<p\le 1$ along with (b) a collection of independent samples from that density with $p=1$
     and (c)  with $p = 1/2$.}
  \label{fig:Gpq-density}
\end{figure}

\begin{figure}[htp]
  \centering
  \includegraphics[width=0.28\textwidth, clip=true, trim= 1cm 0cm 1cm
  0cm]{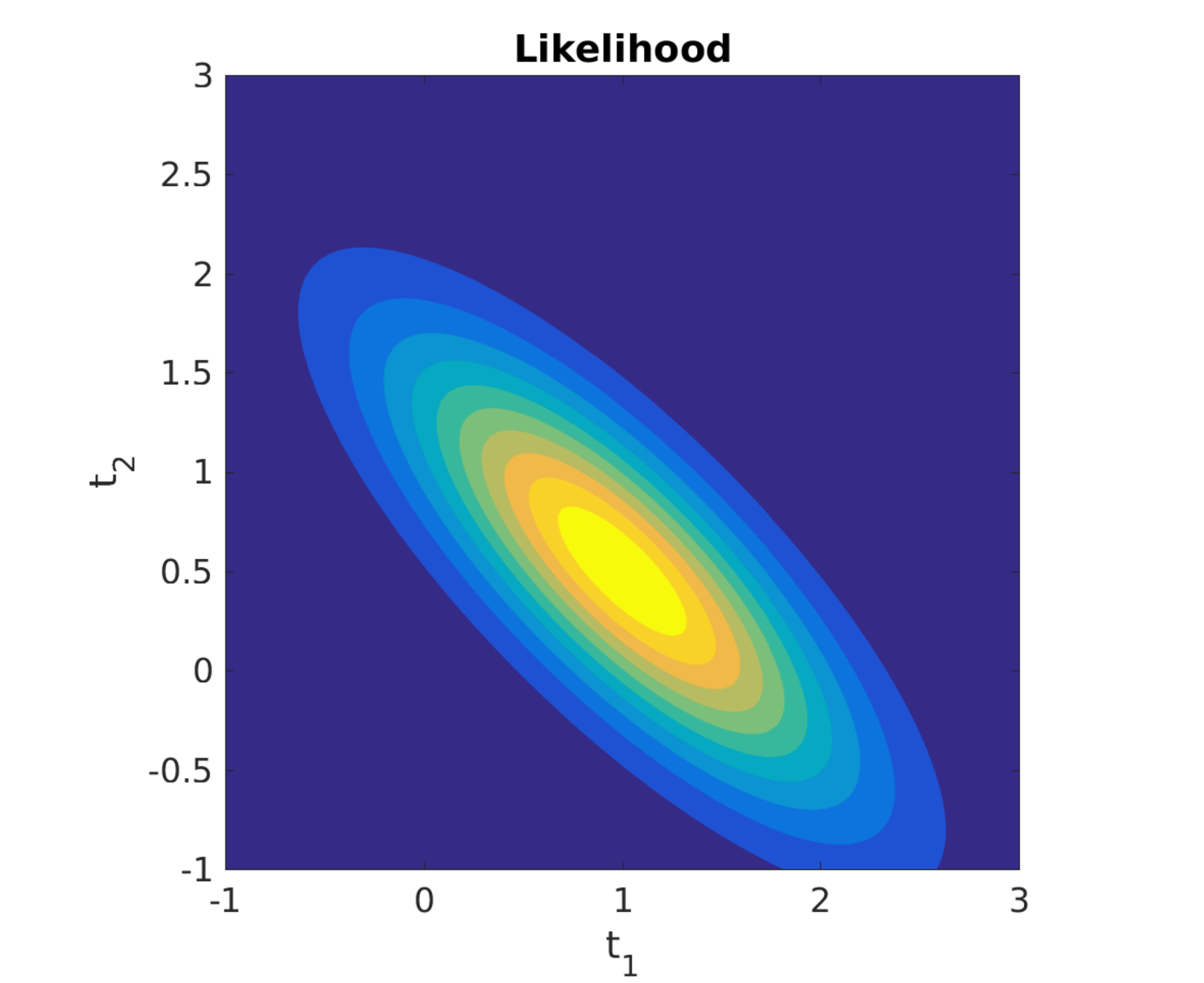}
  \includegraphics[width=0.29\textwidth, clip=true, trim= 1cm 0cm 1cm 0cm]{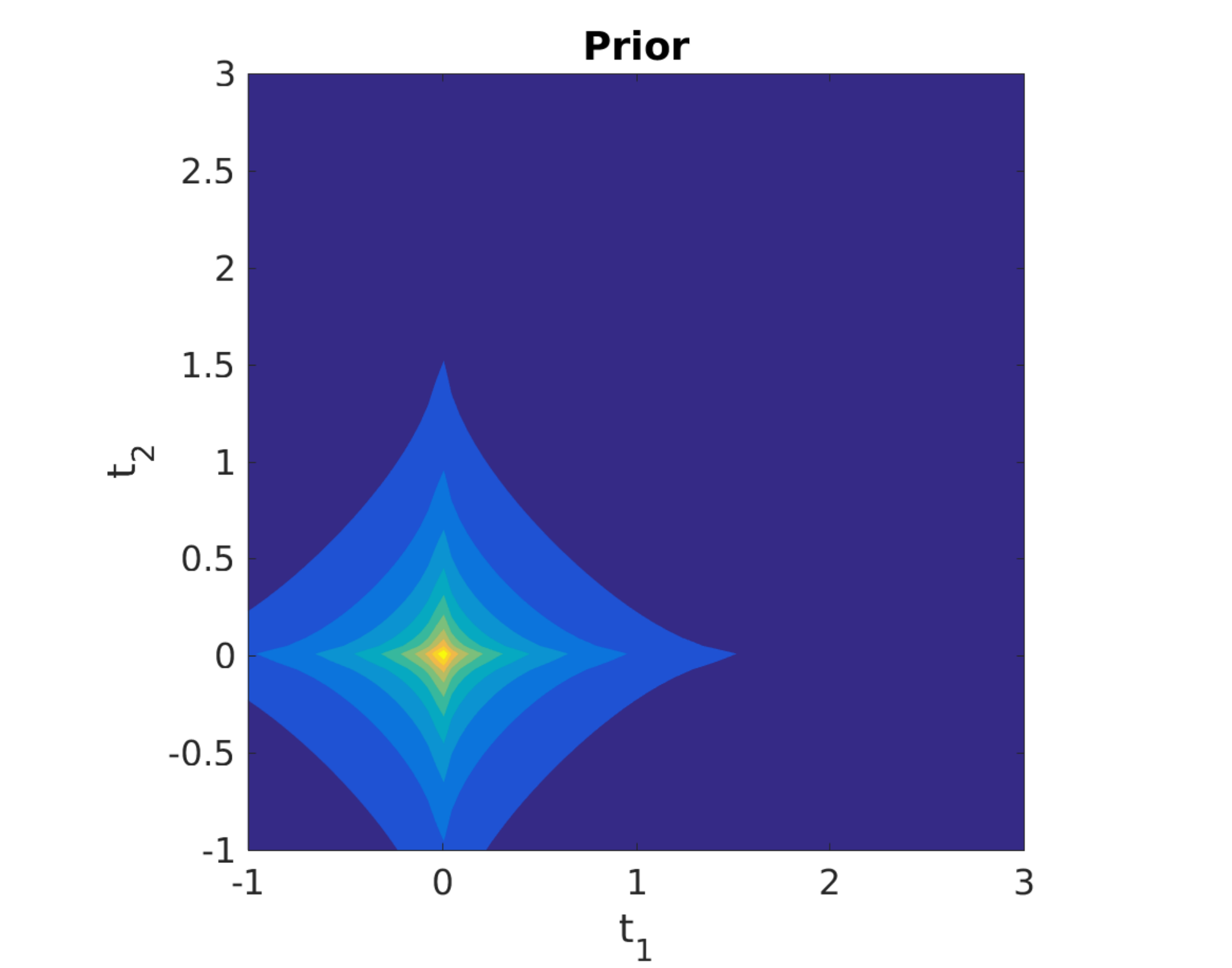}
  \includegraphics[width=0.28\textwidth, clip=true, trim= 1cm 0cm 1cm 0cm]{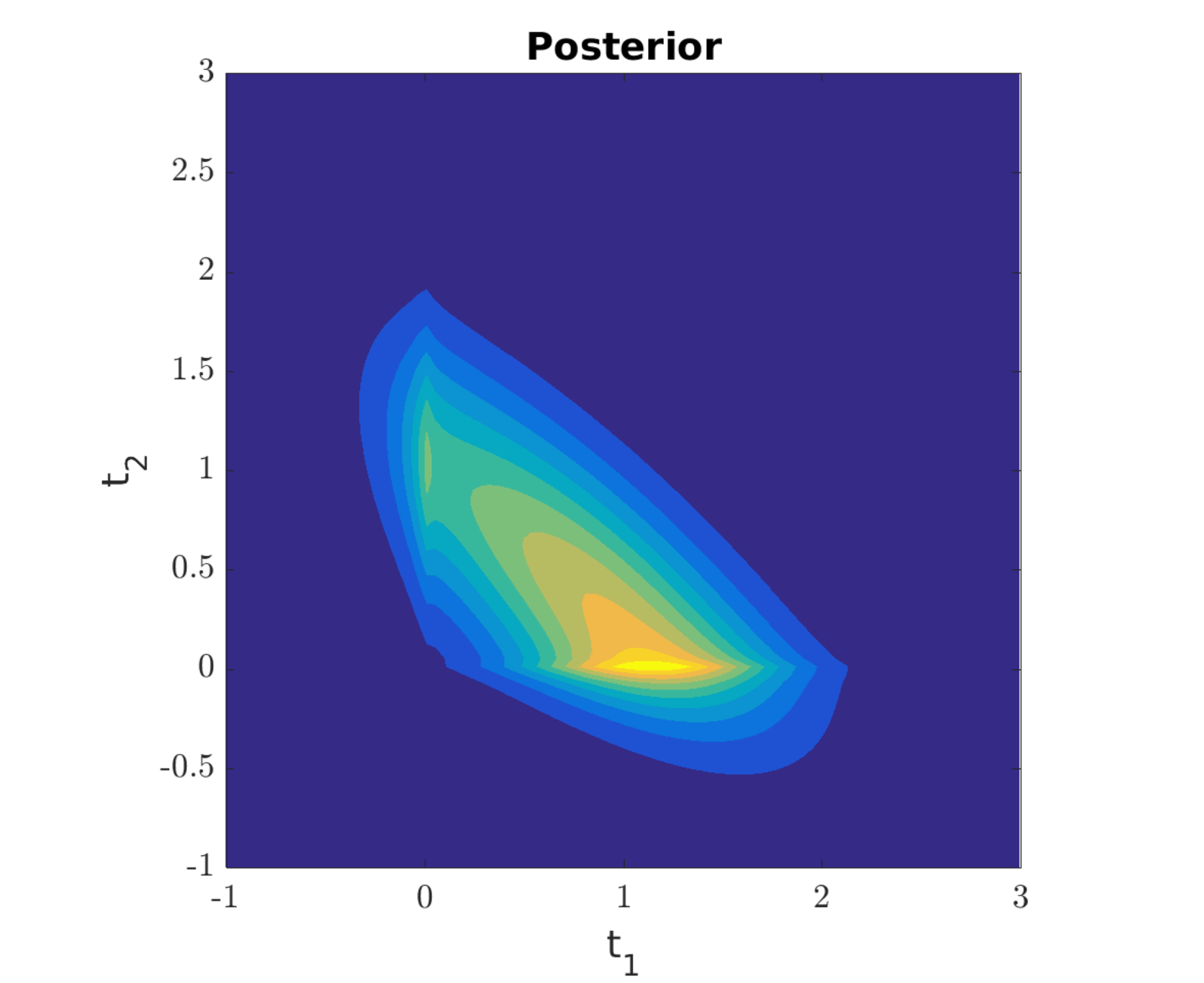}
  \caption{A prototypical example of densities that arise in Example \ref{example-1} in 2D. From left to right: The
    likelihood that arises from the additive Gaussian noise model, the
    prior with $p=2/3$ and $\sigma = 1$, and the resulting posterior density
    \eqref{example-1-posterior}. The densities are renormalized for
    better visualization.
}
  \label{fig:generic-densities-2D}
\end{figure}

\subsection{From gamma  to Bessel-K distributions}

The prior measure $\mu_0$ of \eqref{bessel-density-proto} is a finite-dimensional  Bessel-K 
prior. We now formally introduce
this prior class starting with the one-dimensional version.

\begin{definition}[Bessel-K  distribution]
A real valued random variable $\eta$ is distributed according to a Bessel-K distribution, denoted by $BK(p, \sigma)$ with shape parameter 
$p> 0$ and scale parameter $\sigma > 0$,   if its law has Lebesgue density 
$$
BK(p, \sigma;t) = {\frac{{1}}{ \sqrt{\pi} \Gamma(p) \sigma^{p + 1/2} 2^{p - 1/2}}} 
 \left|{t}\right| ^{p - 1/2}  K_{p - 1/2} \left( \left|\frac{ t}{\sigma} \right| \right) \qquad \text{for} \qquad t \in \reals,
$$
where $K_{p - 1/2}$ is the modified Bessel function of the second kind.   
\end{definition}

The above distribution was first introduced by Pearson et al. \cite{pearson} and was derived 
as the law of the difference  
of two gamma random variables by  Mathai \cite{Mathai}. It is also discussed in \cite[Sec.~4]{kotz-laplace} as a 
generalization of the Laplace distribution. This distribution  is also referred to as the {\it generalized Laplace 
distribution} or the {\it variance gamma model}  but we prefer the term {\it Bessel-K distribution} to 
avoid confusion with other generalizations of the Laplace distribution and also emphasize the fact that the Lebesgue density is given 
in the form of a modified Bessel function. We also note that the Bessel-K density 
closely resembles the Dirichlet-Laplace prior of \cite{bhattacharya2015dirichlet}
for 
sparse parameters.

Let us summarize some useful facts about the Bessel-K distributions. Proofs of these 
results can be found in \cite[Sec.~4]{kotz-laplace}. 
If $\eta \sim BK(p, \sigma)$  then 
\begin{equation}\label{bessel-k-samples}
\eta \dequal {\sigma}(\xi - \xi'),
\end{equation}
where $\xi$ and $\xi'$ are independent $\Gamm(p, 1)$ random variables. Using this observation, the expression  \eqref{gamma-characteristic-function} and the fact that the characteristic function of the sum of two random variables 
is the product of their characteristic functions, we immediately have 
$$
{\widehat{BK}(p, \sigma;s)= \left( 1 + (s \sigma)^2\right)^{-p}.}
$$

{Observe that the $BK(1, \sigma)$ distribution coincides with 
  $\Lap(\sigma)$. Furthermore, the Bessel-K class is closed
  under convolutions, i.e.,
  \begin{equation*}
    BK(p_1, \sigma) \ast BK(p_2, \sigma) = BK(p_1 + p_2, \sigma).
  \end{equation*}
 Since the gamma distribution has 
bounded moments of all \myhl{orders} then so does the Bessel-K distribution. In particular, if $\eta \sim BK(p, \sigma)$ then 
$$
\EE \eta = 0, \qquad \VV \eta = \EE \eta^2 =  \myhl{2p\sigma^2 }.
$$

Given \eqref{bessel-k-samples} and the fact that the gamma distribution is SD we 
deduce that the Bessel-K distribution is also SD. For  $\beta \in (0,1)$, we have
the decomposition 
\begin{equation}\label{BK-innovation}
\eta \dequal \beta \eta' +  \eta_\beta \qquad \text{where} \qquad \eta_\beta \dequal  \xi_\beta - \xi'_\beta. 
\end{equation}
Here $\eta, \eta' \sim BK(p, \sigma)$ and following \eqref{gamma-self-decomposable}
$\xi_\beta, \xi'_\beta \sim \Gamm_\beta(p, \sigma)$ where $\Gamm_\beta(p, \sigma)$ is identified
in \eqref{gamma-innovation}. The random variables $\eta, \eta', \xi_\beta, \xi_\beta'$ are all independent.
We refer to $\eta_\beta$ as the innovation  of 
$\eta$ and denote its law by $BK_\beta(p, \sigma)$. Using \eqref{gamma-innovation-characteristic-function} and \eqref{BK-innovation} we can show
$$
\widehat{BK}_\beta(p, \sigma; s) = \left( \beta^2 + \frac{1 - \beta^2}{{1 + (s\sigma)^2}} \right)^p. 
$$



The Bessel-K distributions are suitable candidates for prior measures in Bayesian inverse
problems given that 
they have bounded moments of all order and so result in 
well-posed inverse problems in the context of the Theorem~\ref{well-posedness}. 
 We will see shortly that this property is inherited by 
 certain infinite-dimensional generalizations of this distribution as well.
 Furthermore, the Bessel-K 
distribution is singular at the origin (see Figure \ref{fig:generic-densities-2D}) meaning that a 
notable portion of its probability mass is concentrated in a neighbourhood of the origin which is a desirable 
 in modelling  compressible parameters (see also \cite{bhattacharya2015dirichlet}
for a detailed analysis of the shrinkage properties of the Dirichlet--Laplace prior 
which is closely related to the Bessel-K distribution).

\subsection{Sampling with Bessel-K priors  in 1D}\label{sec:BK-1D-algorithms}
We now present two prior reversible proposal kernels for the $BK(p,1)$ distributions in
1D. We derive an RCAR proposal using the relationship between
gamma and beta distributions followed by a SARSD proposal 
using
the fact that the $BK(p,1)$  are SD. Although, the SARSD algorithm
is limited to integer shape parameters $p$ 
due to challenges in identifying the reverse kernel $\tilde{\mcl Q}^\ast_\beta$.

\subsubsection{The lifted RCAR algorithm}

Following Appendix~\ref{sec:thinned-gamma-process} we have for $\myhl{u, v} \sim \Gamm(p,\sigma)$
and any $\beta \in (0,1)$
that
\begin{equation*}
  \myhl{ v \dequal \zeta u + w},
\end{equation*}
where $\myhl{\zeta} \sim \Beta(p\beta, p(1- \beta))$
and $\myhl{w}\sim \Gamm(p(1- \beta), \sigma)$
and all random variables are independent. This in turn suggests
a time-reversible RCAR proposal kernel for $\Gamm(p,\sigma)$ distributions
\begin{equation}\label{Gamma-RCAR-Proposal}
  \mcl Q_\beta(u, \myhl{\dd v})
  = \Law\{ \myhl{v= \zeta u +w}\},
\end{equation}
where
\begin{equation*}
 \myhl{\zeta} \sim \Beta(p\beta, p(1-\beta))\quad \text{and} \quad \myhl{w} \sim \Gamm(p(1-\beta), \sigma)\}.
\end{equation*}
By realizing the $BK(p, \sigma)$ distribution as the law of difference of
two independent gamma random variables we can now lift our 1D sampling problem to 2D
and obtain Algorithm~\ref{BK-RCAR-1D-algorithm} for target measures $\nu$ of the form
\eqref{generic-target-measure} with $\mu = BK(p, \sigma)$. We note that
there is an added memory overhead associated \myhl{with} Algorithm~\ref{BK-RCAR-1D-algorithm}
since we need to keep track of the two Markov chains $\{u_1^{(k)}\}$ and
$\{ u_2^{(k)} \}$ rather than a single chain for $\{ u^{(k)}\}$. 

\begin{algorithm}
\caption{Lifted RCAR algorithm for $BK(p,\sigma)$ priors in 1D}
\label{BK-RCAR-1D-algorithm}
Choose $\beta \in (0,1)$ and suppose $\mu = BK(p, \sigma)$ with $p, \sigma >0$.
In the following all random variables are drawn independently.
{\begin{enumerate}[1.]
  \item Set $j = 0$, draw $\myhl{u_1^{(0)}, u_2^{(0)}} \sim \Gamm(p, \sigma)$ and
    set $\myhl{u^{(0)} = u_1^{(0)}
  - u_2^{(0)}}$. 
\item At iteration $j$ propose
  \begin{equation*}
    \begin{aligned}
      \myhl{v_1^{(j+1)}} &= \myhl{\zeta_1 u_1^{(j)} + w_1, \qquad \zeta_1 \sim \Beta(p\beta, p(1- \beta)),
      w_1 \sim \Gamm(p(1- \beta), \sigma),}\\
      \myhl{v_2^{(j+1)}} &= \myhl{\zeta_2 u_2^{(j)} + w_2, \qquad
      \zeta_2 \sim \Beta(p\beta, p(1- \beta)),
        w_2 \sim \Gamm(p(1- \beta), \sigma),} \\
        \myhl{v^{(j+1)}} &= \myhl{v_1^{(j+1)} - v_2^{(j+1)}}.
  \end{aligned}
\end{equation*}
\item With probability
  \begin{equation*}
a( u^{(j)}, v^{(j+1)}) =
  \min\{ 1, \: \exp( \Psi(u^{(j)}) - \Psi(v^{(j+1)} ) \},
\end{equation*}
set
   $u^{(j+1)} = v^{(j+1)}, \myhl{u_1^{(j+1)} = v_1^{(j+1)}, u_2^{(j+1)} = v_2^{(j+1)}}$.
\item Otherwise set $u^{(j+1)} = u^{(j)}, \myhl{u_1^{(j+1)} = u_1^{(j)}, u_2^{(j+1)} = u_2^{(j)}}$.
\item Set $j \leftarrow j +1$ and return to step 2.
\end{enumerate}}
\end{algorithm}

\subsubsection{The lifted SARSD algorithm for integer $p$}
Next, we present a lifted version of the SARSD algorithm for $BK(p, \sigma)$ priors
when $p \in \mbb N$. As mentioned in Subsection~\ref{sec:SARSD-alg} the main
challenge in designing prior-reversible kernels in this case lies in
identifying the reversal of the AR proposals of the form \eqref{ARSD-forward-proposal}.

Note that given $p \in \mbb N, \sigma >0$ and $\eta \sim BK(p,\sigma)$, we have
\begin{equation*}
  \eta \dequal \sum_{j=1}^p \xi_j - \sum_{j=p+1}^{2p} \xi_j, \qquad \xi_j \stackrel{iid}{\sim} \Exp(\sigma).
\end{equation*}
Then using the fact that the class of SD measures is closed under linear transformations
and the results in  Appendix~\ref{app:exponential-dist} we can identify the
innovation $BK_\beta(p,\sigma)$ by the relationship
\begin{equation*}
  \eta_\beta \dequal \sum_{j=1}^p \xi_{\beta,j} - \sum_{j=p+1}^{2p} \xi_{\beta, j}, \qquad \xi_{\beta, j} \stackrel{iid}{\sim} \Exp_\beta(\sigma),
\end{equation*}
with the $\Exp_\beta(\sigma)$ distribution identified by \eqref{exponential-innovation}.
This suggests a forward proposal kernel to update  $\eta$
by updating the $\xi_j$ independently using the forward kernel
given by \eqref{exp-forward-kernel}. Since each $\xi_j$ is an
exponential random variable we can identify their reverse kernel by \eqref{exp-reverse-kernel}.
We can then use the forward and reverse kernels for the $\xi_j$ to construct a lifted version of the SARSD algorithm
for $BK(p, \sigma)$ priors with integer $p$ as outlined in Algorithm~\ref{BK-SARSD-1D-algorithm}.

\begin{algorithm}
\caption{Lifted SARSD algorithm for $BK(p,\sigma)$ priors in 1D}
\label{BK-SARSD-1D-algorithm}
Choose $\beta \in (0,1)$ and suppose $\mu= BK(p, \sigma)$ for $p \in \mbb N$ and $\sigma >0$.
In the following $k=1, \dots, 2p$ and all random variables are drawn independently.
{\begin{enumerate}[1.]
  \item Set $j = 0$, draw $\myhl{u_k}^{(0)} \sim \Exp(1)$
     and set
    $u^{(0)} =
    \sigma \left( \sum_{k=1}^{p}\myhl{ u_k}^{(0)} - \sum_{k=p+1}^{2p} \myhl{u_k}^{(0)}\right).$

  \item At iteration $j$ draw $t \sim \Bern(1/2),
    \myhl{w_k} \sim \Exp(1), \myhl{\zeta_k} \sim \Bern(1 - \beta)$. 

    \item If $t= 1$ propose forward
  \begin{equation*}
    \begin{aligned}
      \myhl{v_k}^{(j+1)} &= \beta \myhl{u_{k}}^{(j)}
      + \myhl{\zeta_k w_k}.
  \end{aligned}
\end{equation*}

\item If $t = 0$ propose backward
  \begin{equation*}
    \begin{aligned}
      \myhl{v_k}^{(j+1)} &= \min\{ \myhl{u_{k}}^{(j)}/\beta, \myhl{w_k}/(1- \beta)\}.
  \end{aligned}
\end{equation*}
\item Set $v^{(j+1)} = \sigma \left( \sum_{k=1}^{p}
    \myhl{v_k}^{(j+1)} - \sum_{k=p+1}^{2p} \myhl{v_k}^{(j+1)}\right).$
  
\item With probability
  \begin{equation*}
a( u^{(j)}, v^{(j+1)}) =
  \min\{ 1, \: \exp( \Psi(u^{(j)}) - \Psi(v^{(j+1)} ) \},
\end{equation*}
set
   $u^{(j+1)} = v^{(j+1)}, \myhl{u_k}^{(j+1)} = \myhl{v_k}^{(j+1)}$.
\item Otherwise set $u^{(j+1)} = u^{(j)}, \myhl{u_k^{(j+1)} = u_k^{(j)}.}$
\item Set $j \leftarrow j +1$ and return to step 2.
\end{enumerate}}
\end{algorithm}

\begin{remark}\label{gamma-prior-algorithms}
  Note that Algorithms~\ref{BK-RCAR-1D-algorithm} and \ref{BK-SARSD-1D-algorithm}
  can be  easily modified  to accommodate  $\Gamm(p,\sigma)$ priors by
  setting $v^{(j+1)} = \myhl{v_1}^{(j+1)}$ in step 2 of Algorithm~\ref{BK-RCAR-1D-algorithm}
  or by setting $v^{(j+1)} = \sigma \sum_{k=1}^p \myhl{v_k}^{(j+1)}$ in
  step 5 of Algorithm~\ref{BK-SARSD-1D-algorithm}. We use such
   algorithms in   Subsection~\ref{sec:example-2-finite-dim-denoising}.
\end{remark}

\begin{remark}\label{remark:RCAR-vs-SARSD}
  Algorithm~\ref{BK-SARSD-1D-algorithm} is more limited in comparison to
  Algorithm~\ref{BK-RCAR-1D-algorithm} in two main aspects. First, the SARSD
  algorithm requires lifting the parameter space to $2p$ dimensions as compared
  to $2$ dimensions in the case of RCAR. Secondly,  SARSD
   is limited to integer values of $p$ while RCAR remains valid for
  all $p >0$. However, to the best of our knowledge the convergence properties
  of these algorithms are unknown beyond reversibility, and so it is difficult to decide which
  algorithm performs better in practice. In Section~\ref{sec:numerical-experiments} we
  compare statistical performance of the two algorithms in the context of
  some numerical experiments. 
\end{remark}

\subsection{Generalization to infinite dimensions}\label{subsec:bessel-K-generalization-to-inf-dim}

We now generalize the Bessel-K distributions and the lifted RCAR and SARSD algorithms
to measures on Hilbert spaces with an orthonormal basis. 
We recall a technical result concerning product priors on Hilbert spaces. 

\begin{theorem}[{\cite[Thm.~2.3 and 2.4]{hosseini-sparse}}]\label{product-prior-properties}
Let $X$ be a Hilbert space with an orthonormal basis $\{ r_k\}$ and consider the random variable 
$u = \sum_{k=1}^\infty \gamma_k \xi_k r_k$
where $\{ \gamma_k\} \in \ell^2$ and $\{\xi_k\}$ is a sequence of i.i.d. random variables 
in $\reals$ distributed according to a Radon measure and with  bounded raw moments of order $q \ge 2$. Then 
\begin{equation}
  \label{product-prior}
 { \mu = \Law \left\{ u =  \sum_{k=1}^\infty \gamma_k \xi_k r_k \right\} \in P(X),}
\end{equation}
  and $\| u\|_X < \infty$ $\mu$-a.s.
and {$\| \cdot \|_X \in L^q(X, \mu)$}. 
\end{theorem}

Since the Bessel-K distributions have bounded variance we can immediately generalize them to infinite dimensions.


\begin{definition}[$BK(p, \mcl{R})$ prior]\label{SG-prior-definition}
Given a constant $p> 0$ and a Hilbert-Schmidt operator $\mcl{R}:X \mapsto X$ with 
eigenvalues $\{\gamma_k \} \in \ell^2$
and eigenvectors $\{ r_k \}$,
 we define the $BK(p, \mcl{R})$  prior
as the law of the random variable
\begin{equation}\label{SG-prior-expansion}
u = \sum_{k=1}^\infty \gamma_k \eta_k r_k,
\end{equation}
where $\{ \eta_k \}$ is an i.i.d.
sequence of $BK(p,1)$ random
variables.
\end{definition}

The definition of the  $BK(p, \mcl{R})$ prior is inspired by the \myhl{Karhunen-Lo\`{e}ve} expansion of
Gaussian random variables \cite[Thm.~3.5.1]{bogachev-gaussian}.
The following theorem summarizes some basic facts about $BK(p, \mcl R)$ priors and
 is a direct consequence of Theorem~\ref{product-prior-properties}.

\begin{theorem}\label{SG-prior-properties}
Suppose $\mu = BK(p, \mcl{R})$ then $\mu \in P(X)$,
 $\| \cdot \|_X < \infty$ $\mu$-a.s. and $\| \cdot\|_X \in L^q(X, \mu)$  
for {all} $q \in \integers$.
\end{theorem}
Similar to their finite-dimensional counterparts, the $BK(p, \mcl{R})$ priors are also SD. 
\begin{theorem} \label{SG-prior-decomposition}
The $BK(p, \mcl{R})$ priors for  $p >0$  are SD.
Given $\beta \in (0,1)$ we
have 
$$ \widehat{BK}(p, \mcl{R};\varrho) = \widehat{BK}(p, \beta\mcl{R}; \varrho)
\widehat{BK}_\beta(p, \mcl{R};\varrho) \qquad \forall \varrho \in X^\ast.$$
Here, $\widehat{BK}_\beta(p,\mcl{R}; \cdot): X^\ast \mapsto \mbb{C}$ is the characteristic function of a
probability measure $BK_\beta(p, \mcl{R})$ (the innovation  of $BK(p,
\mcl{R})$)  that coincides with the  law of the random variable 

\begin{equation}\label{innovation-expansion}
v = \sum_{k=0}^\infty \gamma_k {\eta_{\beta,k}} r_k,
\end{equation}
where $\{ {\eta_{\beta,k}}\}_{k=0}^\infty$ is an
i.i.d. sequence of innovations with distribution $BK_\beta(p, 1)$ identified  by \eqref{BK-innovation} and 
\eqref{gamma-innovation}. 
\end{theorem}
\begin{proof}
Let $\mu = BK(p, \mcl{R})$ and consider $\varrho \in X^\ast$ and denote its Riesz representer in $X$ with
$\rho$. Then $\rho = \sum_{k=0}^\infty \rho_k r_k $  where
$\rho_k = \langle \rho, r_k \rangle$ following the assumption that $\{ r_k\}$ form
an orthonormal basis in $X$. 
Then
$$
\begin{aligned}
  \widehat{\mu}(\varrho) = \int_X \exp( i \langle u, \rho \rangle )
  \dd \mu(u) = \EE \exp\left( i \sum_{k=0}^\infty \gamma_k\rho_k 
    \eta_{k}\right) = \prod_{k=0}^\infty \EE \exp( i \rho_k \gamma_k\eta_k)
 = \prod_{k=0}^\infty \widehat{BK}(p, 1; \gamma_k \rho_k).
\end{aligned}
$$
However, we have  $\widehat{BK}(p,1; s) =
\widehat{BK}(p,1;\beta s) \widehat{BK}_\beta(p, 1)$ for any $\beta \in (0,1)$ and so
we can write
$$
\begin{aligned}
  \widehat{\mu}(\varrho) & = \prod_{k=0}^\infty \widehat{BK}(p, 
  1; \beta \gamma_k \rho_k)
  \widehat{BK}_\beta( p, 1; \gamma_k \rho_k) \\& = \left( 
\prod_{k=0}^\infty \widehat{BK}(p, 
  1; \beta \gamma_k \rho_k)
 \right) \left( \prod_{k=0}^\infty
  \widehat{BK}_\beta( p, 1; \gamma_k \rho_k)
    \right).
\end{aligned}
$$
At this point it is straightforward to check that the term in the first
bracket corresponds to the characteristic function of the pushforward measure $\mu \circ \beta^{-1}$
evaluated at $\rho$ while the second term coincides with the
characteristic function of  the random variable $v \sim BK_\beta(p, \mcl R)$  evaluated at $\varrho$. It follows from Theorem~\ref{product-prior-properties} that the law of $v$ belongs to $P(X)$. 
Then the claim 
follows from the fact that 
two Radon probability measures on a separable Hilbert
space are equivalent when their characteristic functions coincide pointwise.
\end{proof}

\subsubsection{Connection to Besov priors}\label{sec:besov-connection}

 When $X = L^2(\mathbb{T}^d)$ and for a specific choice of the operator $\mcl{R}$, the $BK(1, \mcl{R})$ priors coincide with a certain subset of   Besov priors   \cite{lassas-invariant, dashti-besov}. 
Let us recall the definition of this prior class.

\begin{definition}[$B^s_{qq}(\mbb T^d)$ prior] \label{besov-prior}
Suppose $1\le q < \infty$,  $s>0$ and
$\{ r_k\}$ is a $\ell$-regular wavelet basis for $L^2(\mathbb{T}^d)$ with $ \ell > s$. Let 
\begin{equation}\label{Besov-prior}
\mu = {\rm Law} \left\{ u = \sum_{k=0}^\infty (k+1)^{- \left( \frac{s}{d} + \frac{1}{2} - \frac{1}{q} \right)} \xi_k  r_k \right\},
\end{equation}
where $\{\xi_k\}$ is a sequence of real valued i.i.d. random variables with Lebesgue density proportional to
$$
\exp\left( -\frac{1}{2} |t|^q \right) \quad \text{for} \quad t \in \reals.
$$
Then $\mu$ is a $B^s_{qq}(\mbb{T}^d)$ prior. Furthermore, $\| u\|_{B^s_{qq}(\mbb{T}^d)} < \infty$ a.s. and $\EE \exp( \kappa \| u\|^q_{B^s_{qq}(\mbb{T}^d)} )< \infty$ for any $\kappa \in (0, 1/2)$ where,
\begin{equation*}
{ \| u\|_{B^s_{qq}(\mbb{T}^d)} := \left( \sum_{k=0}^\infty (k+1)^{\left( \frac{sq}{d} + \frac{q}{2} - 1 \right)} \langle u,  r_k \rangle^q \right)^{1/q}. }
\end{equation*}
\end{definition}
 
Now consider
the case where $q =1$ and $s$ is large enough so that $\frac{s}{d} - \frac{1}{2} \ge 1$. Then $ \{ (k+1)^{- \left( \frac{s}{d} - \frac{1}{2}  \right)} \} \in \ell^2$ and the $B^s_{11}(\mbb{T}^d)$ prior coincides with a $BK(1, \mcl{R})$ prior  on $L^2(\mbb{T}^d)$
with 
$$
{\mcl{R}(v) := \sum_{k=0}^\infty (k+1)^{- \left( \frac{s}{d} - \frac{1}{2} \right)} \langle v, r_k\rangle r_k.}
$$

Since Laplace random variables are SD  we can use the same argument as in the 
proof of Theorem~\ref{SG-prior-decomposition} to infer that the $B^s_{11}(\mbb{T}^d)$ priors are also SD.
Furthermore, the innovation of the $B^s_{11}(\mbb{T}^d)$ prior coincides with the law of the 
random variable 
\begin{equation}\label{Besov-innovation}
{v = \sum_{k=0}^\infty  (k+1)^{- \left( \frac{s}{d} - \frac{1}{2}  \right)} {\eta_{\beta,k}} r_k,}
\end{equation}
where $\{ {\eta_{\beta,k}} \}$ are i.i.d. random variables with distribution 
$BK_\beta(1, 1)$.
 {We highlight that the assumption $\frac{s}{d}  - \frac{1}{2} \ge 1$ is rather
   strong and is only sufficient to ensure a.s. convergence of the
   sums in \eqref{Besov-prior} and \eqref{Besov-innovation}. The SD property of $B^s_{11}(\mbb{T}^d)$  and the representation \eqref{Besov-innovation} remain valid for smaller values of $s$ so long as the sums converge a.s.
   so  $\mu$ is  well-defined.}

 \subsection{Sampling with Bessel-K priors on Hilbert spaces}
\label{subsec:sampling-with-bessel-K-Hilbert-space}
 We are now in position to generalize the lifted RCAR and SARSD algorithms
 of Subsection~\ref{sec:BK-1D-algorithms} to Hilbert spaces.
 The key is to use the 1D proposal kernels of
 Algorithms~\ref{BK-RCAR-1D-algorithm} and \ref{BK-SARSD-1D-algorithm}
 to construct a Markov chain for each coefficient $\eta_k$ 
 in \eqref{SG-prior-expansion} independently. The main advantage of this approach
 is that since the $\eta_k$ and their corresponding
 proposal kernels are independent of each other we
 can update them all at once.
 Of course, since there are countably infinitely many  $\eta_k$
 we cannot use the resulting algorithms in practice but 
 we 
 can easily approximate them by
truncating the  sum in \eqref{SG-prior-expansion}.
The following result allows us to use 1D proposal kernels to construct
a proposal kernel on the space $X$ for product measures of the form \eqref{product-prior}.

\begin{theorem}\label{reversible-product-kernel}
  Suppose
  \begin{equation}\label{product-kernel-display}
    \mu = \Law\left \{ u= \sum_{k=1}^\infty \gamma_k \xi_k r_k \right\} \in P(X),
    \end{equation}
 \myhl{ where  $\{r_k\}$ is an orthonormal basis in $X$, $\{ \xi_k\}$
   are independent  random variables with law $\mu_k \in P(\mbb R)$, and
   $\{ \gamma_k\}$ is a fixed
 sequence in $\mbb R$ that decays sufficiently fast
 so that $ \| u\|_X < +\infty $ a.s. and $\mu$ is well-defined.} Suppose $\mcl Q_k$
    are probability kernels that are  $\mu_k$-reversible 
    and let $\mcl Q(u, \dd v)$ be the transition kernel corresponding to the following operations:
    \begin{enumerate}
    \item For $k= 1, 2, 3,  \dots$ draw $\zeta_k \sim  \mcl Q_k(\xi_k, \dd \zeta_k)$
      where $\xi_k$ are the coefficients in \eqref{product-kernel-display}.
      \item Set $v = \sum_{j=1}^\infty \gamma_k \zeta_k r_k$. 
      \end{enumerate}
      Then $\mcl Q$ satisfies detailed balance with respect to $\mu$. 
    \end{theorem}

    \begin{proof}
      Let $\tau_k(\xi_k,\zeta_k) = \mcl Q_k(\xi_k, \dd \zeta_k) \dd \mu_k(\xi_k)$. Since $\mcl Q_k$
      satisfy detailed balance then
      \begin{equation*}
        \hat{\tau}_k(s,s') = \hat{\tau}_k(s',s), \qquad \forall s,s' \in \mbb R.
      \end{equation*}
      Now let $\tau(u,v) = \mcl Q(u, \dd v) \dd \mu(u)$ and take $\varrho , \varrho' \in X^\ast$ 
     and  denote their Riesz representers with $\rho, \rho' \in X$
     and let $\rho_k = \langle \rho, r_k \rangle$ and $\rho'_k = \langle \rho', r_k \rangle$
      be the basis coefficients of $\rho$ and $\rho'$. 
      Then
      \begin{equation*}
        \begin{aligned}
          \hat{\tau}( \varrho,  \varrho')
          & = \EE \exp(i \langle \varrho, u \rangle) \exp( i \langle \varrho', v \rangle) \\
          & = \EE \exp \left( i \sum_{k=1}^\infty \gamma_k \rho_k \xi_k   \right)
          \exp \left( i \sum_{k=1}^\infty \gamma_k \rho'_k \xi_k' \right)\\
          & = \EE \prod_{k=1}^\infty \exp( i \gamma_k \rho_k \xi_k + i\gamma_k \rho'_k \xi_k)  \\
           &= \prod_{k=1}^\infty \hat{\tau}_k( \gamma_k \rho_k, \gamma_k \rho'_k)
           = \prod_{k=1}^\infty \hat{\tau}_k( \gamma_k \rho'_k, \gamma_k \rho_k) 
           = \hat{\tau}( \varrho',  \varrho).
        \end{aligned}
      \end{equation*}
      Thus the characteristic function of $\tau$ is symmetric implying that
      $\mcl Q$ is reversible with respect to $\mu$.
    \end{proof}

\subsubsection{The lifted RCAR and SARSD algorithms on Hilbert spaces}
In light of Theorem~\ref{reversible-product-kernel} we now present the 
infinite-dimensional analogues of the lifted RCAR and SARSD algorithms of
Subsection~\ref{sec:BK-1D-algorithms} for product measures
of the form \eqref{product-prior} on Hilbert spaces.
We summarize the full algorithms below under Algorithms~\ref{BK-RCAR-algorithm} and \ref{BK-SARSD-algorithm}.

\begin{remark}
  Note that Remark~\ref{remark:RCAR-vs-SARSD} remains true when $X$ is a
  separable Hilbert space, that is, the RCAR algorithm has
significantly lower memory overhead in comparison to the SARSD algorithm
specially when $p$ is large.
\end{remark}

\begin{remark}
We highlight that since $BK(1, 1) = \Exp(1)$, then 
by taking $\gamma_k = ( k +1)^{-(s/d -1)}$
and $r_k$ to be sufficiently regular wavelet basis for $L^2(\mbb T^d)$
in accordance with Definition~\ref{besov-prior}, we  have that the $BK(1, \mcl R)$
coincides with
a $B^s_{11}(\mbb T^d)$ prior
\myhl{so long as $s$ is large enough that the
  spectral sums converge a.s.}, and so both Algorithms~\ref{BK-RCAR-algorithm}
and \ref{BK-SARSD-algorithm} can be used to sample posteriors that arise from
$B^s_{11}(\mbb T^d)$ priors.
\end{remark}

\begin{algorithm}
\caption{Lifted RCAR algorithm for $BK(p,\mcl R)$ priors.}
\label{BK-RCAR-algorithm}
Choose $\beta \in (0,1)$ and suppose $\mu = BK(p, \mcl R)$ with $p >0$ and
 a Hilbert-Schmidt operator $\mcl{R}$
with eigenpairs $\{\gamma_\ell, r_\ell\}$. In the following
$k =1,2$, 
 $\ell = 1, 2 ,3 ,\dots $ and all random variables are drawn independently.
{\begin{enumerate}[1.]
  \item Set $j = 0$, draw 
    $\myhl{u_{k,\ell}}^{(0)} \sim \Gamm(p, 1)$ and set
    \begin{equation*}
u^{(0)} = \sum_{\ell=1}^\infty
    \gamma_\ell \left(\myhl{u_{1,\ell}^{(0)} - u_{2,\ell}^{(0)}}\right)r_\ell.
  \end{equation*}

  \item At iteration $j$ draw 
    \begin{equation*}
      \zeta_{k,\ell} \sim \Beta(p \beta, p(1- \beta)),
       \quad \myhl{w_{k,\ell}} \sim \Gamm(p(1- \beta), 1).
    \end{equation*}
  \item Propose the sequence
  \begin{equation*}
    \begin{aligned}
      \myhl{v_{k,\ell}}^{(j+1)}
      &= \zeta_{k,\ell} \myhl{u_{k,\ell}}^{(j)}
      + \myhl{w_{k,\ell}}.
      \end{aligned}
    \end{equation*}
  \item Set $v^{(j+1)}
    = \sum_{\ell=1}^\infty \gamma_\ell \left(\myhl{v_{1,\ell}^{(j+1)} - v_{2,\ell}^{(j+1)}}\right)
    r_\ell.$
\item With probability
  \begin{equation*}
a( u^{(j)}, v^{(j+1)}) =
  \min\{ 1, \: \exp( \Psi(u^{(j)}) - \Psi(v^{(j+1)} ) \},
\end{equation*}
set
   $u^{(j+1)} =v^{(j+1)}, \myhl{u_{k,\ell}^{(j+1)} = v_{k,\ell}^{(j+1)}}$.
\item Otherwise set $u^{(j+1)} = u^{(j)}$ and $ \myhl{u_{k,\ell}^{(j+1)} = u_{k,\ell}^{(j)}}$.
\item Set $j \leftarrow j +1$ and return to step 2.
\end{enumerate}}
\end{algorithm}

\begin{algorithm}
\caption{Lifted SARSD algorithm for $BK(p,\mcl R)$ priors.}
\label{BK-SARSD-algorithm}
Choose $\beta \in (0,1)$ and suppose $\mu = BK(p, \mcl R)$ with $p \in \mbb N$ and
 a Hilbert-Schmidt operator $\mcl{R}$
with eigenpairs $\{\gamma_\ell, r_\ell\}$. In the following
$k =1,2, \dots, 2p$, 
 $\ell = 1, 2 ,3 ,\dots $ and all random variables are drawn independently.
{\begin{enumerate}[1.]
  \item Set $j = 0$, draw 
    $\myhl{u_{k,\ell}}^{(0)}, \sim \Gamm(p, 1)$ and set
    \begin{equation*}
u^{(0)} = \sum_{\ell=1}^\infty
\gamma_\ell \left( \sum_{k=1}^p\myhl{u_{k,\ell}^{(0)}}
  -  \sum_{k=p+1}^{2p}\myhl{u_{k,\ell}^{(0)}}\right)r_\ell.
  \end{equation*}

\item At iteration $j$ draw
  \begin{equation*}
   t \sim \Bern(1/2) \quad \myhl{w_{k,\ell}} \sim \Exp(1), \quad 
   \myhl{\zeta_{k,\ell} \sim \Bern(1 - \beta)}.
\end{equation*}

\item If $t =1$ propose forward 
    \begin{equation*}
      \myhl{v_{k,\ell}^{(j+1)}} = \beta \myhl{u_{k,\ell}^{(j+1)}} +\myhl{ \zeta_{k,\ell} w_{k,\ell}}.
    \end{equation*}
  \item If $t=0$ propose backward
    \begin{equation*}
      \myhl{v_{k,\ell}^{(j+1)}} = \min \{\myhl{u_{k,\ell}^{(j)}}/\beta, \myhl{w_{k,\ell}}/(1- \beta)  \}.
    \end{equation*}
  \item Set $v^{(j+1)} =
    \sum_{\ell=1}^\infty \gamma_\ell \left(\sum_{k=1}^p
      \myhl{v_{k,\ell}^{(j+1)}} -
      \sum_{k=p+1}^{2p} \myhl{v_{k,\ell}^{(j+1)}}\right)
    r_\ell.$
\item With probability
  \begin{equation*}
a( u^{(j)}, v^{(j+1)}) =
  \min\{ 1, \: \exp( \Psi(u^{(j)}) - \Psi(v^{(j+1)} ) \},
\end{equation*}
set
   $u^{(j+1)} =v^{(j+1)}, \myhl{u_{k,\ell}^{(j+1)} = u_{k,\ell}^{(j+1)}}$.
 \item Otherwise set $u^{(j+1)} = u^{(j)}$ and
   $\myhl{u_{k,\ell}^{(j+1)} = u_{k,\ell}^{(j)}}$.
\item Set $j \leftarrow j +1$ and return to step 2.
\end{enumerate}}
\end{algorithm}

\section{Well-posed Bayesian inverse problems with Bessel-K priors}\label{sec:well-posedness}

We briefly discuss well-posedness and consistent approximations of Bayesian inverse problems with Bessel-K priors. The main references for the results used in this section are  
\cite{stuart-bayesian-lecture-notes, hosseini-sparse, sullivan} where the 
 theory of well-posed Bayesian inverse problems with non-Gaussian priors 
was developed. 

As in Section~\ref{sec:introduction} we consider the additive noise model 
$$
y = \mcl{G}(u) + \epsilon, \qquad \epsilon \sim \mcl{N}(0, \pmb{\Sigma}).
$$ 
The parameter  $u$ belongs to the Hilbert space  $X$ and the data $y \in \reals^M$ 
for $M\in \integers$.  We recall the following definition of well-posedness for Bayesian inverse problems.

\begin{definition}[Hellinger Well-posedness {\cite{hosseini-convex}}] \label{def-well-posedness}
Define the Hellinger metric
$$
 d_H( \mu_1, \mu_2) := \left( \frac{1}{2} \int_X  \left( \sqrt{ {\dd \mu_1}/{ \dd \Theta} } - \sqrt{ {\dd \mu_2}/{ \dd \Theta} }  \right)^2 \dd \Theta \right)^{1/2}
$$
on $P(X)$ where $\Theta \in P(X)$ is such that $\mu_1 \ll \Theta$ and $\mu_2 \ll \Theta$.
For a choice of the
prior measure $\mu_0$ and the likelihood potential $\Phi$, the Bayesian inverse problem  \eqref{bayes-rule}
 is well-posed if:
  \begin{enumerate}
  \item There exists a unique posterior probability measure $\mu^{y} \ll \mu_0$.
  \item  For any $\epsilon>0$, there is a $ \delta>0$ such that if $\|y-y'\|_2 <\delta$ then  $d_H( \mu^{y},\mu^{y'} ) <\epsilon$.
  \end{enumerate}
\end{definition}

With a definition of well-posedness at hand, we can now identify conditions on the 
prior measure $\mu_0$ and the forward map $\mcl{G}$ that result in a well-posed problem.
 The following theorem 
is a direct consequence of \cite[Cor.~3.5 and Thm.~3.8]{hosseini-sparse}.

\begin{theorem}\label{well-posedness}
 Suppose  $\Phi$ has the form \eqref{quadratic-likelihood} and the forward map
 $\mcl{G}:X \mapsto \reals^M$ is locally Lipschitz continuous, and 
polynomially
bounded, i.e., 
\begin{equation}\label{polynomially-bounded-forward-map}
\exists C>0 \text{ and } q \in\integers \qquad \text{such that} \qquad \| \mcl{G}(u) \|_2 \le C \max\{1, \|u\|_X^q\} 
\qquad \forall u \in X. 
\end{equation}
If $\mu_0 \in P(X)$ and { $\| \cdot \|_X\in L^{q}(X, \mu_0)$} then the
Bayesian inverse problem \eqref{bayes-rule} is well-posed.
\end{theorem}

We emphasize that the requirements of this result on the prior measure $\mu_0$ are fairly relaxed. For example, 
if $\mcl{G}$ is bounded and linear then we need $\mu_0$ to have bounded first moments in order to achieve 
well-posedness. Putting Theorem~\ref{well-posedness} together with Theorem~\ref{SG-prior-properties} we obtain the 
following corollary.

\begin{corollary}
Suppose  $\Phi$ has the form \eqref{quadratic-likelihood} and
 the forward map
 $\mcl{G}:X \mapsto \reals^M$ satisfies the conditions of Theorem~\ref{well-posedness}.
If $\mu_0$ is a $BK(p, \mcl{R})$ prior with $p > 0$ and a Hilbert-Schmidt operator $\mcl{R}:X \mapsto X$
 then the
Bayesian inverse problem \eqref{bayes-rule} is well-posed.
\end{corollary}

Let us now discuss approximations of the
posterior $\mu^y$. Since $X$ can be infinite dimensional  we 
cannot solve  \eqref{bayes-rule} directly. Instead, we 
approximate the posterior $\mu^{y}$ with a sequence  
$\{ \mu_N^{y}\} \in P(X)$
 that can be described in a feasible manner (we will make this notion precise shortly). We say that the sequence
of measures $\{ \mu^{y}_N\}$ are a {\it consistent approximation} to $\mu^{y}$ if $d_{H}(\mu_N^{y}, \mu^{y}) \to 0$ as $N \to \infty$.

Now suppose that $\mcl{G}_N$ is an approximation to $\mcl{G}$  parameterized by $N$. We
think of $N \in \integers$ as a discretization parameter such as the number of 
terms in a truncated spectral expansion or the number of elements in a finite element method. Now define the sequence of measures
\begin{equation}\label{discrete-bayes-rule}
\frac{\dd \mu_N^{y}}{\dd \mu_0}(u) = \frac{1}{Z_N(y)} \exp\left( - \Phi_N(u;y)
\right), \qquad \qquad Z_N(y) := \int_X \exp( -\Phi_N(u;y)) \dd \mu_0(u),
\end{equation}
where 
$$
\Phi_N(u;y) := \frac{1}{2} \left\| \mcl{G}_N(u)  - y \right\|_{\pmb{\Sigma}}^2.
$$
We have the following consistency theorem as a direct consequence of Theorem~\ref{well-posedness}  and  
 \cite[Thm.~4.3]{hosseini-sparse}.

\begin{theorem}\label{consistent-discretization}
  Suppose that the forward map 
$\mcl{G}: X \mapsto \reals^M$ and its approximations $\mcl{G}_N:X \mapsto \reals^M$ are 
locally Lipschitz continuous and satisfy condition \eqref{polynomially-bounded-forward-map}
with uniform constants \myhl{$q \in \mbb{N}$} and $C> 0$ for all 
$N$. Furthermore, assume that there exists a bounded function $\psi(N)$ so that $\psi(N) \to 0$ as $N \to \infty$ and
$$
\exists C' >0 \qquad \text{so that} \qquad  \| \mcl{G}_N(u) - \mcl{G}(u)\|_2 \le C' \| u\|_X^q \psi(N) \qquad \forall u \in X.
$$
If the prior $\mu_0 \in P(X)$ and ${\| \cdot \|_X \in L^{2q}(X, \mu_0)}$ then 
the measures $\mu^{y}$ and $\mu^{y}_N$, defined via 
\eqref{bayes-rule}
and \eqref{discrete-bayes-rule} respectively, are well-defined and 
there exists a constant $D >0$ independent of $N$ so that 
$d_H( \mu_N^{y} , \mu^{y} ) \le D \psi(N).$
\end{theorem}

Thus the error in approximation of $\mcl{G}$ translates directly into
the Hellinger distance between the true posterior $\mu^{y}$ and 
the approximation $\mu_N^{y}$.
A particularly useful method for approximation of $\mcl{G}$ is discretization by
 Galerkin projections whenever $\mu_0$ has a product structure such as the case of 
the Bessel-K priors.


Suppose that $X$ has an  
orthonormal basis $\{ r_k\}$ and let $X_N := \text{span} \{ r_k\}_{k=1}^N $.
 Let $\Pi_N : X \mapsto X$ be the projection operator
onto $X_N$ and define
$\mcl{G}_N(u) := \mcl{G}( \Pi_Nu)$. This type of approximation 
is used in  our numerical experiments in Subsection~\ref{sec:deconv}.
Now following Theorem~\ref{consistent-discretization} and standard arguments using the fact that 
the forward map $\mcl G$ is locally Lipschitz and polynomially bounded we obtain the following corollary.

\begin{corollary}
Suppose  the likelihood potential $\Phi$ has the form \eqref{quadratic-likelihood},
 the forward map
 $\mcl{G}:X \mapsto \reals^M$ is locally Lipschitz continuous and  
\begin{equation*}\label{polynomially-bounded-forward-map-1}
\exists C>0 \text{ and } q \in\integers \qquad \text{such that} \qquad \| \mcl{G}(u) \|_{\pmb{\Sigma}} \le C \max\{1, \|u\|_X^q\} 
\qquad \forall u \in X. 
\end{equation*}
Suppose $\mu_0$ is a $BK(p, \mcl{R})$ prior with $p >0$ and a Hilbert-Schmidt operator $\mcl{R}:X \mapsto X$ with eigenpairs $\{ \gamma_k, r_k\}$. Let $\Pi_N: X \mapsto X$ be the projection operator onto the span of $\{r_k\}_{k=1}^N$ for $N \in \integers$ and suppose that $\mcl{G}_N:= \mcl{G} \circ \Pi_N $. 
 Then there exists a constant $D >0$ so that 
$
d_{H}(\mu^y_N, \mu^y) \le D \| {I} - \Pi_N \|
$
where ${I}$ denotes the identity operator on $X$ and the
  difference $\| I - \Pi_N\|$   is measured in the operator norm on $X$.
\end{corollary}

\section{Numerical experiments}\label{sec:numerical-experiments}

In this section we present numerical experiments that demonstrate the effectiveness of the
RCAR and SARSD 
algorithms of Sections~\ref{sec:MH-algorithms} and \ref{sec:bessel-k} in the
context of inverse problems with Bessel-K priors. 
We begin with simple finite-dimensional examples that demonstrate the ability of 
our algorithms to sample 
the right distribution. We then consider examples  in higher dimensions and
study the performance of our algorithms as a function of dimensions.
Throughout this section we primarily use the lifted RCAR
and SARSD algorithms outlined in Algorithms~\ref{BK-RCAR-algorithm} and \ref{BK-SARSD-algorithm}
and some of their derivatives.
\subsection{Example 5: Density estimation for a linear inverse problem in 2D}
\label{subsec:density-estimation}
We start with the linear inverse problem of Example \ref{example-1} with $N=2$.
We consider this simple setting since  the analytic posterior can be visualized easily.
Take $X = \reals^2$ and
model the data by
$$
y = \mb{G} u + \epsilon \qquad \epsilon \sim \mcl{N}(0, \sigma^2\mb{I}_2)
$$ 
where
$$
\mb{G} = 
\begin{bmatrix}
  1 &  1/2 \\
0 & 1 
\end{bmatrix} \qquad \text{and} \qquad \sigma = 1/2. 
$$
Let  $u_0 = (3/2, 1/2)^T$ and take {$y_0 = \mb{G}u_0$}, i.e., we assume the exact data is measured. Under these assumptions  
$$
{\Phi(u;y_0) = \frac{1}{2 \sigma^2} \| \mb{G}u - y_0 \|_2^2. }
$$
The prior measure is taken to be the product of two $BK(p, 1)$ priors on $\reals$
$$
\frac{\dd \mu_0}{\dd \Lambda}(t) =\left( 
{\frac{{1}}{ \sqrt{\pi} \Gamma(p) 2^{p - 1/2}}} \right)^2 \prod_{j=1}^2  
 \left|{t_j}\right| ^{p - 1/2}  K_{p - 1/2} \left( \left| t_j \right| \right) \qquad t = (t_1, t_2)^T \in \reals^2.
$$
We consider the values of $p =1, 2/3, 1/3$ and sample the posterior measure using lifted RCAR. 
The results of our computations are summarized in Figure~\ref{fig:2D-sampling-example}.
In all cases we took $\beta = 0.3$ which resulted in an acceptance ratio of 
approximately $0.2$ across the different values of $p$ (precise values were $0.1746$, $ 0.1970$
and $0.2234$ for $p = 1, 2/3$ and $1/3$ respectively). We used $8\times 10^5$ samples  to generate the 
2D histograms with a burnin of $10^4$ samples. These numbers are much larger than what is needed to
get an stable estimate of the  mean and variance but we need them to capture 
the details of the posterior densities in the histograms.
Visual comparison of the analytic and numerical posteriors 
serves as evidence that RCAR has  sampled the correct distribution.

We also compared the RCAR and SARSD algorithms for $p=1$. We show the empirical and analytic
posteriors in Figure~\ref{fig:2D-sampling-RCAR-v-SARSD}.
We ran SARSD  with the same parameters values as above. The average acceptance
rate for SARSD was $0.1574$ which is slightly less than $0.1746$ for RCAR
with the same choice  of $\beta = 0.3$. We also show an instance of the
traceplots of both algorithms in Figure~\ref{fig:trace-plot-and-acceptance-2D-example}
demonstrating good mixing of the chains.

\begin{figure}
  \centering
\begin{tabular}{c c c}
 &\includegraphics[width = .32\textwidth]{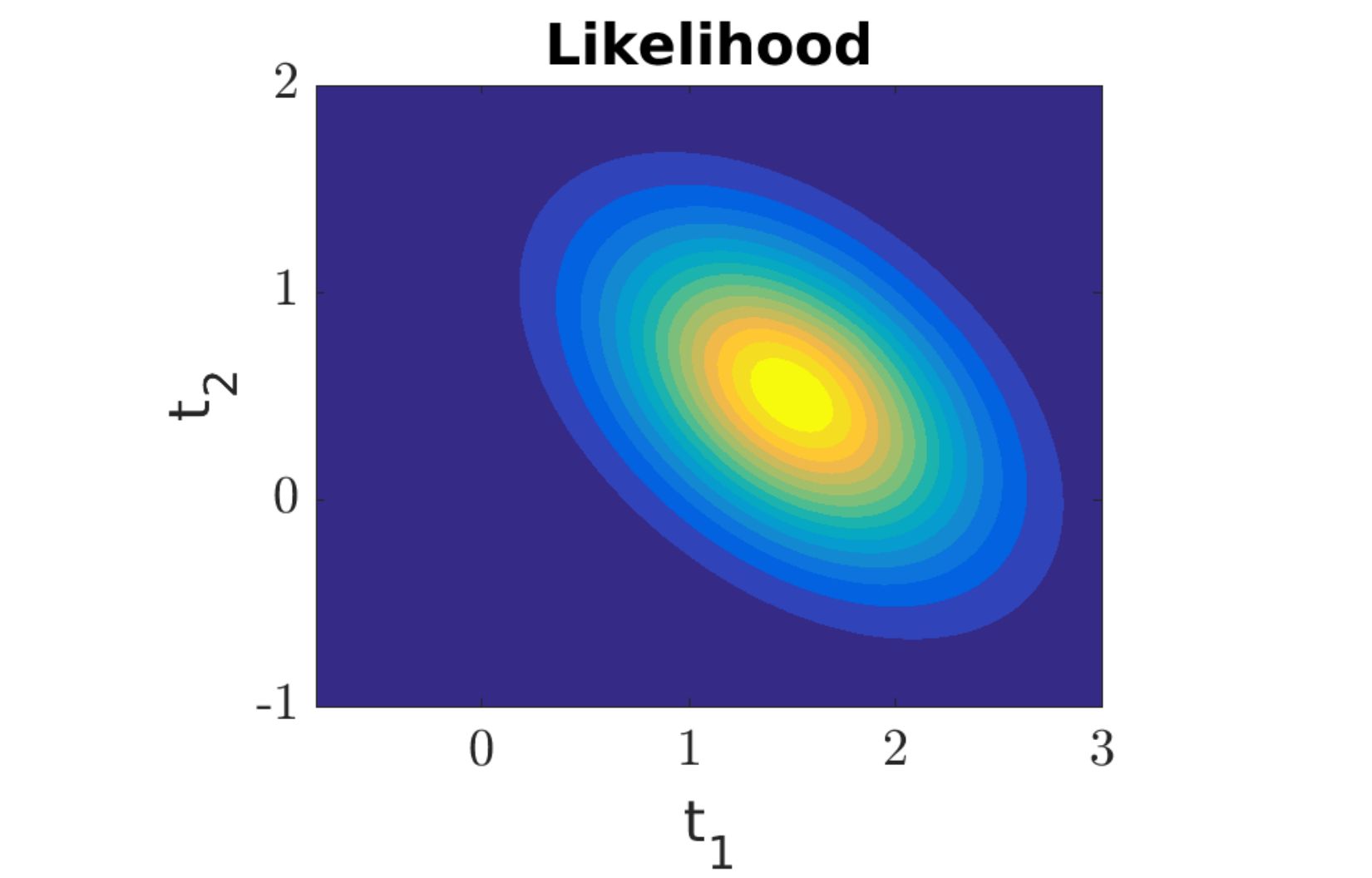}& \\
  \includegraphics[width = .32\textwidth]{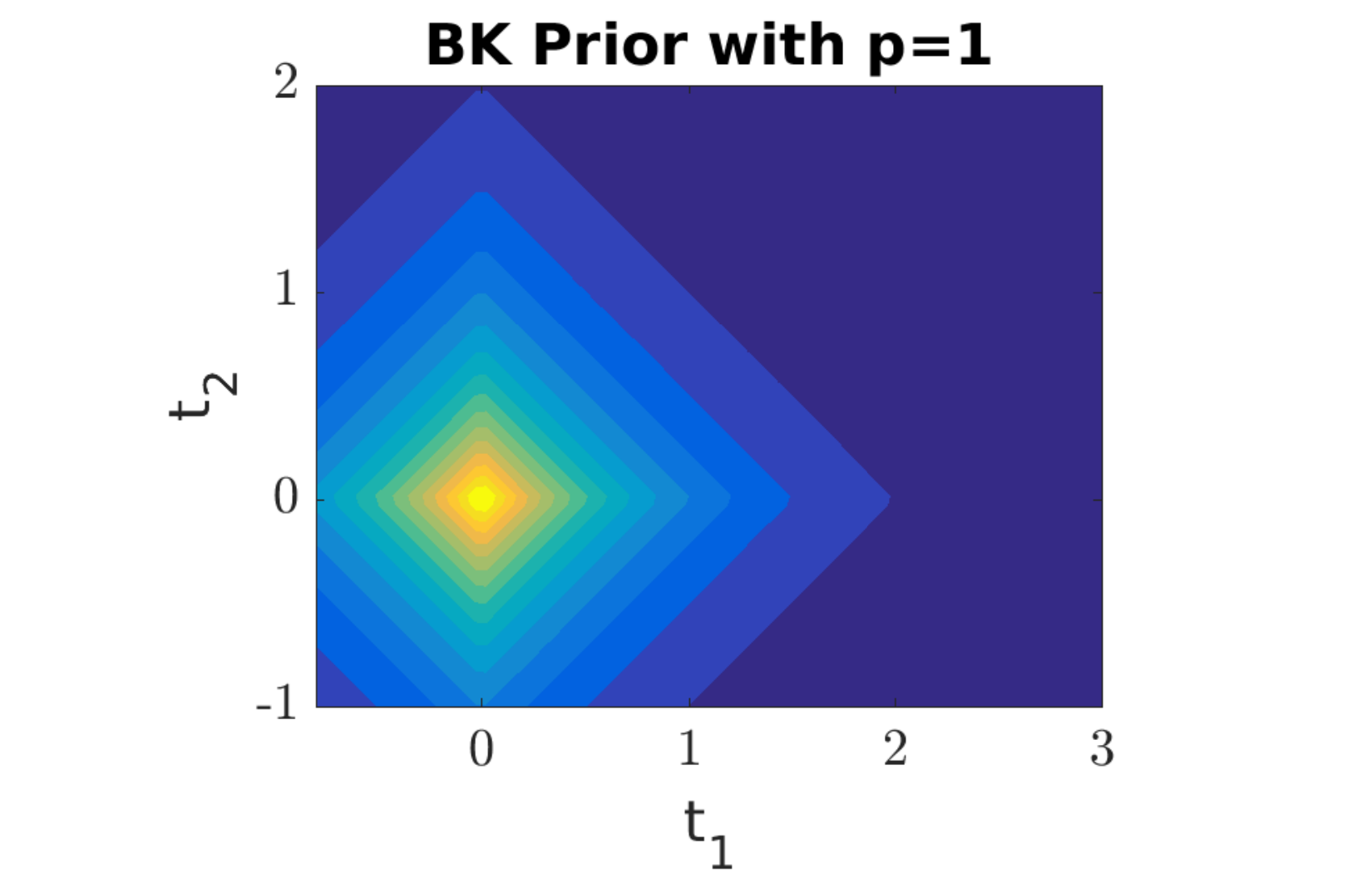} 
  &\includegraphics[width = .32\textwidth]{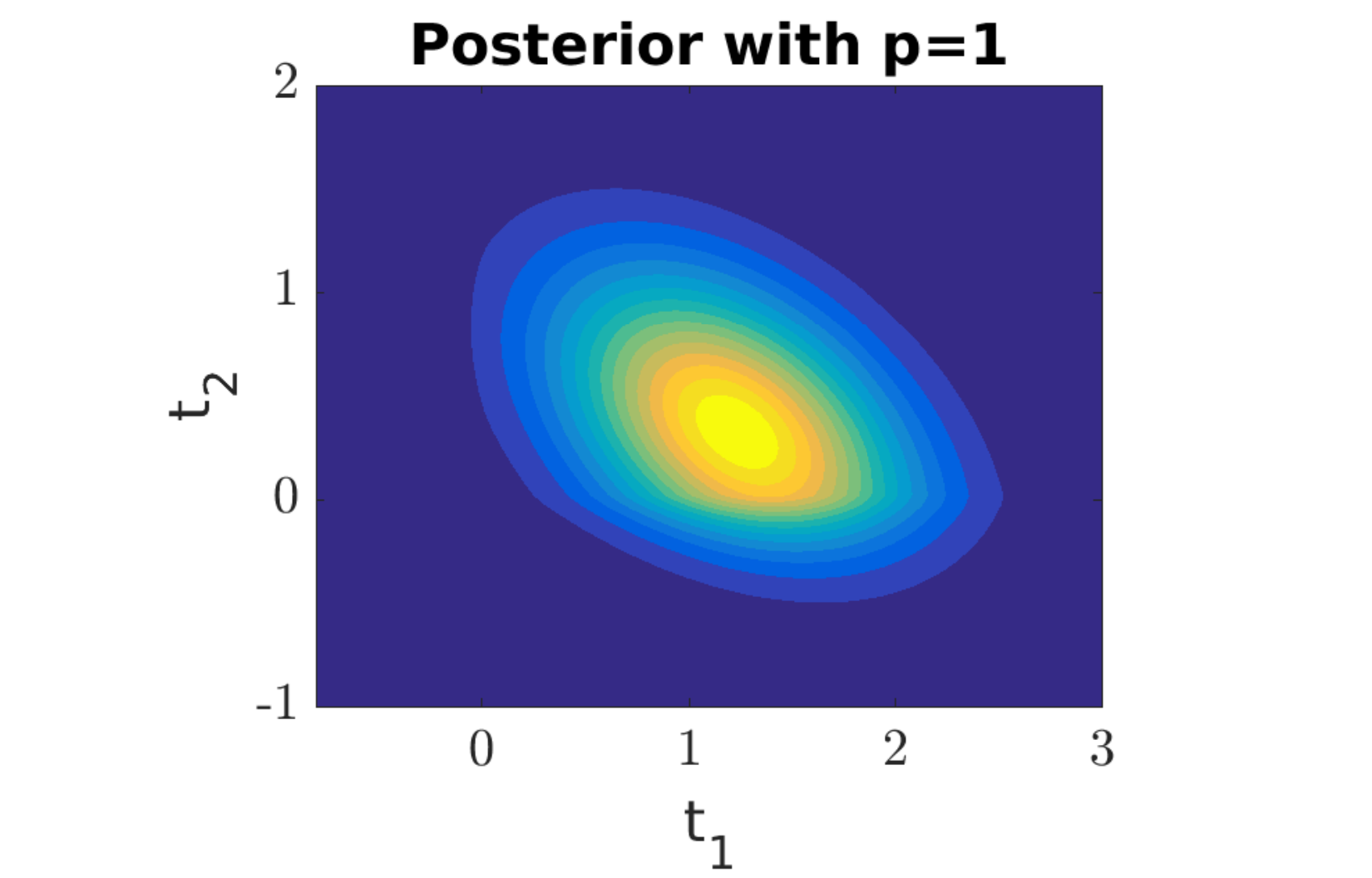} 
  &\includegraphics[width = .32\textwidth]{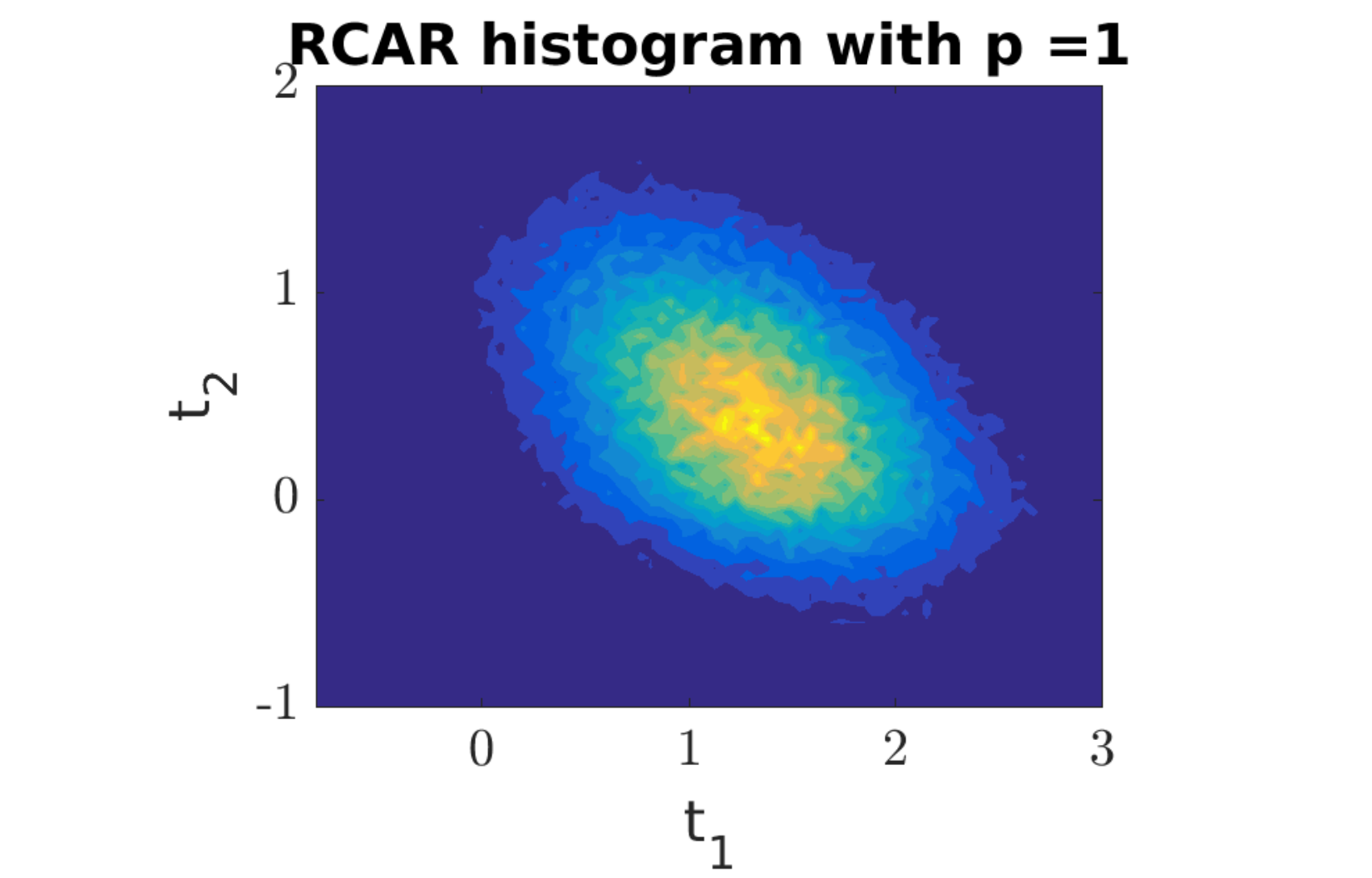} \\ 
  \includegraphics[width = .32\textwidth]{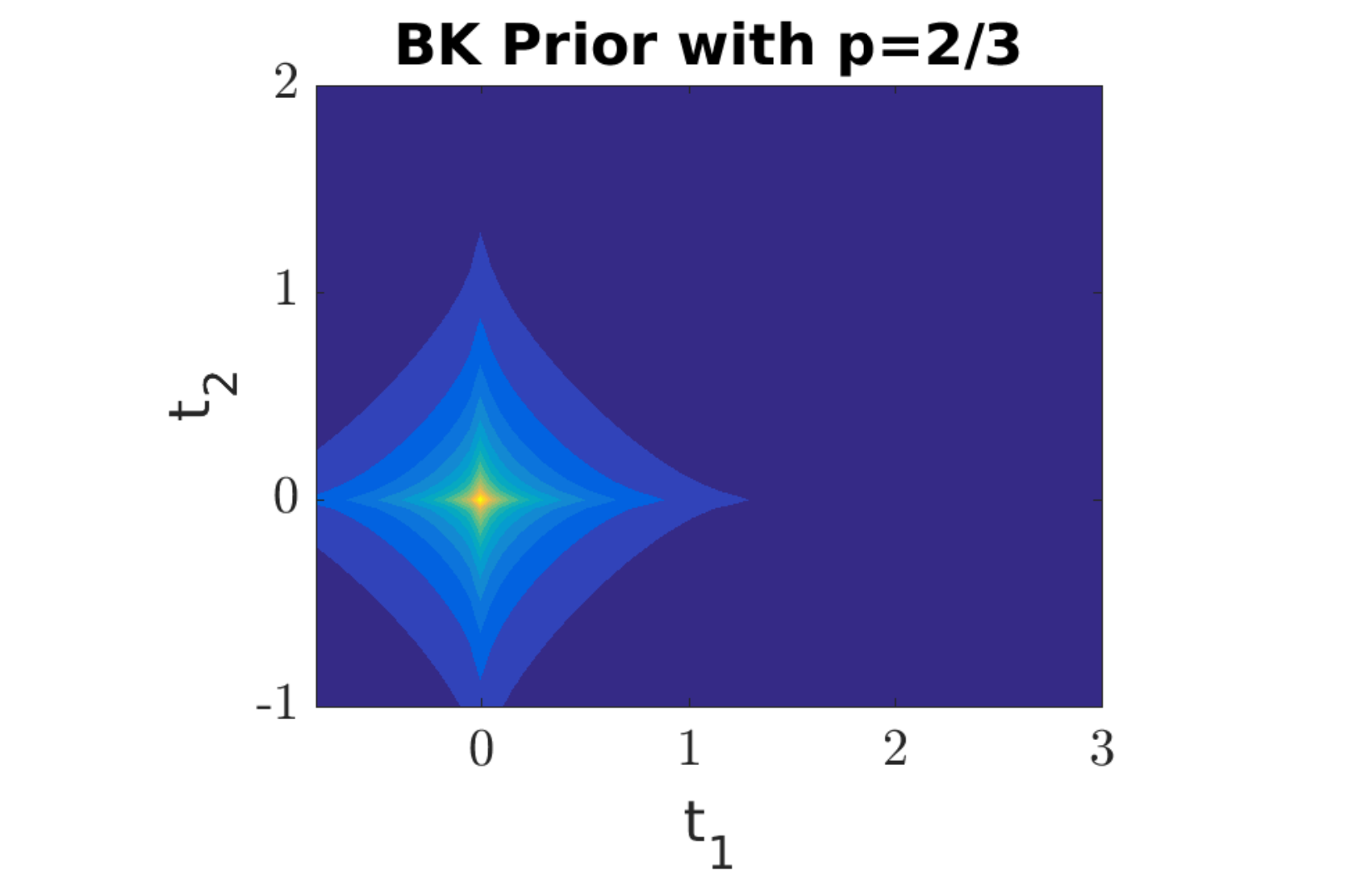} 
  &\includegraphics[width = .32\textwidth]{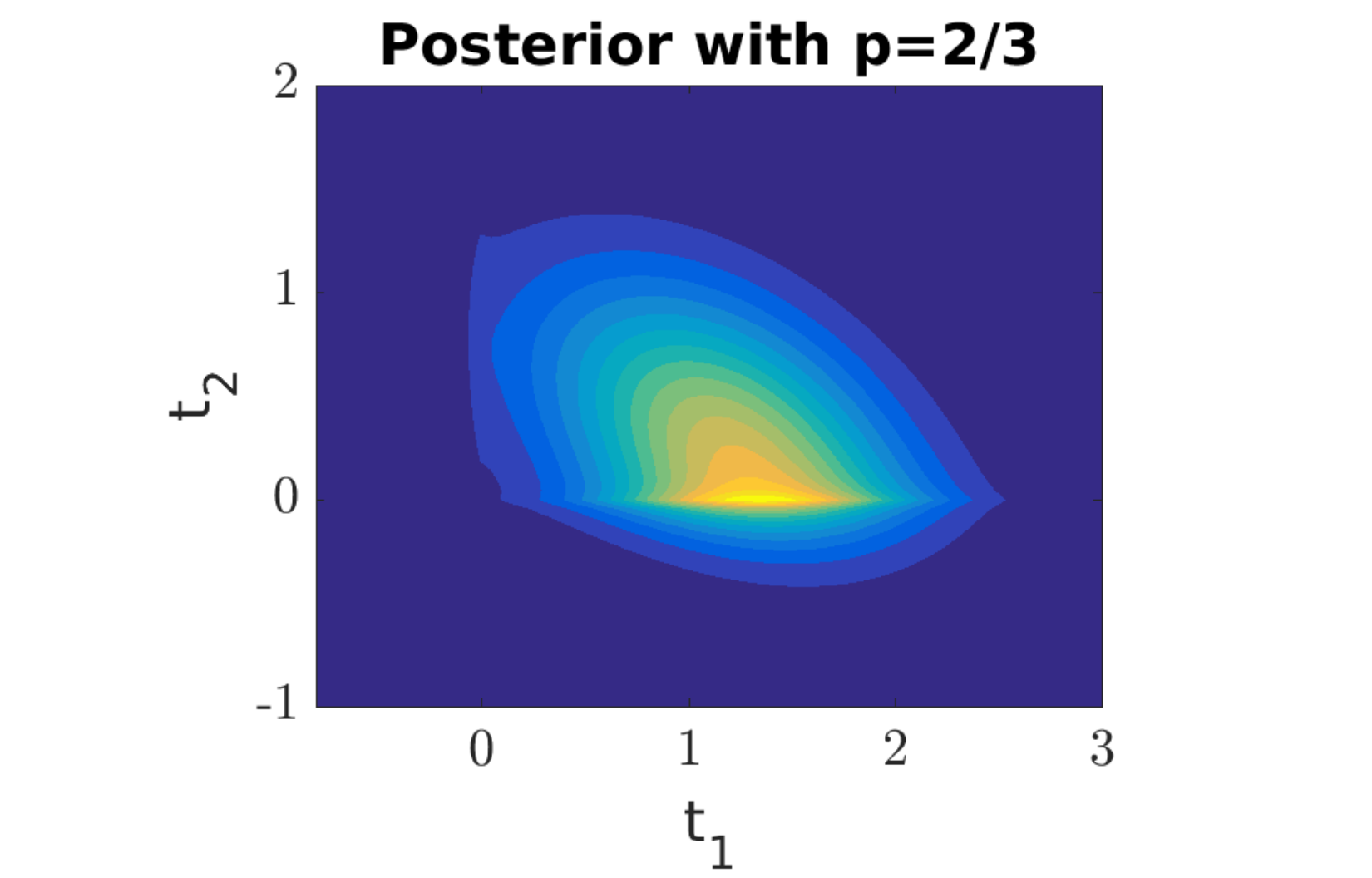} 
  &\includegraphics[width = .32\textwidth]{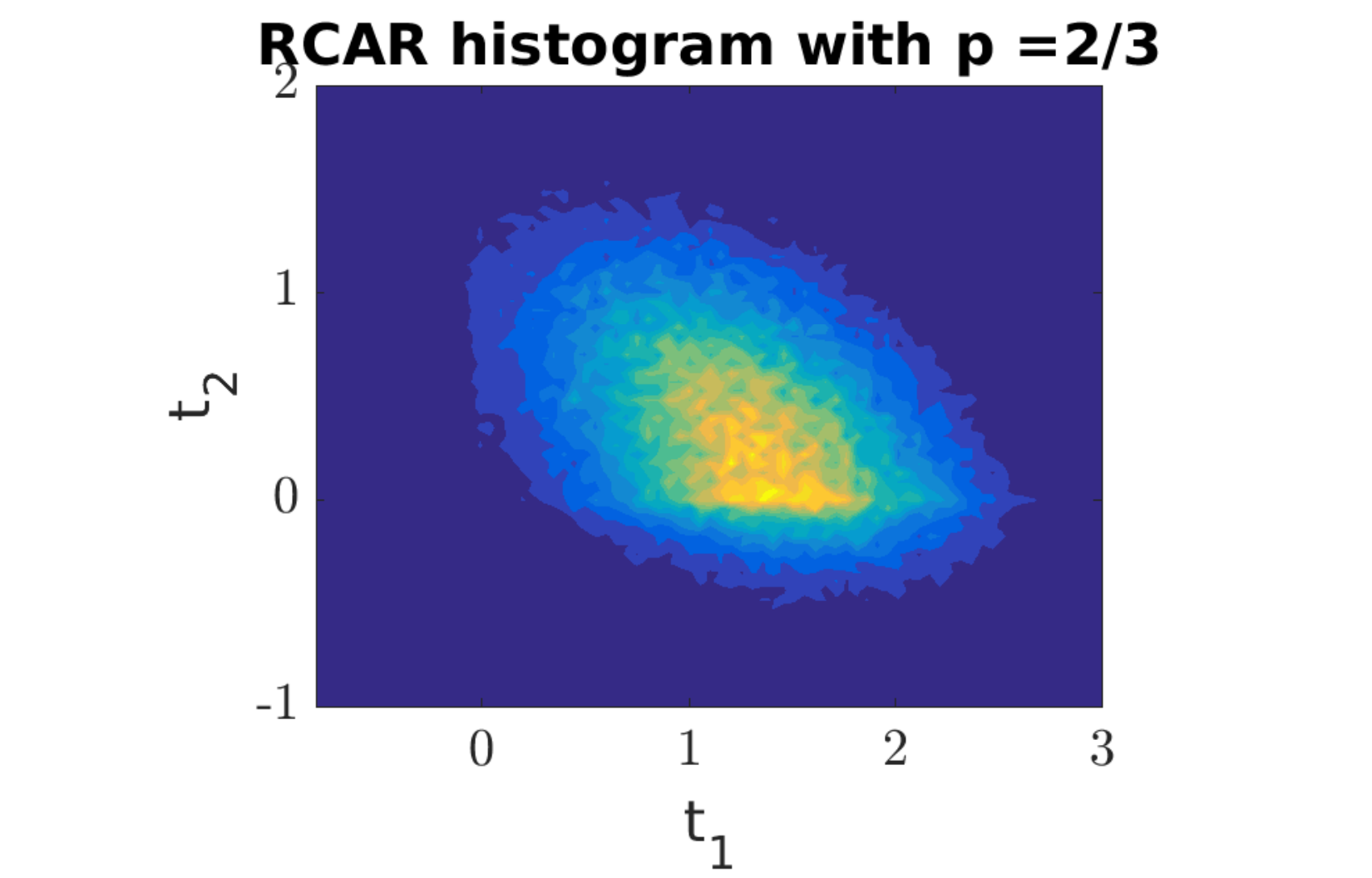} \\
 \includegraphics[width = .32\textwidth]{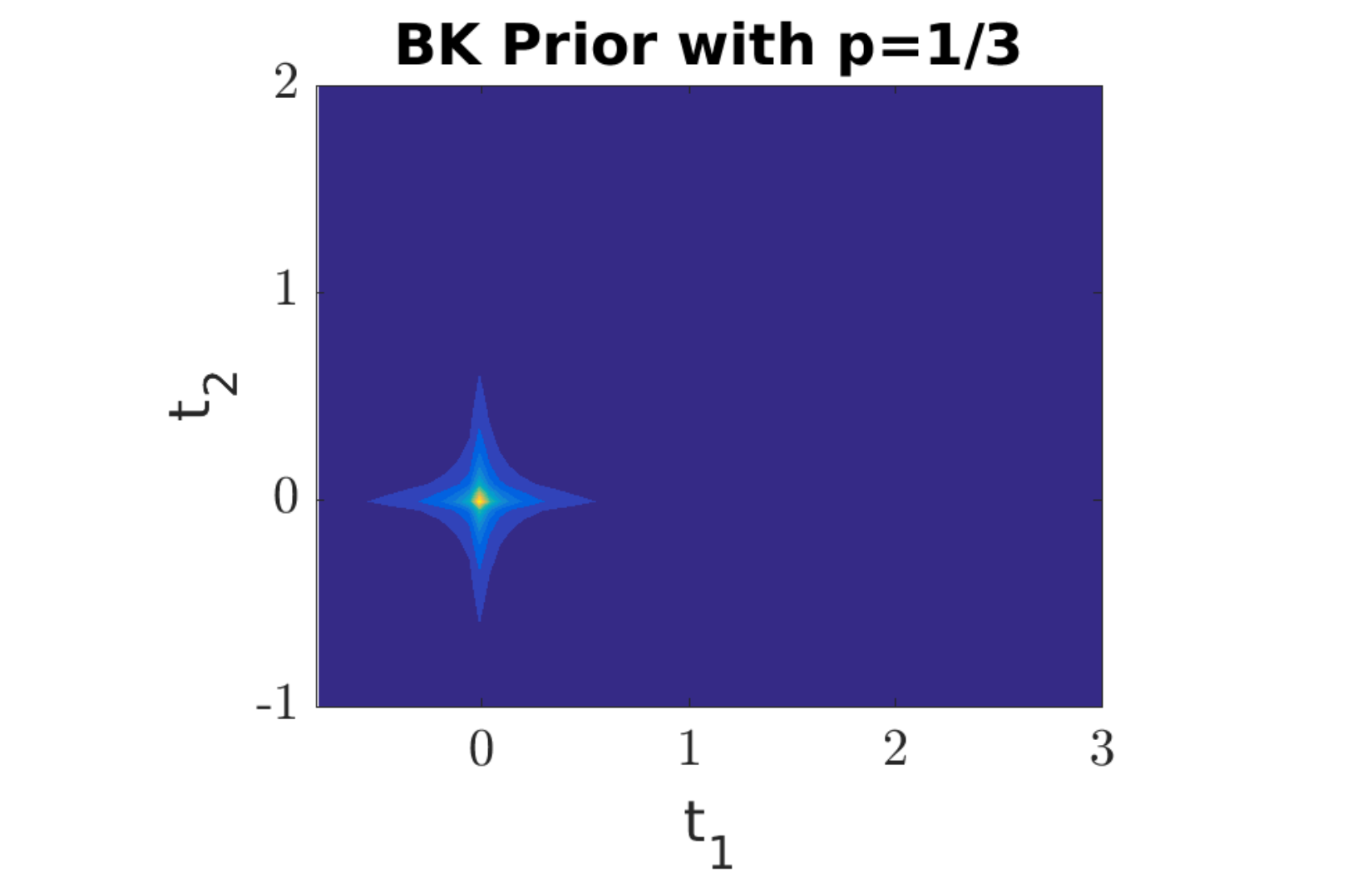} 
  &\includegraphics[width = .32\textwidth]{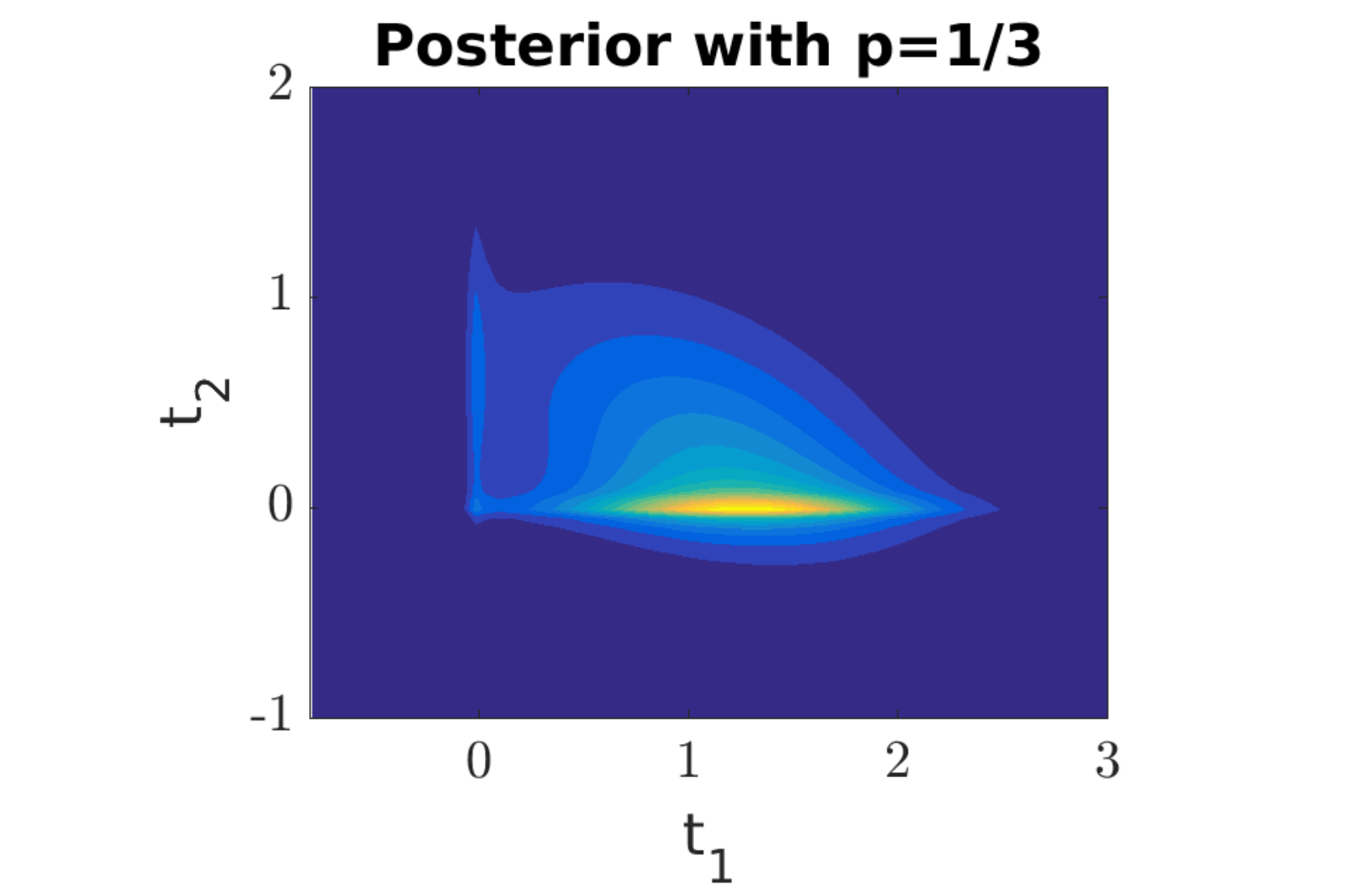} 
  &\includegraphics[width = .32\textwidth]{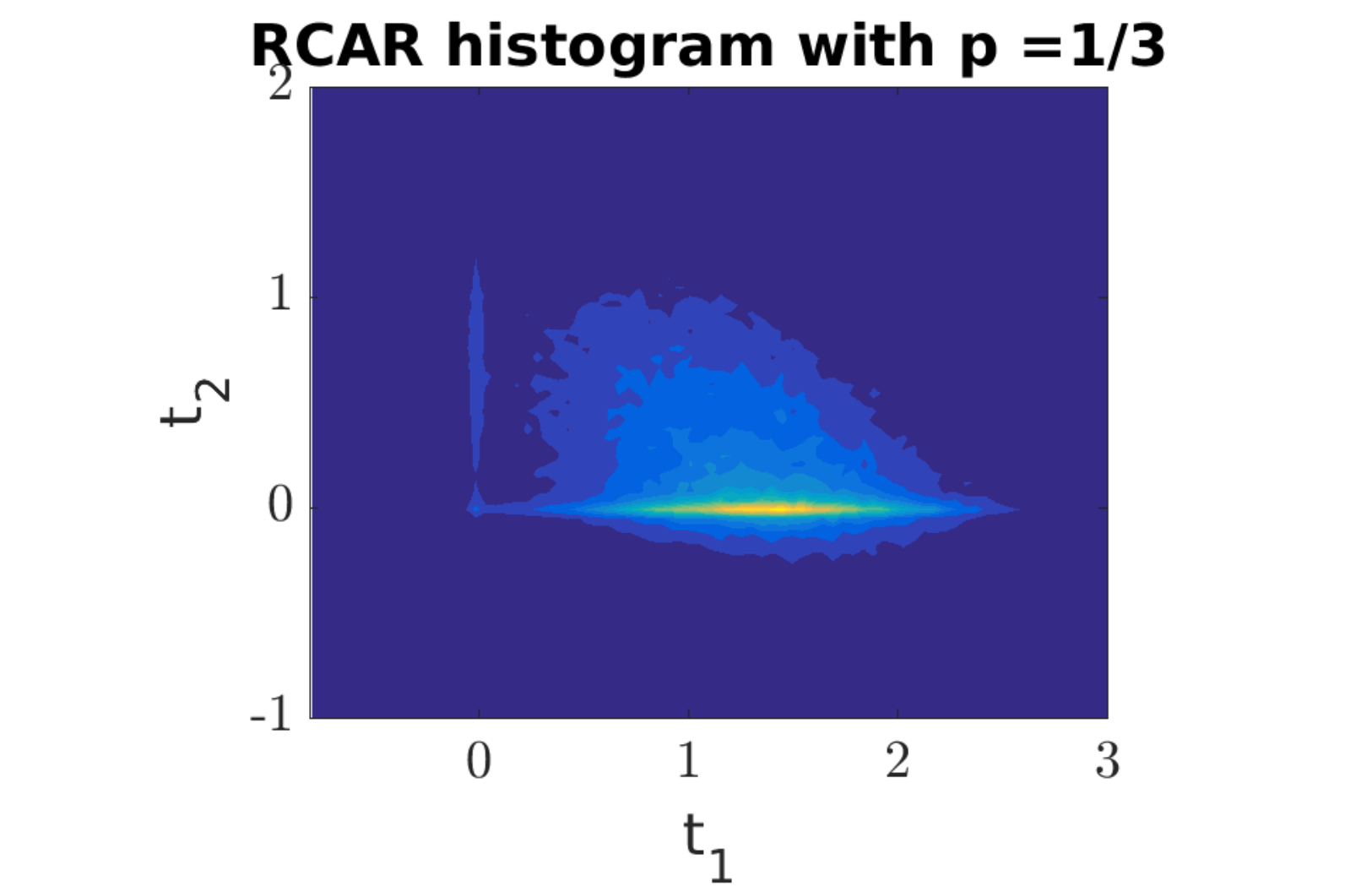} \\
\end{tabular}
\caption{Depiction of analytic prior, likelihood and
  posterior of Example 5 along with the 2D histogram of the RCAR samples for
  different choices of the shape parameter $p= 1, 2/3$ and $1/3$.
  The likelihood is shown  at the very top. The left column shows the prior measures. The analytic posterior densities are shown in the middle column and the empirical posteriors are shown in the right column.}
\label{fig:2D-sampling-example}
\end{figure}

\begin{figure}
  \centering
\begin{tabular}{c c c} 
  \includegraphics[width = .32\textwidth]{./figs/2D-example-analytic-posterior-p1} 
  &\includegraphics[width = .32\textwidth]{./figs/RCAR-2D-example-hist-p1}
    & \includegraphics[width = .32\textwidth]{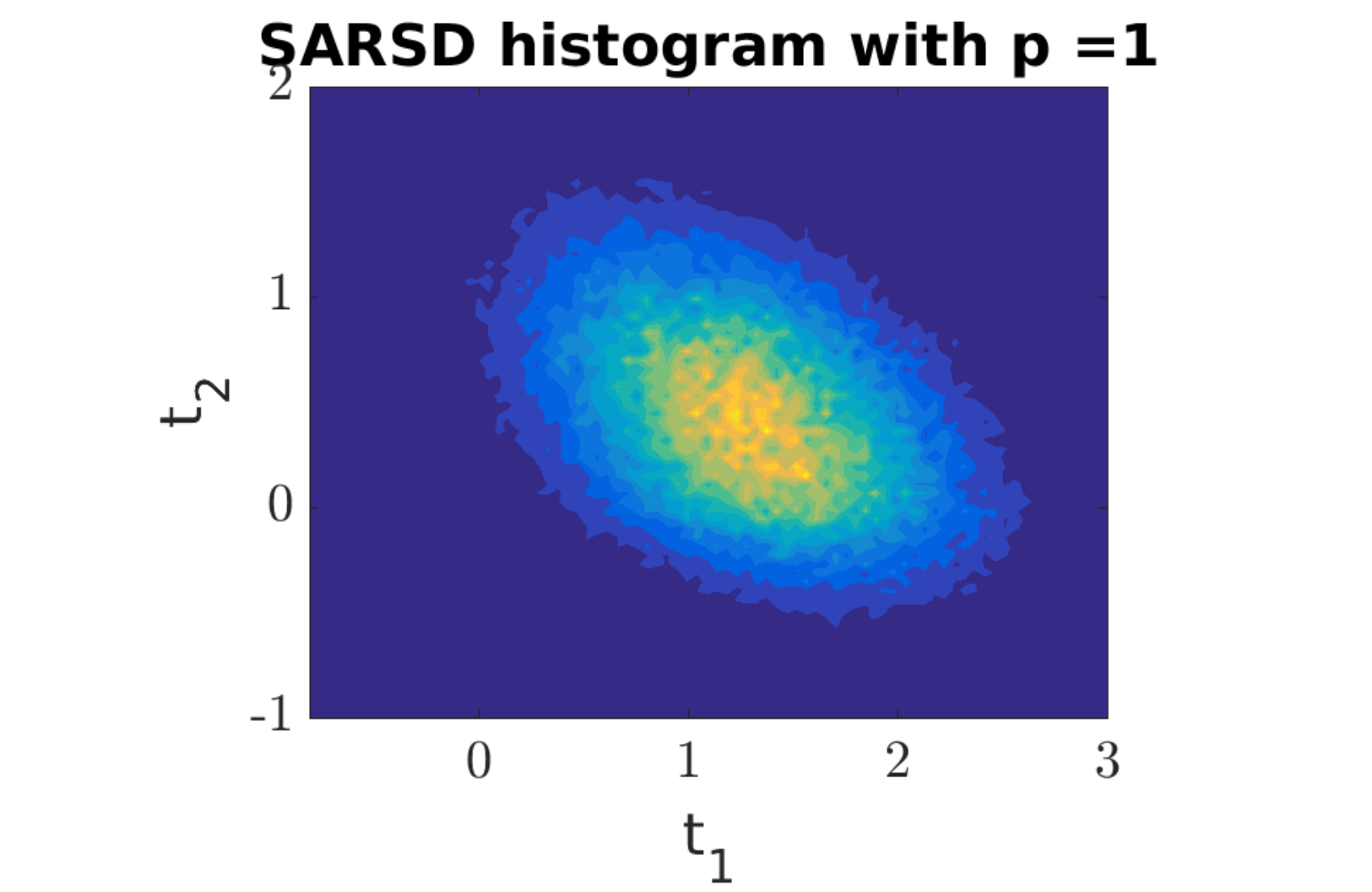}
\end{tabular}
\caption{Comparison between the analytic posterior and the MCMC histogram of RCAR and
  SARSD algorithms for posterior measures of Example 5 with shape parameter  $p=1$.}
\label{fig:2D-sampling-RCAR-v-SARSD}
\end{figure}

\begin{figure}[htp]
  \centering
  \includegraphics[width = .48\textwidth]{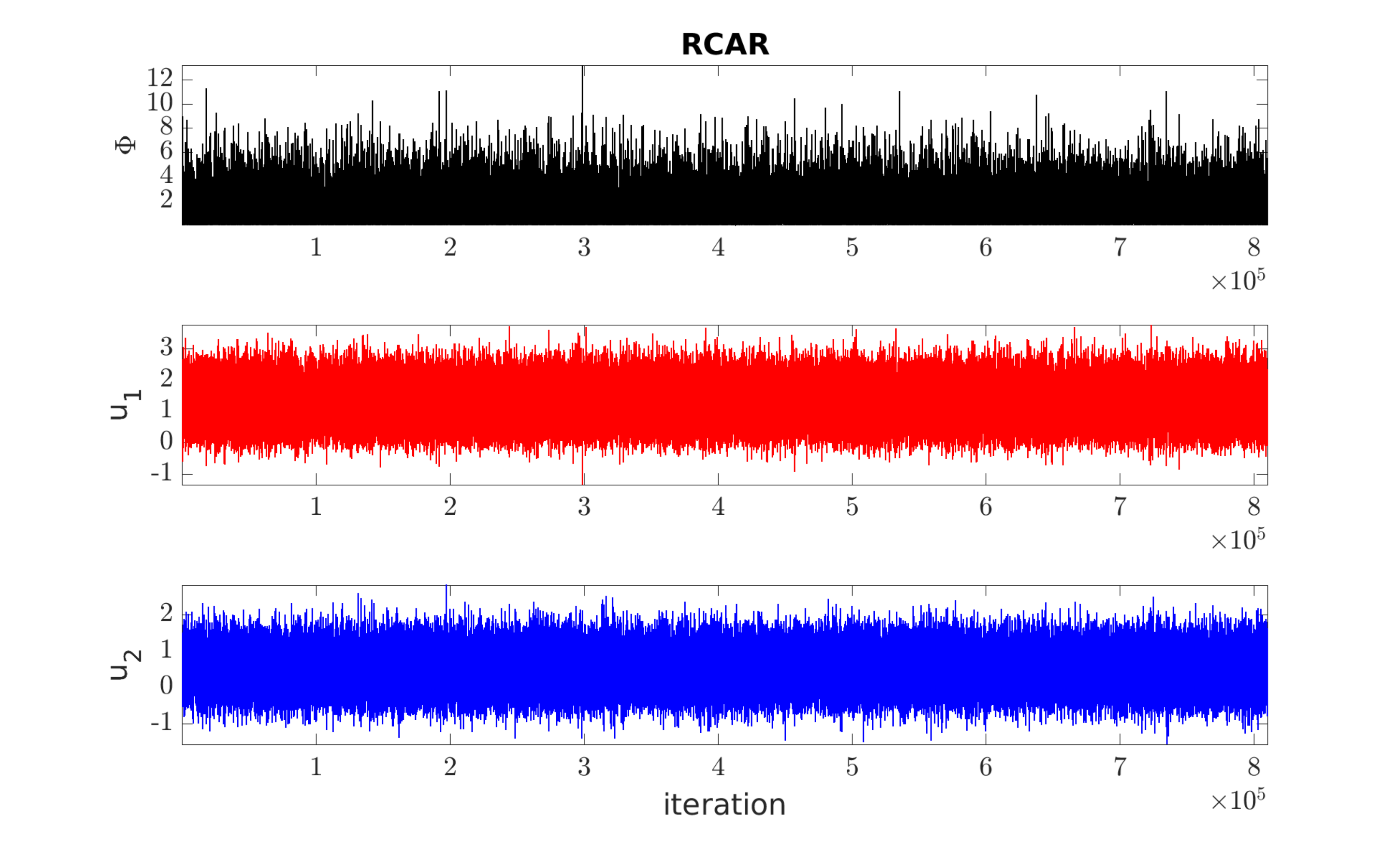}
  \includegraphics[width = .48\textwidth]{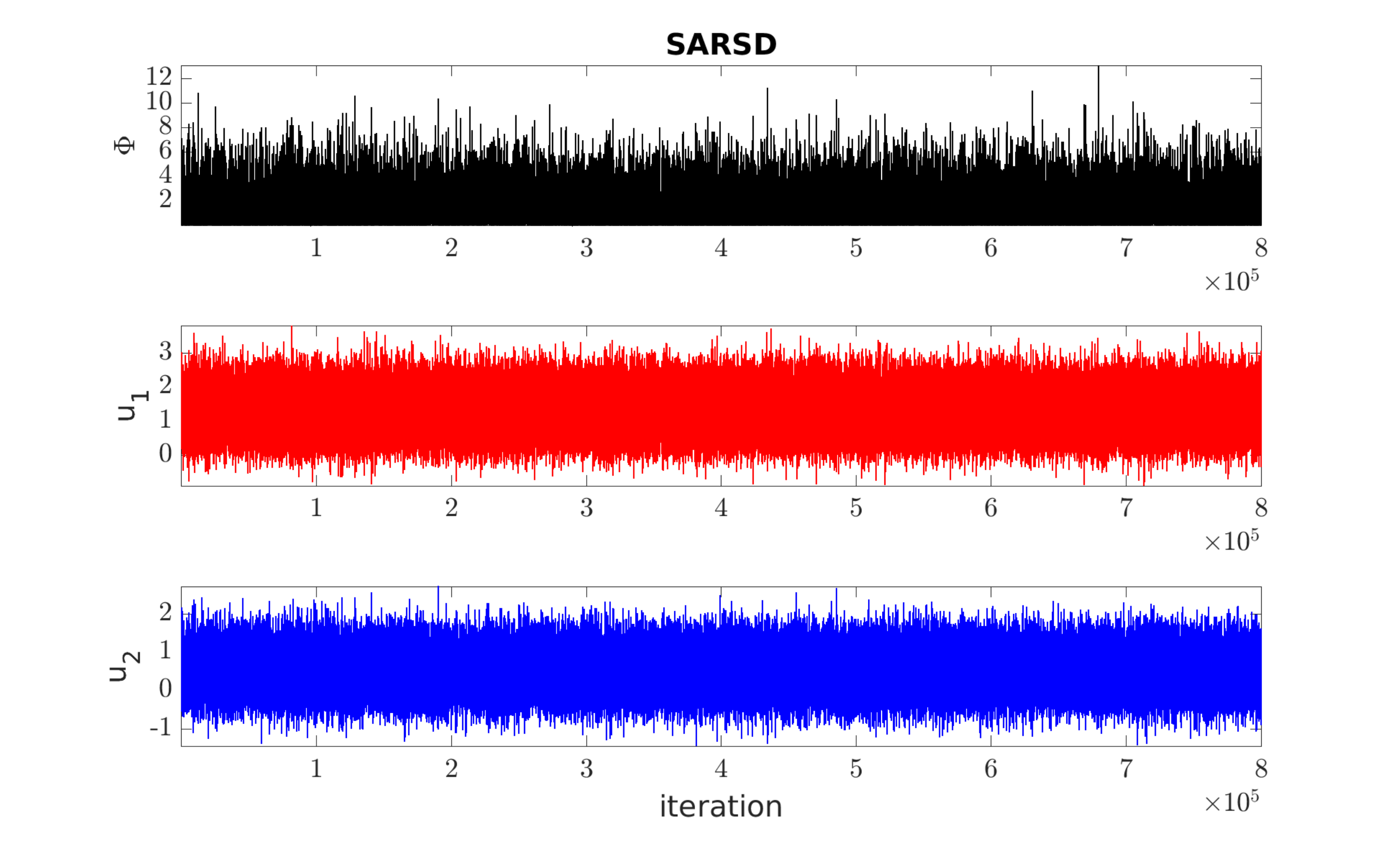}
  \caption{An instance of trace plots of RCAR and SARSD algorithms in Example 4
    with shape parameter $p=1$.}
  \label{fig:trace-plot-and-acceptance-2D-example}
\end{figure}
 
\subsection{Example 6: Denoising in finite dimensions with a gamma prior}
\label{sec:example-2-finite-dim-denoising}
We now turn our attention to an inverse problem with a larger parameter space. 
 Consider the column vector 

 \begin{equation}\label{denoising-truth}
u_0 = (0,0, 1, 0, 0, 1, 0,0, 1, 0,0 , \dots)^T \in \reals^N. 
\end{equation}
That is, every third element is one and the rest of the entries are zero. Now suppose that we observe a noisy version of this sparse vector 
$$
y = u_0 + \epsilon, \qquad \epsilon \sim \mcl{N}(0, \sigma^2 \mb{I}_N),
$$
and we wish to recover $u_0$ from a realization of $y$. We refer to this inverse problem
 as the denoising problem. To solve this inverse problem we employ a 
gamma prior  
$$
\frac{\dd \mu_0}{\dd \Lambda}(t) = \frac{1}{\Gamma(p)^N} \prod_{j=1}^N 
t_j^{p-1} \exp( -t_j) \mb{1}_{(0,\infty)}(t_j) \qquad t = (t_1, t_2, \dots, t_N)^T \in \reals^N.
$$
For the experiments
in this section we took $\sigma=1/4$ and  
$N=10, 20, 40$. Note that as  $N$ changes the size of the data $y$ changes 
as well and so for larger $N$ we are dealing with a larger parameter space and 
more data.  Also, our prior assumption is that the components of $u$ are independent of each 
other and have the same variance. Thus, we expect our algorithms to degrade as $N$ becomes larger.
To sample the posterior we  modified Algorithms~\ref{BK-RCAR-algorithm} and
\ref{BK-SARSD-algorithm} following
 Remark~\ref{gamma-prior-algorithms}.
Our primary goal here was to compare the performance of the RCAR and SARSD
algorithms as a function of the dimension $N$ and step size parameter $\beta$.
We also considered performance of the posterior mean
as a predictor of $u_0$ for both RCAR and SARSD algorithms when $p=1$ and
also for the RCAR algorithm  when $p =1, 2/3$ and $1/3$.

In Figure~\ref{fig:denoising-ACF} we show the 
autocorrelation functions (ACF) of the  
components of the $\{ u^{(j)} \}$  chain for RCAR and SARSD. 
We used a burnin of $5\times 10^4$
and a sample size of $4 \times 10^4$. 
Table~\ref{tab:denoising-ESS-and-IACF-RCAR-v-SARSD}
summarizes our choices of the step size $\beta$ for these simulations as well as 
average acceptance rates, integrated ACF (IACF), and effective sample sizes (ESS).
We tuned the $\beta$ values to achieve an average acceptance ratio of roughly
$0.25$ in all cases.
The reported values of IACF and ESS correspond to the worst performing
(slowest mixing)
component of the chain in each case.

As expected, performance of both algorithms 
suffered with larger $N$. However, an interesting observation is that
the RCAR  performed more consistently across the
different components of the chain. This can be seen clearly in Figure~\ref{fig:denoising-ACF}
where the ACF of the RCAR chain drops consistently across different components
while SARSD has few components that performed well and others that were more correlated.
This behavior also explains the noticeable difference in the reported $\min$ ESS values
for the two algorithms in Table~\ref{tab:denoising-ESS-and-IACF-RCAR-v-SARSD}.
We also show trace plots of two components of RCAR and SARSD with $N=40$
in Figure~\ref{fig:denoising-traceplots}. Both plots appear to have
converged to an stationary distribution but the SARSD trace appears
 thinner than the RCAR trace which is in line with our observation that
the SARSD ACFs decayed slower than RCAR's.

Next, we studied the dependence of the acceptance ratios as a function of
$N$ and $\beta$ for both algorithms. Our results are
summarized in  Figure~\ref{fig:denoising-ave-acceptance-ratio}.
For each value of $N$ and $\beta$  we used a burnin of $4\times 10^4$ iterations and
a total of $2 \times 10^4$ samples and restarted the chain five times with
random initial conditions. We
then averaged the acceptance ratios across the five simulations and for
the entire Markov chain.
As expected, for larger
$N$ a smaller step size $\beta$  was needed to achieve the same
acceptance ratios for both RCAR and SARSD. A noticeable
difference between the two algorithms was that for fixed $\beta$, SARSD
appeared to consistently have a higher acceptance rate than RCAR (see Figure~\ref{fig:denoising-ave-acceptance-ratio}). 

Finally,  we compare
the posterior mean and median of RCAR and SARSD as pointwise
estimators of  $u_0$ in Figure~\ref{fig:denoising-summary-stats-RCAR-v-SARSD}. We did not
observe a significant difference in the quality of  the
mean and the median as predictors of $u_0$ and both algorithms appear to have
converged in mean and median. In Figure~\ref{fig:denoising-summary-stats-RCAR-p-study}
we show the posterior mean of RCAR against $u_0$ for different values of the
shape parameter $p$. We observe that smaller values of $p$ shrank the mean towards
zero resulting in better approximation of the zero components but
worse approximation of some of the non-zero components of $u_0$. In the next
section we will thoroughly study the effect of the $p$ parameter on
the performance of RCAR.
\begin{figure}[htp]
  \centering
  \includegraphics[height=0.222\textwidth]{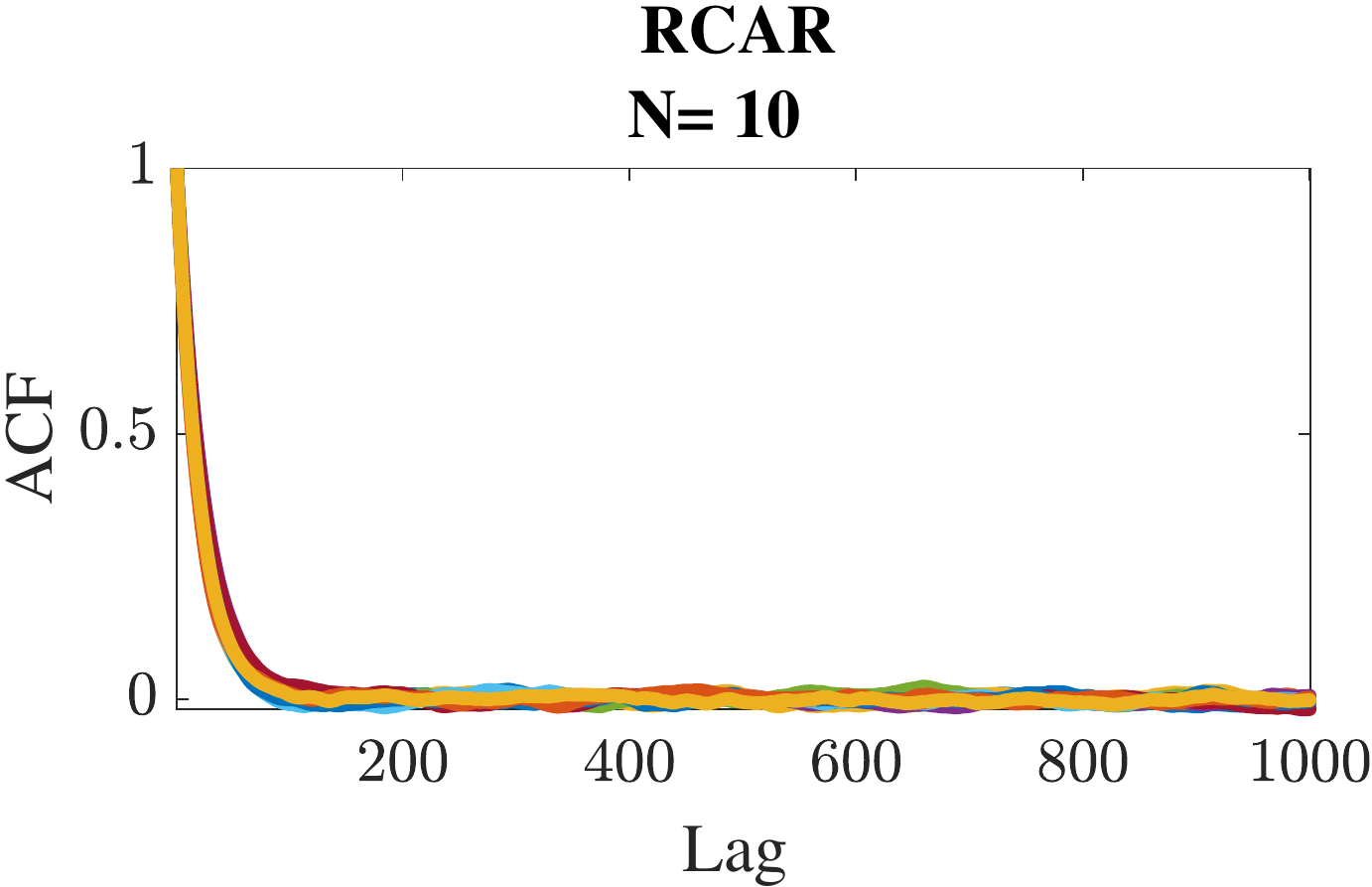}\qquad
   \includegraphics[height=0.222\textwidth]{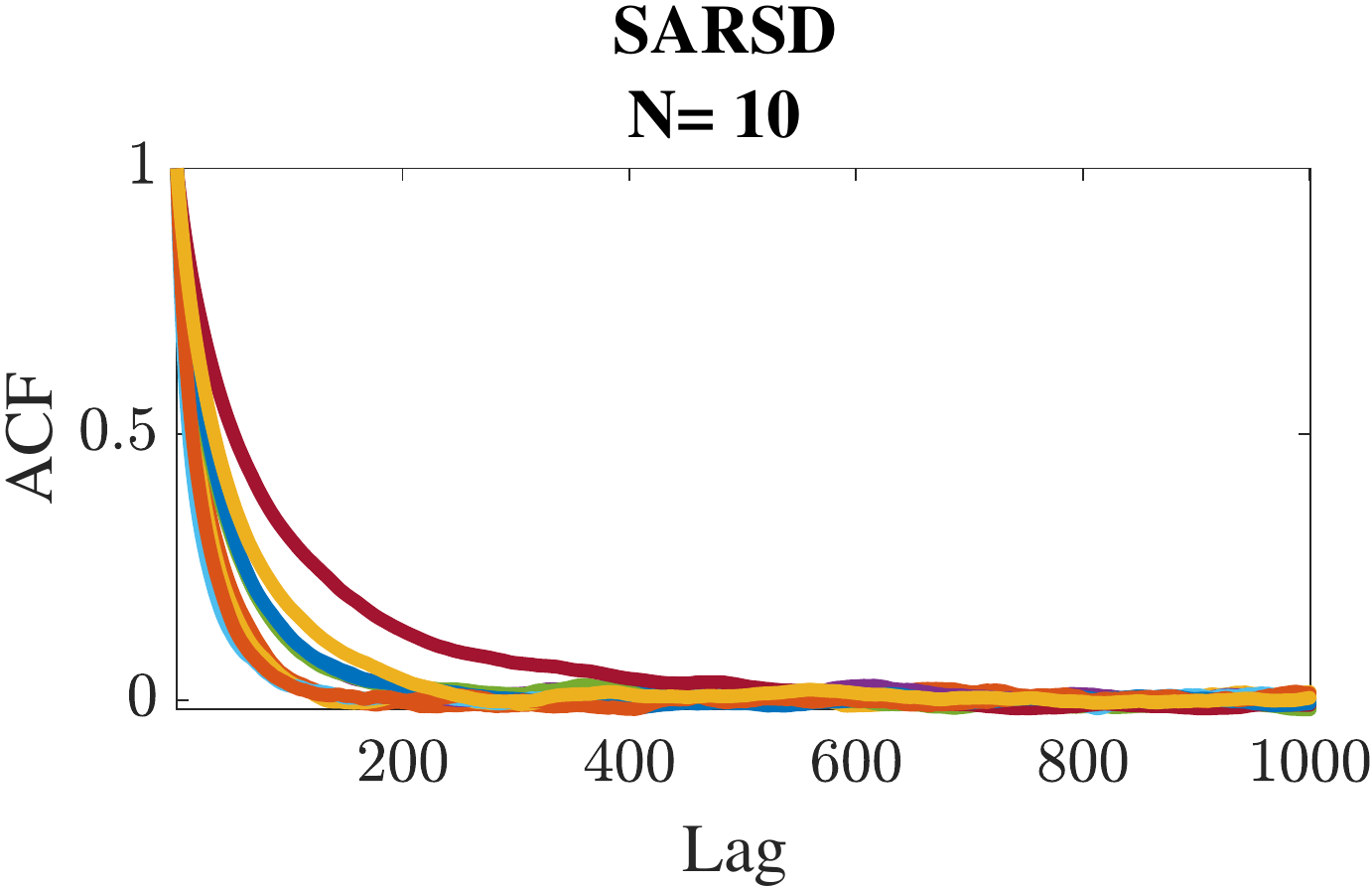}\\
   \includegraphics[height=0.2\textwidth]{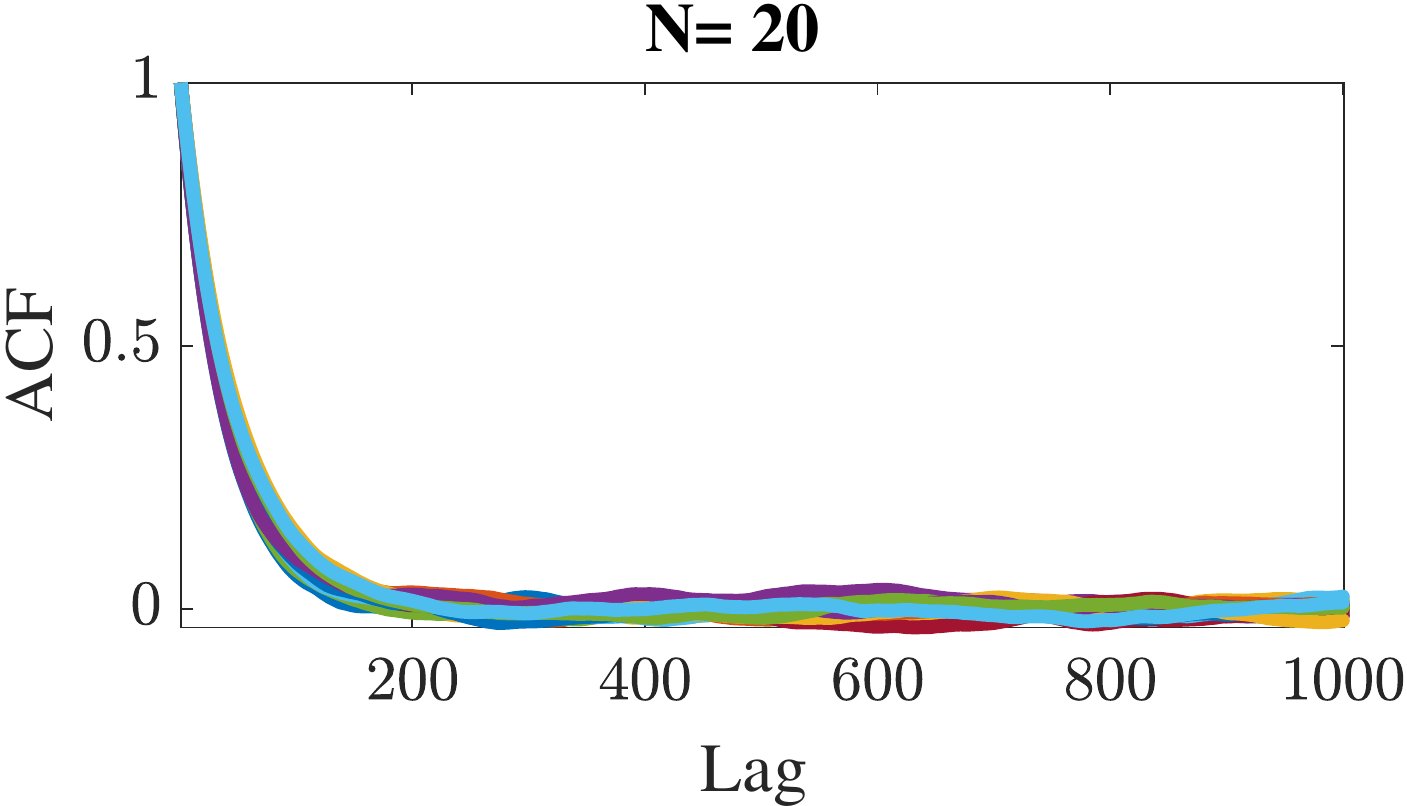}\qquad
    \includegraphics[height=0.2\textwidth]{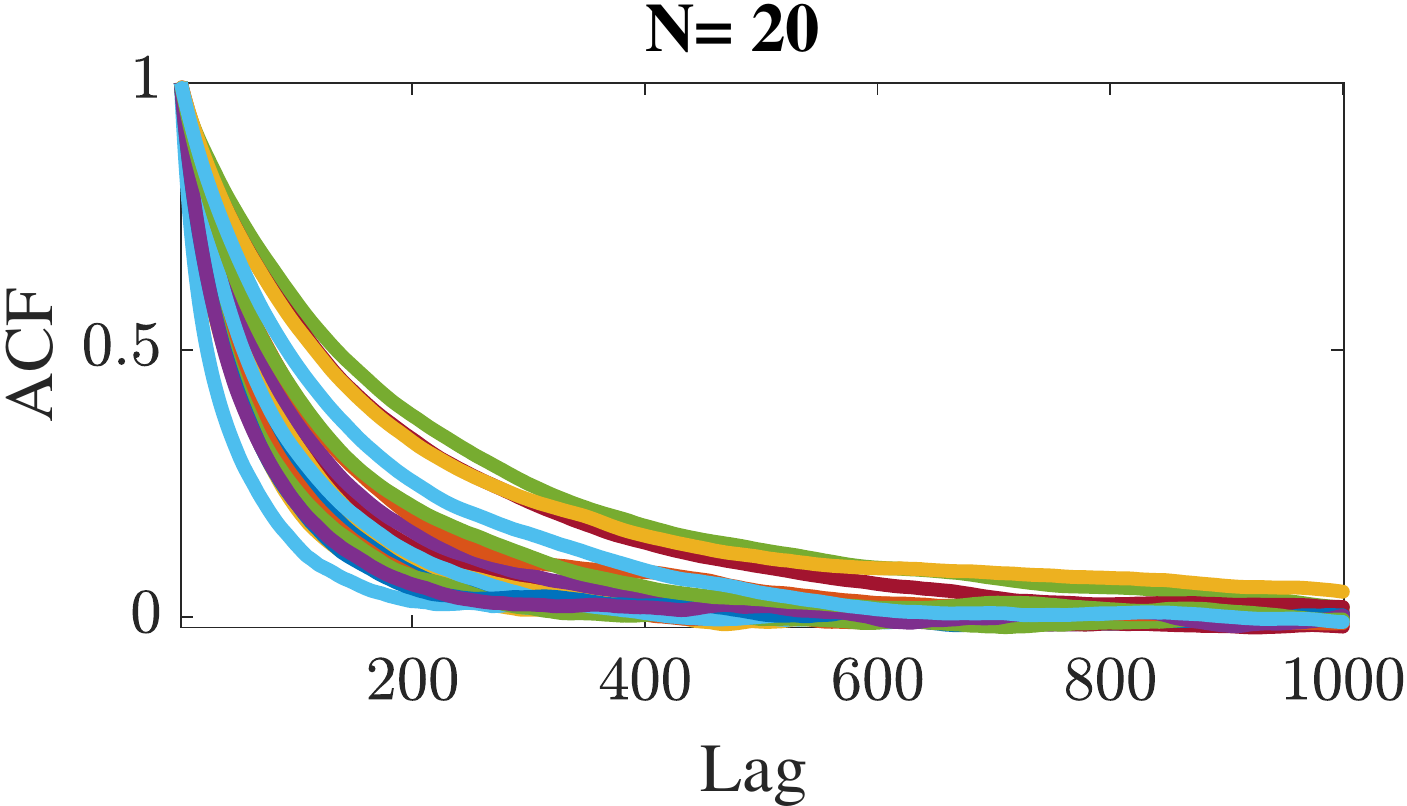}\\
    \includegraphics[height=0.2\textwidth]{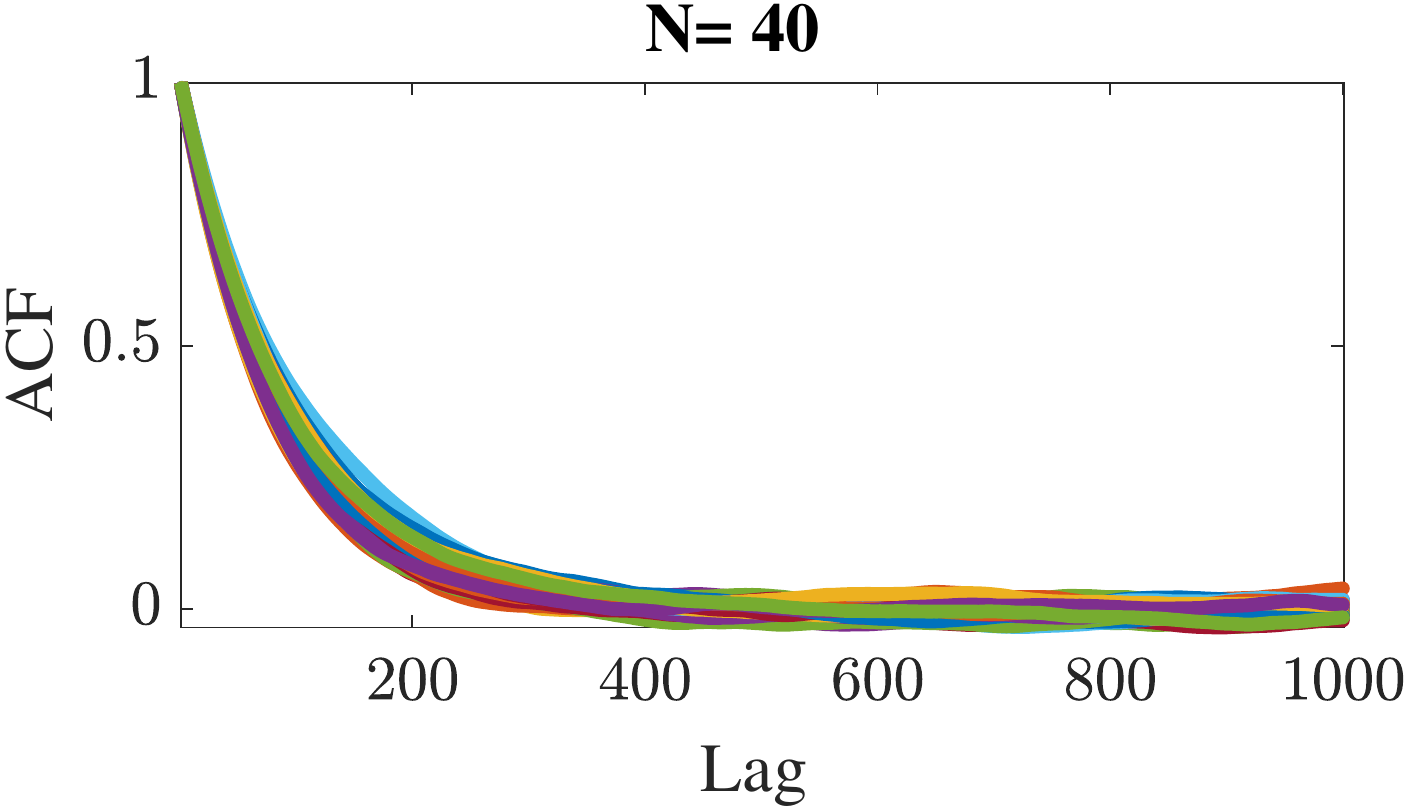}\qquad
     \includegraphics[height=0.2\textwidth]{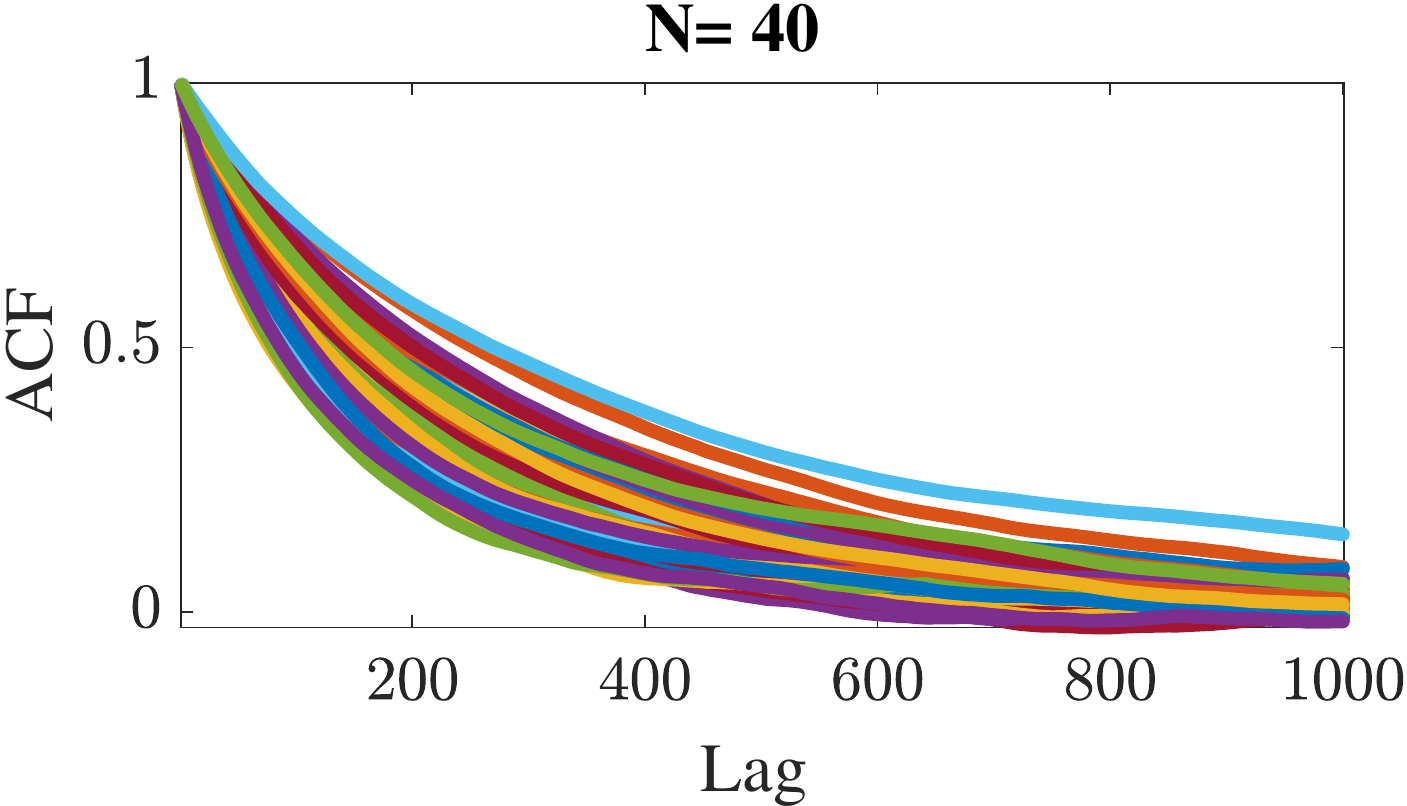} 
     \caption{Approximated autocorrelation functions of separate components of the chain for
       RCAR and SARSD algorithms in different  dimensions $N = 10, 20$ and $40$ in Example 6.
       Values of $\beta$ for each experiment are  given in Table~\ref{tab:denoising-ESS-and-IACF-RCAR-v-SARSD}
       and both algorithms are tuned to achieve acceptance ratio of roughly $0.25$.}
  \label{fig:denoising-ACF}
\end{figure}

\begin{table}[htp]
  \centering
  \begin{tabular}{| c| c | c | c | c | c |  }
     \hline
 & \multicolumn{1}{c|}{$N$} & \multicolumn{1}{c|}{$\beta$} & \multicolumn{1}{c|}{average $a(\cdot,\cdot)$ } & \multicolumn{1}{c|}{$\max$ IACF} & \multicolumn{1}{c|}{ $\min$ ESS (per $10^4$ steps)}\\
 \hline
 \parbox[t]{2mm}{\multirow{3}{*}{\rotatebox[origin=c]{90}{RCAR}}} & 10 & 0.900& 0.25 & 49.24 & 202 \\ 
     &20 & 0.950& 0.25 & 105.15 & 95 \\ 
     &40 & 0.975& 0.23 & 220.78 & 45 \\ 
    \hline
    \parbox[t]{2mm}{\multirow{3}{*}{\rotatebox[origin=c]{90}{SARSD}}} & 10 & 0.800& 0.22 & 185.22 & 53 \\ 
     &20 & 0.900& 0.24 & 447.66 & 22 \\ 
     &40 & 0.950& 0.25 & 771.25 & 13 \\ 
 \hline
   \end{tabular}
   \caption{Summary statistics of the Markov chains of the RCAR and SARSD algorithms in  Example 6
     for different values of $N$. Performance of both algorithms deteriorated for larger $N$ as evident in the ESS values. The $\max$ IACF and $\min$ ESS values were computed over the $N$ components of
   the  Markov chains. }
  \label{tab:denoising-ESS-and-IACF-RCAR-v-SARSD}
\end{table}

\begin{figure}[htp]
  \centering
  \includegraphics[height=.2\textwidth]{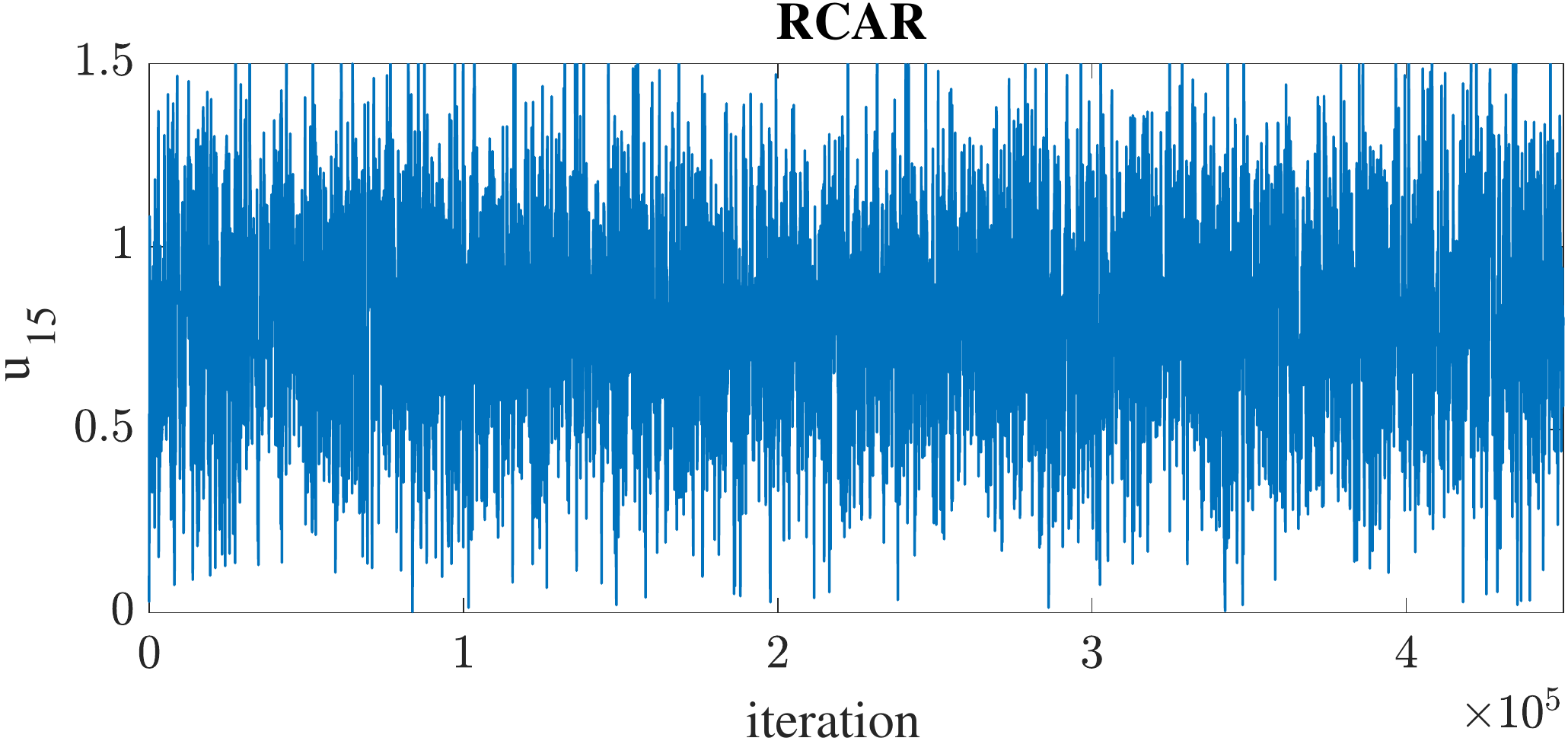}\qquad
  \includegraphics[height=.2\textwidth]{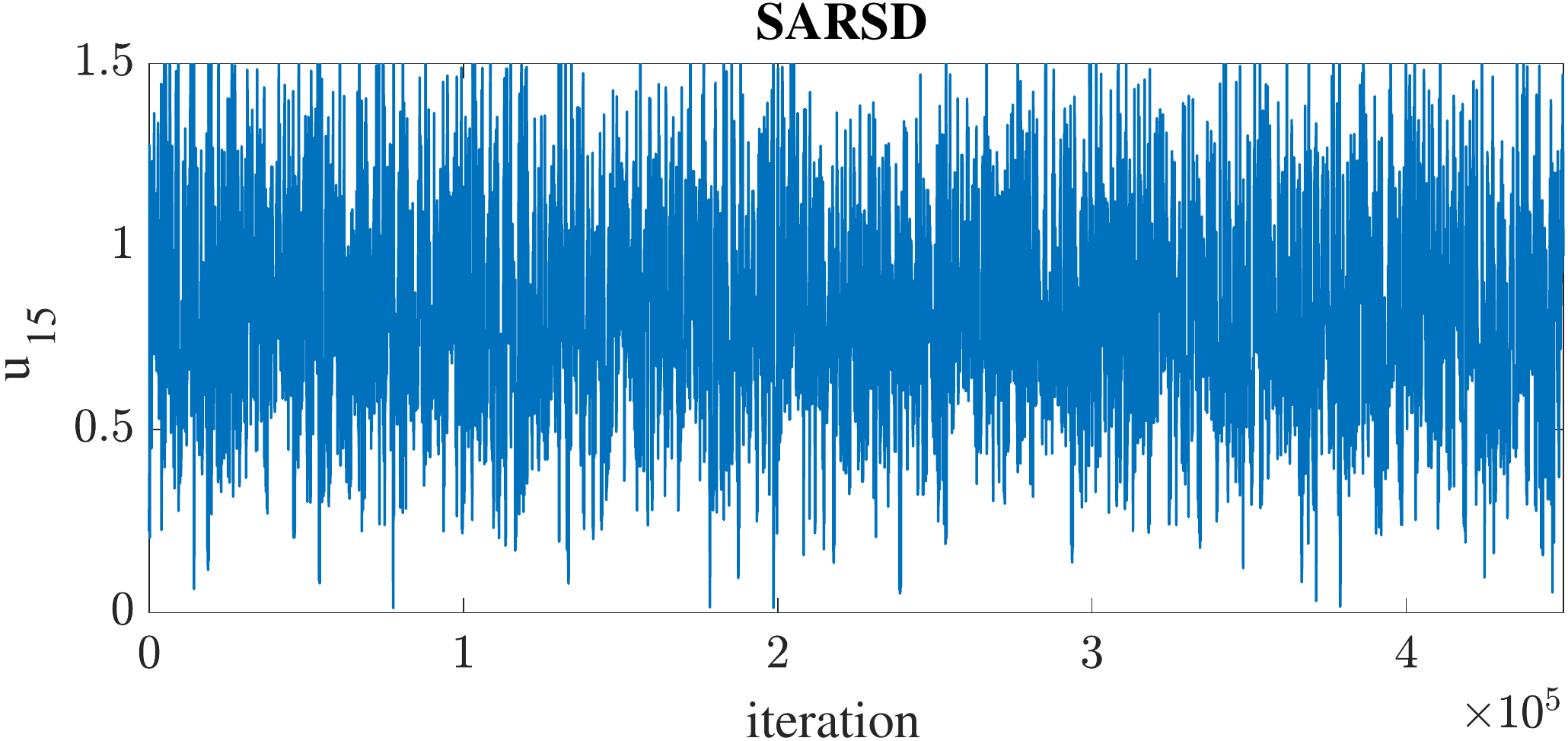}\\
  \includegraphics[height=.19\textwidth]{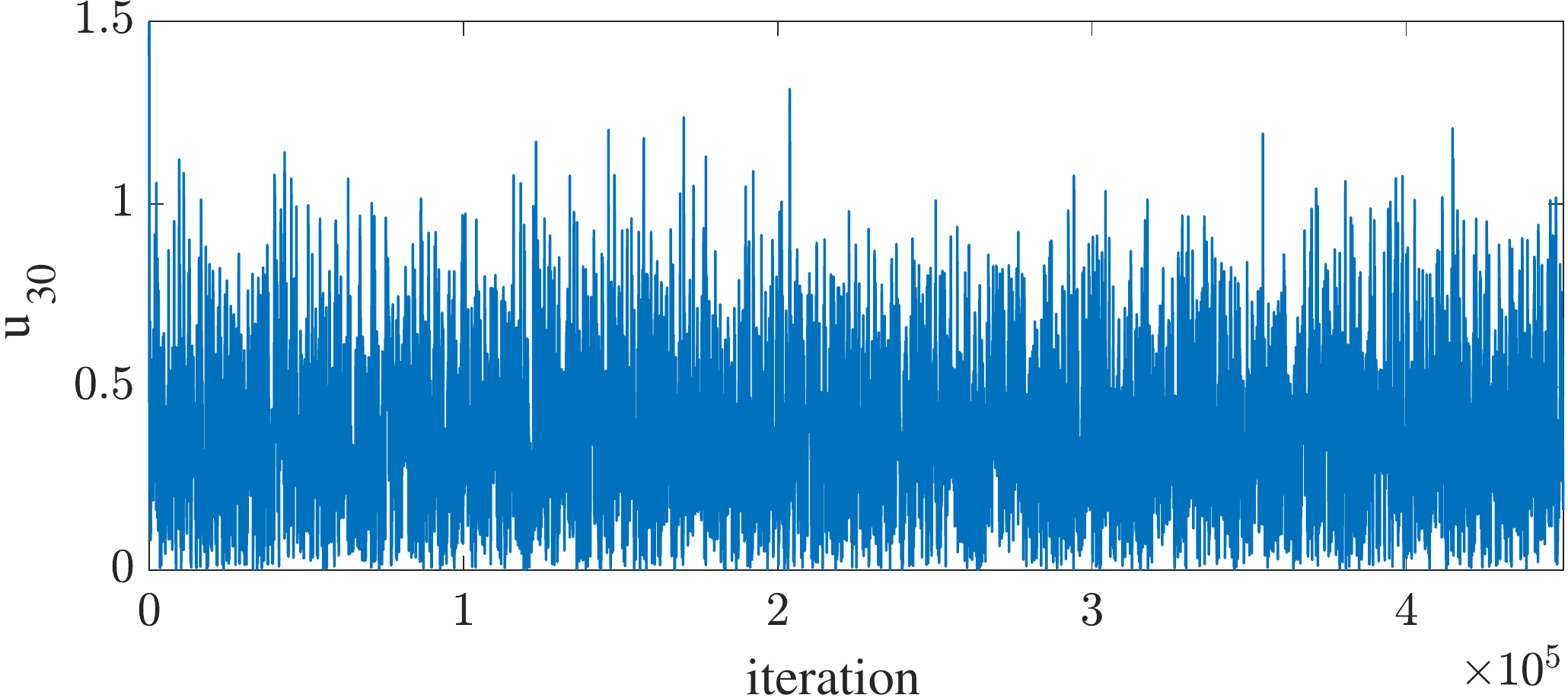}\qquad
  \includegraphics[height=.19\textwidth]{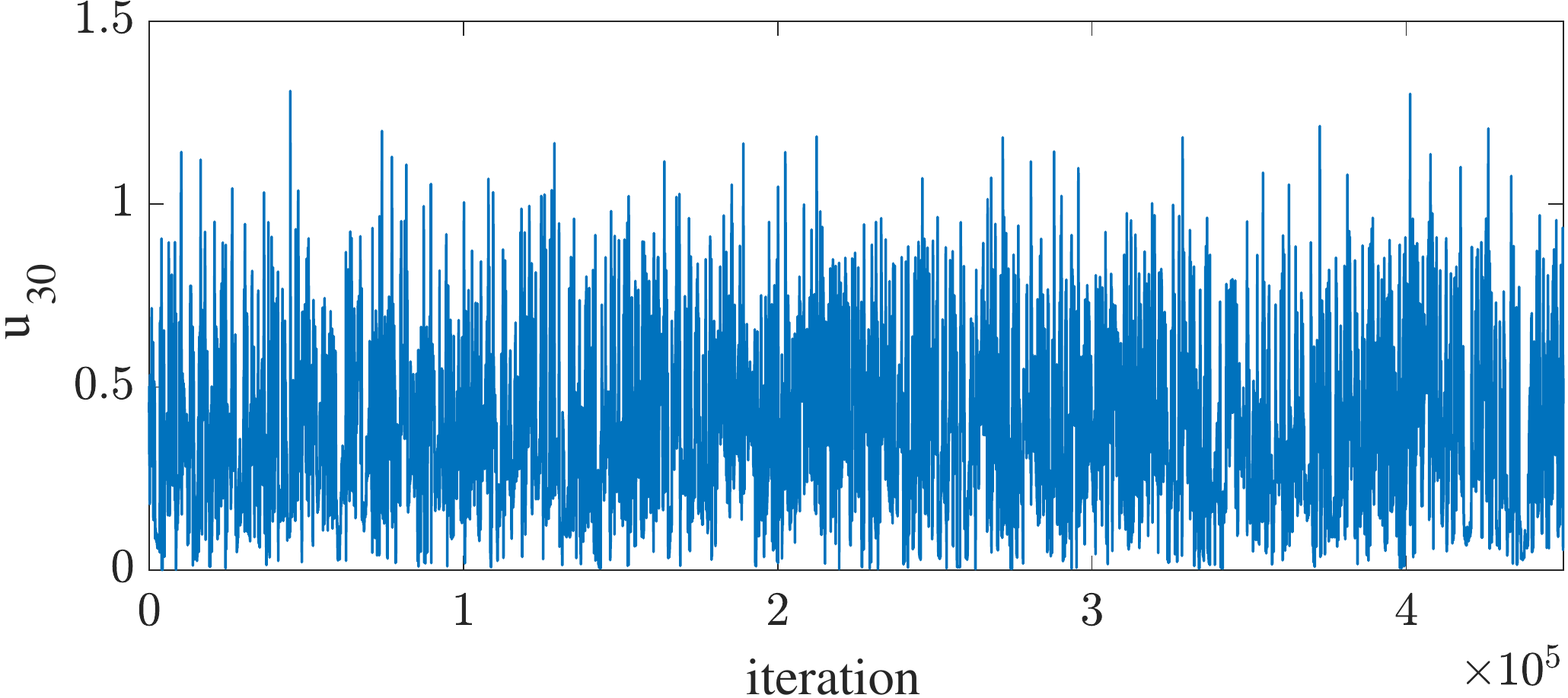}
  \caption{Side by side comparison of an instance of RCAR and SARSD trace plots
    for the $15^{th}$ and $30^{th}$ components of the Markov chains with $N=40$ and
    $p=1$.}
\label{fig:denoising-traceplots}
\end{figure}

\begin{figure}[htp]
  \centering
  \includegraphics[height=.35\textwidth]{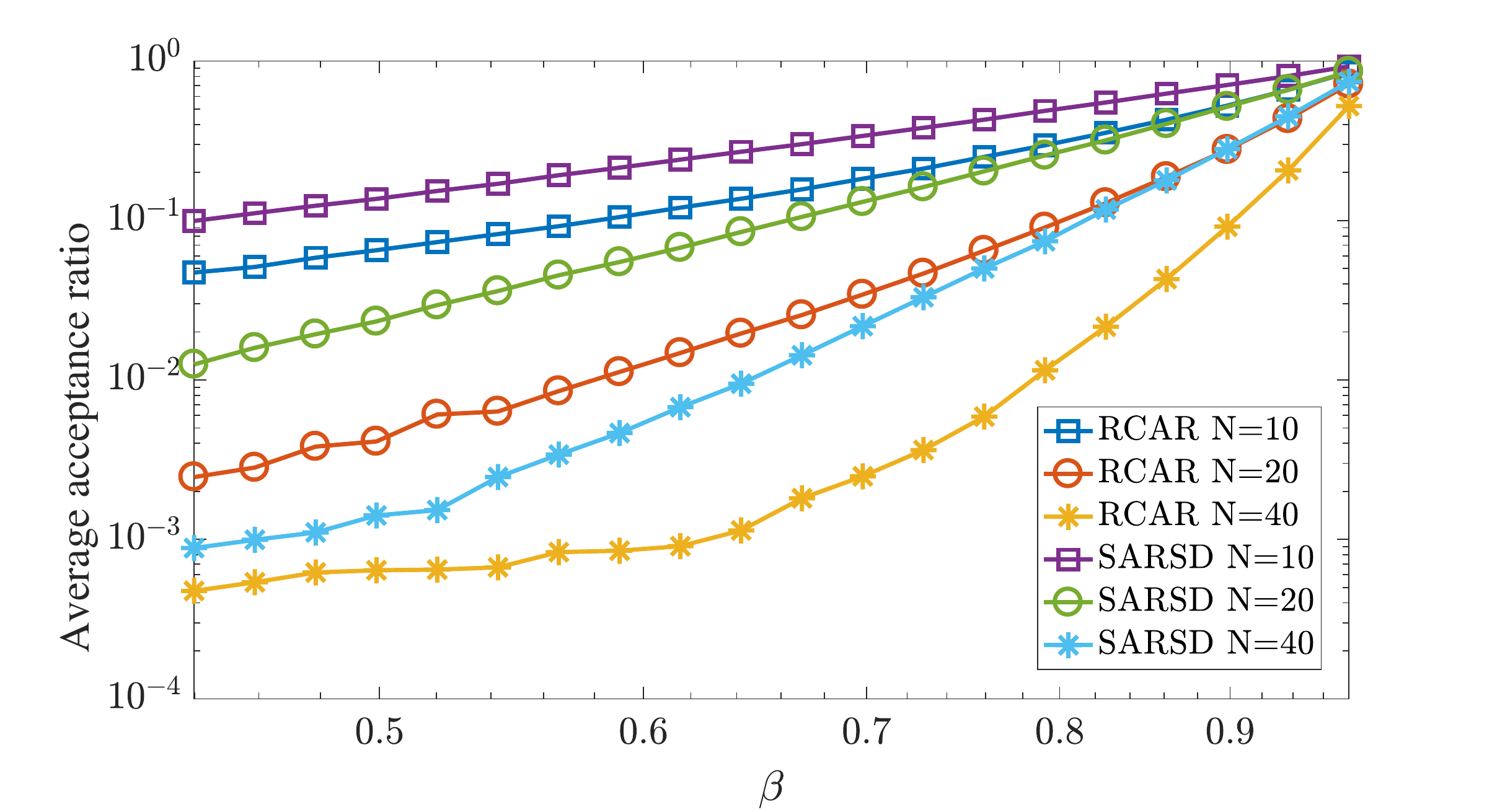}
\caption{ Average acceptance ratio of RCAR and SARSD for different values of $N$
as a function of $\beta$ in the denoising problem of Example 6. The parameter $p =1$ in all cases.
 Average acceptance ratio of both  algorithms  deteriorates as $N$ becomes larger.
For fixed values of $\beta$ the SARSD algorithm appears to have consistently higher acceptance rates. }
\label{fig:denoising-ave-acceptance-ratio}
\end{figure}

\begin{figure}[htp]
  \centering
  \includegraphics[height=.24\textwidth]{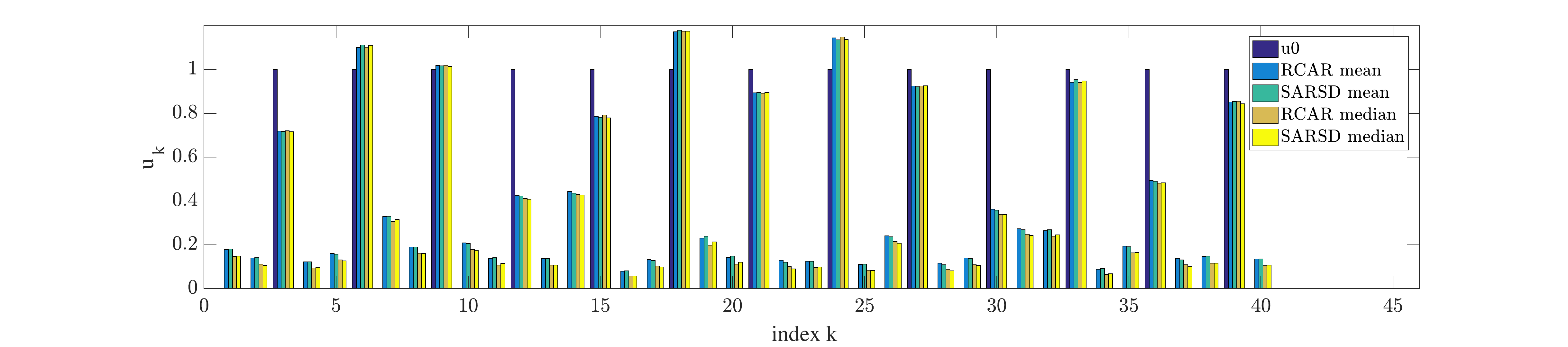}
  \caption{Comparison of posterior mean and median of $u = (u_1, \dots, u_{40})$ obtained from RCAR and SARSD algorithms against the original vector $u_0$
    as in \eqref{denoising-truth} with  $p =1$.}
\label{fig:denoising-summary-stats-RCAR-v-SARSD}
\end{figure}

\begin{figure}[htp]
  \centering
  \includegraphics[height=.24\textwidth]{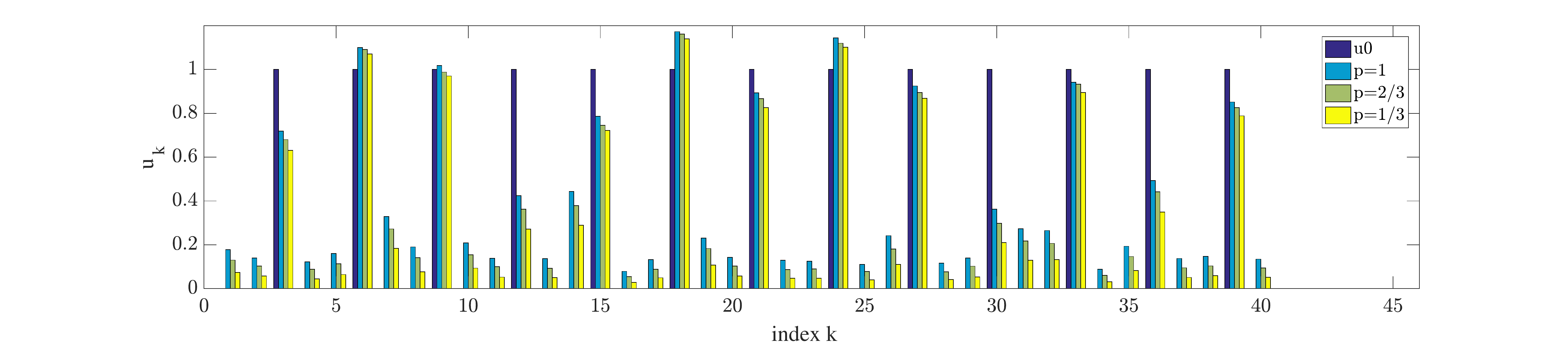}
  \caption{Comparison of the posterior mean  of $u = (u_1, \dots, u_{40})$
    obtained from RCAR  against the original vector $u_0$
    as in \eqref{denoising-truth} for different values of  $p =1, 2/3, 1/3$.}
\label{fig:denoising-summary-stats-RCAR-p-study}
\end{figure}

\subsection{Example 7: Deconvolution on the circle with Bessel-K priors}
\label{sec:deconv}
Here we consider an inverse problem on $L^2(\mbb T)$. For this
example we only used the lifted RCAR algorithm since we want to take the shape
parameter $p <1$ to promote compressibility of the mean and the samples.

 Consider the problem 
of estimating a function $u_0 \in L^2(\mbb{T})$ from a few point values of its convolution 
with a kernel $\kappa$. This  is a classic benchmark problem in the inverse problems
literature referred to as the deconvolution problem \cite{ somersalo, lucka-fast-Gibbs, vogel, marzouk-randomize-optimize}.
 Take the original function
$$
u_0(t) = \left\{ 
  \begin{aligned}
    &1 \qquad &&t \in [1/4, 3/4], \\
    & 0 \qquad &&\text{otherwise},
  \end{aligned}
\right.
$$
consider the kernel
$$
\kappa_0(t) = \left\{
  \begin{aligned}
    & 1-|t| \qquad && |t| \le 1 \\ 
    & 0 \qquad && \text{otherwise},
  \end{aligned}
\right.
$$
and define the family of convolution kernels
\begin{equation}
  \label{deconvolution-kernel}
  \kappa_\varepsilon := \frac{1}{\varepsilon} \kappa_0(t/\varepsilon).
\end{equation}
Suppose  measurements are obtained  as pointwise values
of $(\kappa_\varepsilon \ast u_0)(t)$ on a uniform
grid of size $M=20$ points on $[0.01, 0.99]$. By putting the convolution and pointwise evaluation 
operators together we can define a forward map $\mcl{G}: L^2(\mbb{T}) \mapsto \reals^M$
taking the function $u_0$ to the measurements $y$. We further assume that 
measurement noise is additive Gaussian and so
$$
y = \mcl{G}(u_0) + \epsilon, \qquad \epsilon \sim \mcl{N}(0, \sigma^2\mb{I}),
$$
giving rise to a quadratic likelihood potential of the form \eqref{quadratic-likelihood}.

We now define our prior.
 Let $\{ r_k\}_{k=1}^\infty$ be the Haar wavelet basis in $L^2(\mbb{T})$:
$$
r_0(t) =1, \qquad 
r_1(t) = 2(\mb{1}_{\{t \le 1/2\}}(t) - 1/2)
$$
and for $j = 1,2, \dots$ and $m = 0, 1, 2,\dots, 2^j-1$ define 
$$
r_{2^j + m}(t) = 2^{j/2} r_1(2^{j}t - m).
$$
Also consider the sequence $\{\gamma_k\} \in \ell^2$:
$$
\gamma_0 =\gamma_1 = 1, \qquad \text{and} \qquad \gamma_{2^j + m} = 2^{-2j}
$$
for $j =1,2,3, \dots$ and $m = 0, 1, 2, \dots, 2^j-1$. We then define 
the prior measure

\begin{equation}\label{deconvolution-prior}
 \mu_0 = \Law \left\{ u = \lambda \sum_{k=0}^\infty   \gamma_k \eta_k r_k ,\qquad \text{where} \qquad \eta_k
  \stackrel{iid}{\sim} BK(p,1)\right\}.
\end{equation}
Here $\lambda \in (0,\infty)$ is a fixed hyperparameter that can be used to control 
the global variance of the wavelet modes. With the likelihood and prior identified
we turn our attention to solving the inverse problem.

We discretized the problem at two stages. We approximated the prior 
$\mu_0$ with $\mu_{0,N}$ by truncating the sum in \eqref{deconvolution-prior}
up to $N$ terms and  discretized the convolution operator 
using the composite midpoint rule on a uniform grid of size $128$ points.
We performed wavelet transforms using the Rice 
Wavelet Toolbox \cite{rice-wavelet}
and
employed linear interpolation to approximate the pointwise evaluations.
For the numerical experiments we generated a fixed synthetic dataset by solving
the discrete forward problem with added \myhl{Gaussian noise with
  standard deviation $0.05$. We used a different mesh to generate the data to
avoid the so-called inverse crimes.} 
In 
Figure~\ref{fig:deconvolution-data-posterior-mean-samples}(a) we show the original function $u_0$, the convolution $\kappa_\varepsilon \ast u_0$ with $\varepsilon = 1/16$
and the fixed realization of the dataset $y$.
For the time being we fix 
 $\varepsilon= 1/16$ and the dataset $y$  shown in Figure~\ref{fig:deconvolution-data-posterior-mean-samples}(a). We discuss the effect of the
dilation parameter $\varepsilon$ in Subsection~\ref{sec:kernel-width-study}.
\subsubsection{Posterior statistics}

We begin by presenting certain posterior statistics obtained  from the lifted RCAR algorithm.
We fixed $p = 2/3$, $\lambda=1$ and discretized the prior by truncating 
\eqref{deconvolution-prior} up to  $N = 8, 16, 32, 64, 128$ terms (the dimension of
the parameter space is $N$). We used a burnin of $5 \times 10^4$ samples 
and ran  lifted RCAR for $5\times 10^5$ steps with $\beta = 0.97$.
We chose this value of $\beta$ to achieve an acceptance ratio in the range of 
$0.25$ to $0.3$ for all values of $N$. In Subsection~\ref{sec:RCAR-deconvolution-performance}
we further analyze the acceptance ratio and its dependence on $N$.
 
 In Figure~\ref{fig:deconvolution-data-posterior-mean-samples}(b)
 we show the posterior mean for different choices of $N$. The mean appears to
 converge
as $N$ increases \myhl{ and is  able to find the discontinuities in the 
original function and  match  their height.}
The mean is less regular as compared to the true solution $u_0$, most likely
due to noise in the data $y$. 

In Figure~\ref{fig:deconvolution-data-posterior-mean-samples}(c) we show a few
independent samples from the posterior 
in the case when $N=128$. The samples were chosen to be far enough apart that they can be regarded 
as independent according to the estimated ESS of the worst performing component of the chain.
We note that the posterior samples also have the correct location of
the discontinuities. 

Figure~\ref{fig:deconvolution-wavelet-modes} shows the posterior mean and standard deviation of the 
first eight wavelet coefficients of the Markov chain (i.e., $\{ \gamma_k\eta_k \}_{k=1}^8$).
We observe that the posterior mean is a close match to the
true value of the wavelet coefficients of $u_0$ which reaffirms our
initial observation that the posterior mean is a good predictor of $u_0$.
An interesting observation is that posterior standard deviations of
the wavelet modes were consistent across different modes. Indicating that, at least
the first few 
 modes of $u_0$ are approximated with more or less the same uncertainty. 

 Finally, Figure~\ref{fig:deconvolution-2D-posterior-hist} shows two-dimensional histograms
 of the first five
 wavelet modes of $u$. In comparing the fifth wavelet coefficient $u_5$ against
 $u_1$ to $u_4$ (i.e, the last row in Figure~\ref{fig:deconvolution-2D-posterior-hist})
 we observe some concentration of the posterior mass around $u_5= 0$.

\begin{figure}[htp]
  \centering
  \raisebox{.24\textwidth}{a)}
  \includegraphics[height=.24\textwidth]{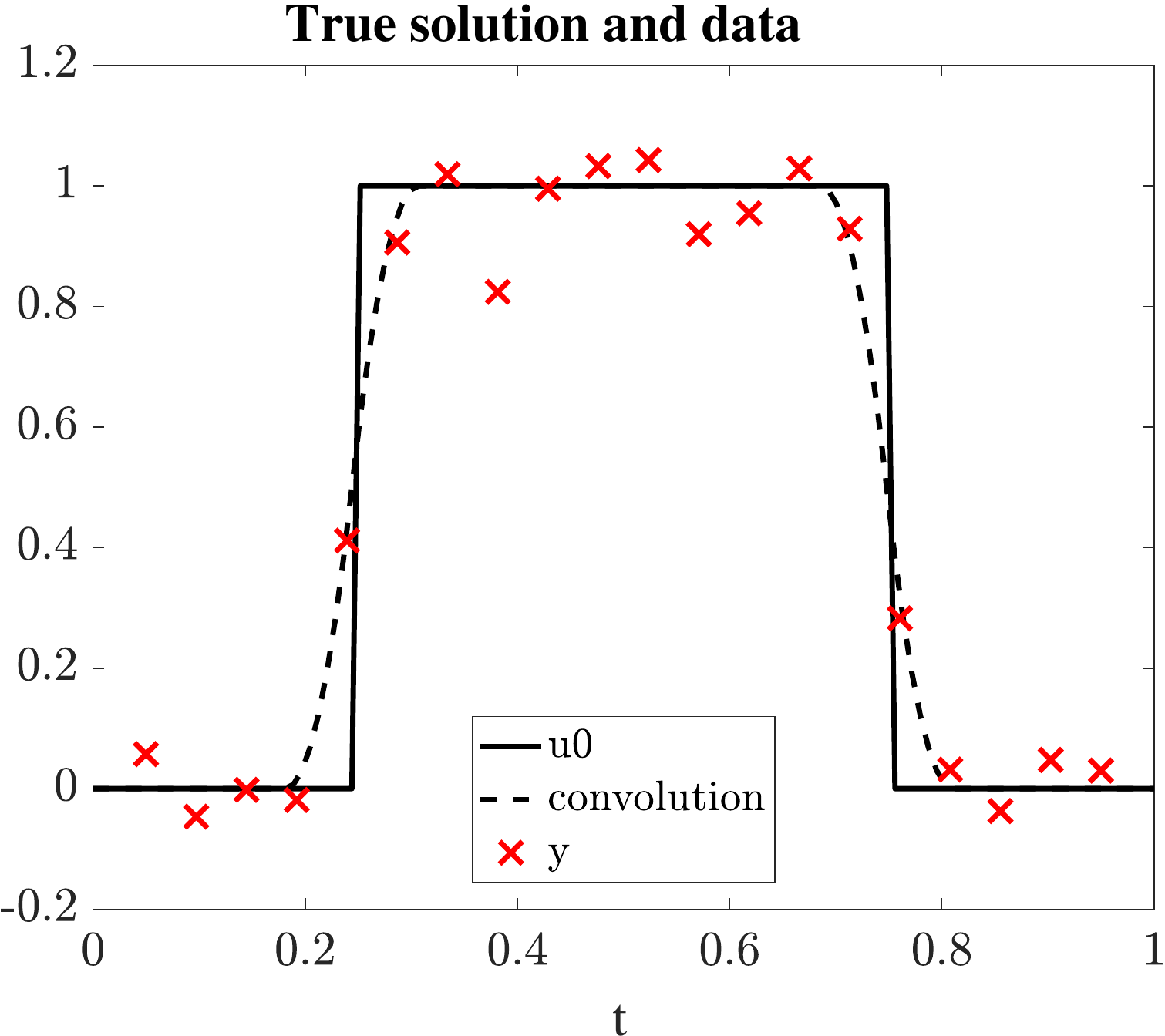}
  \raisebox{.24\textwidth}{b)}
  \includegraphics[height=.24\textwidth]{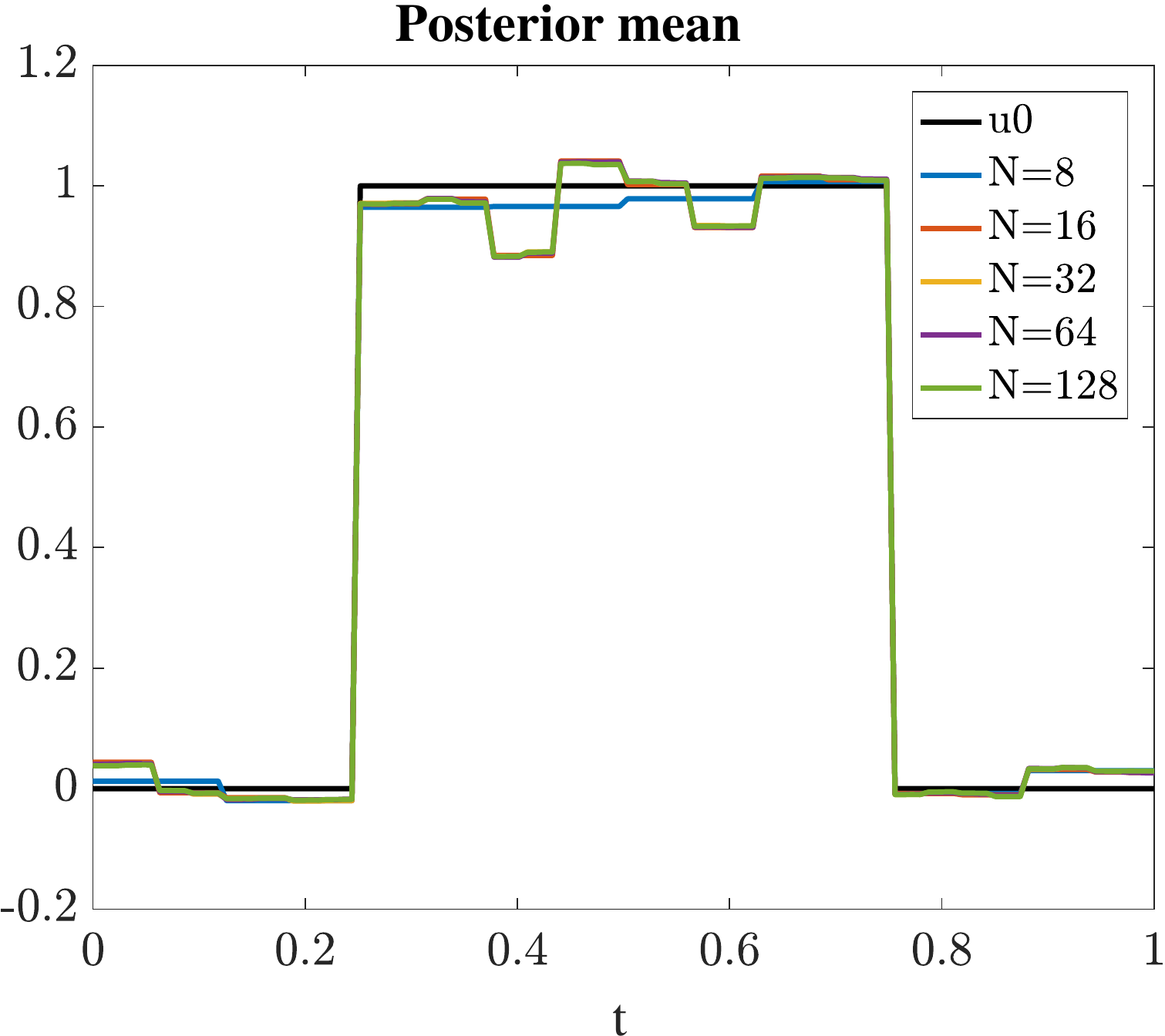}
  \raisebox{.24\textwidth}{c)}
  \includegraphics[height=.24\textwidth]{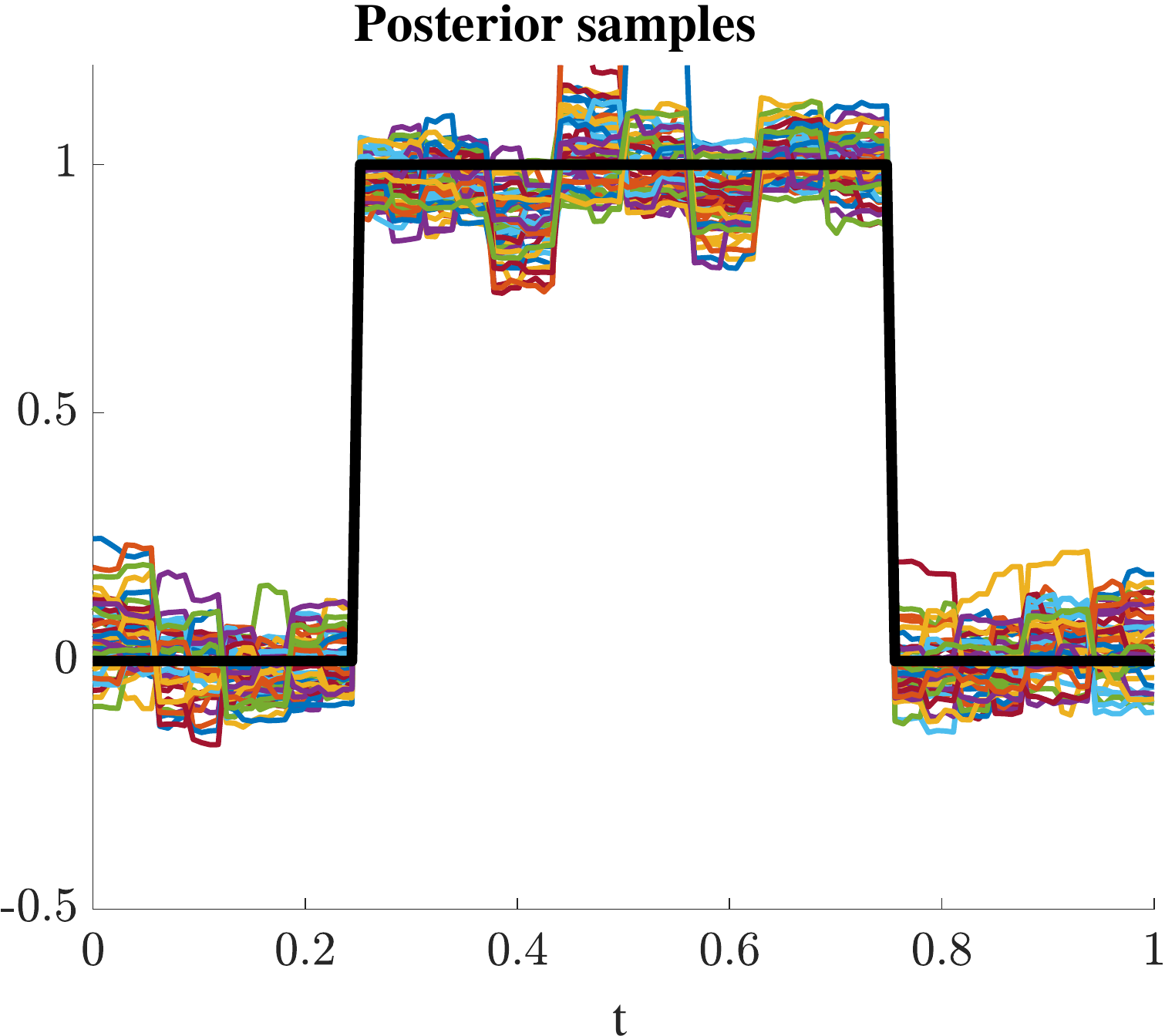}
  \caption{a) True function $u_0$, its convolved version with
    kernel width $\varepsilon = 1/16$ and the noisy pointwise measurements $y$ in the deconvolution problem of Example 7. b) The 
    posterior mean of the deconvolution problem for different values of $N$ (the number of wavelet coefficients). c) A few posterior samples in the deconvolution example with $N=128$. The samples were
    taken to be far enough from each other 
that they can be considered as independent according to the ACF of the worst performing component of the chain.}
  \label{fig:deconvolution-data-posterior-mean-samples}
\end{figure}

\begin{figure}[htp]  \centering
  \includegraphics[width= .48 \textwidth]{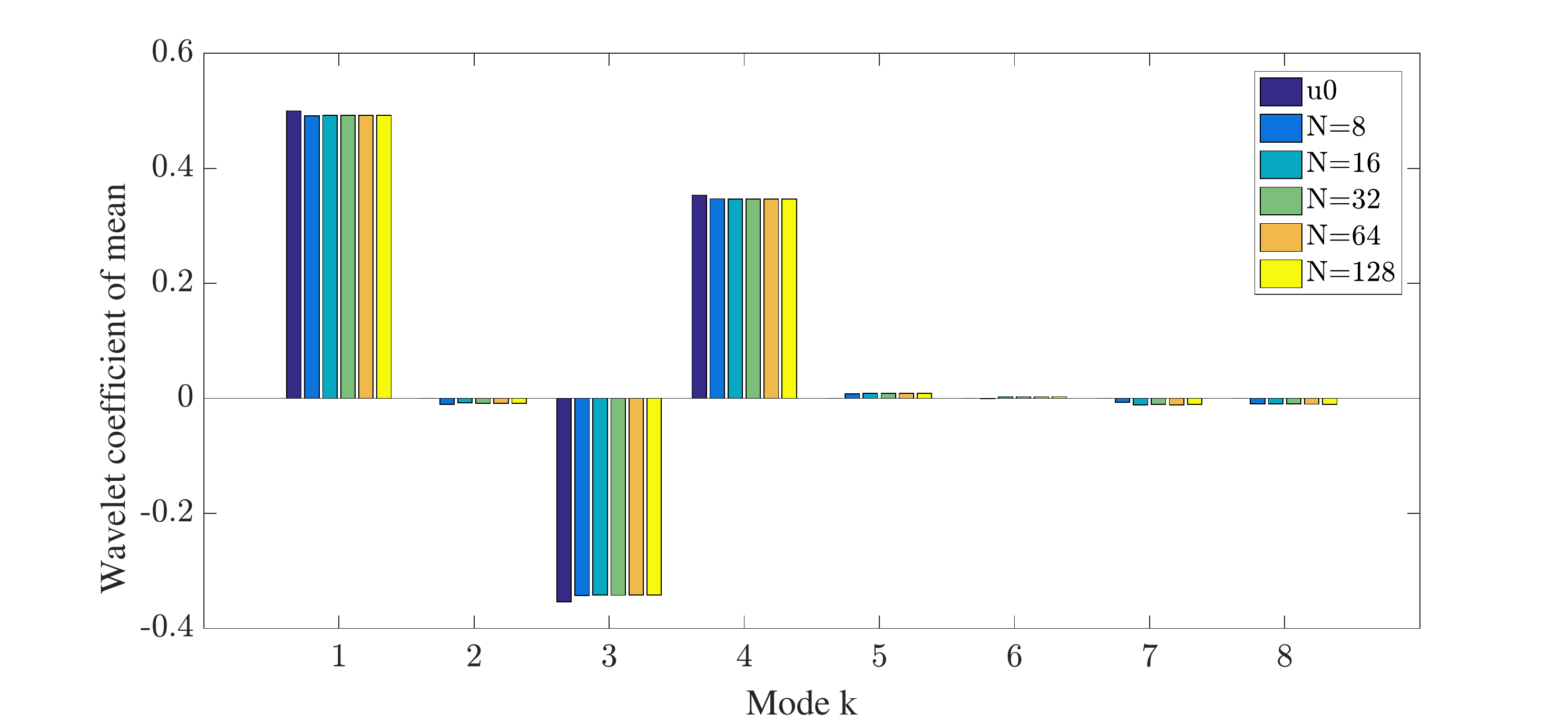}
\quad 
  \includegraphics[width= .48 \textwidth]{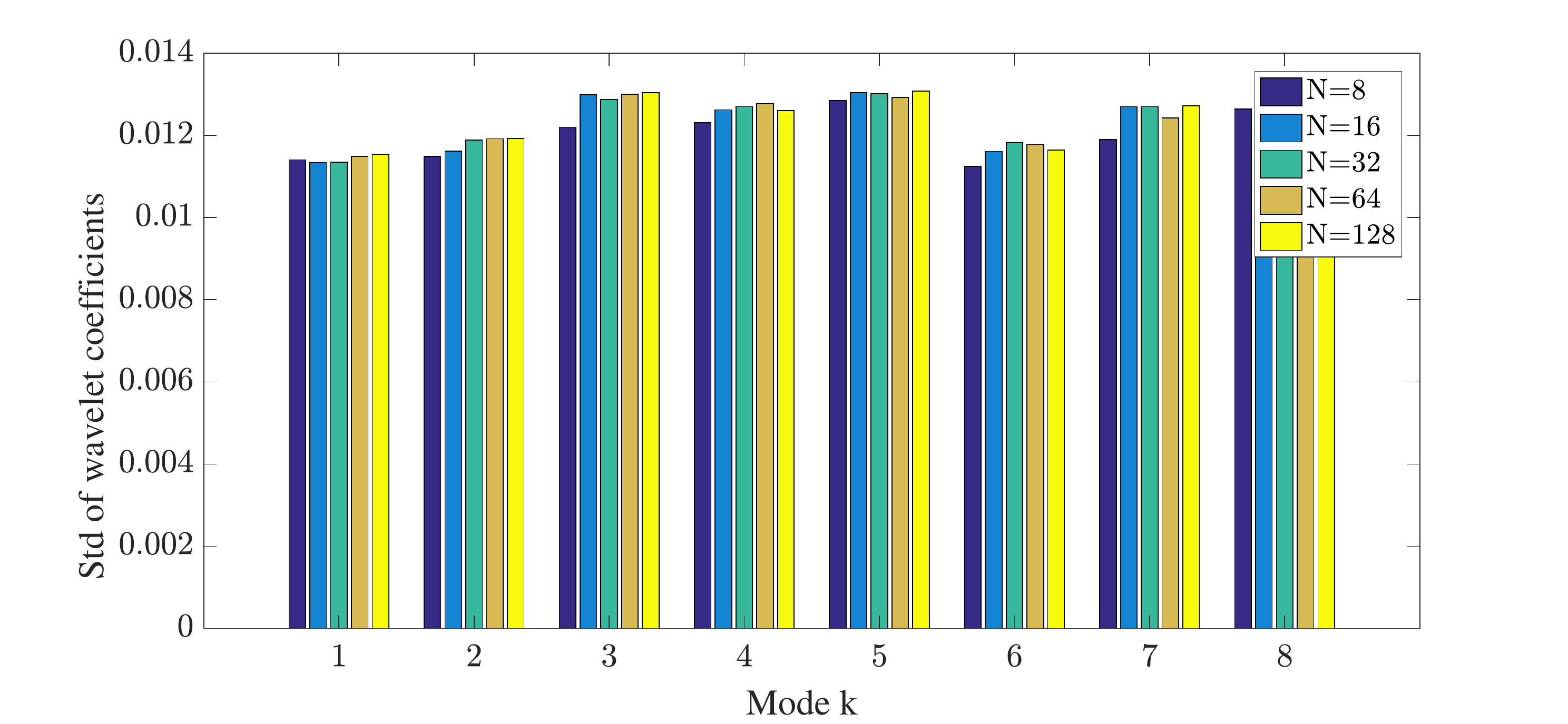}
  \caption{Posterior mean and standard deviation of the first eight wavelet coefficients in the deconvolution problem with $p=2/3, \lambda = 1$ and for different values of $N$.}
\label{fig:deconvolution-wavelet-modes}
\end{figure}

\begin{figure}[htp]
  \centering
\begin{tabular}{c c c c c}
  \includegraphics[width = .2\textwidth]{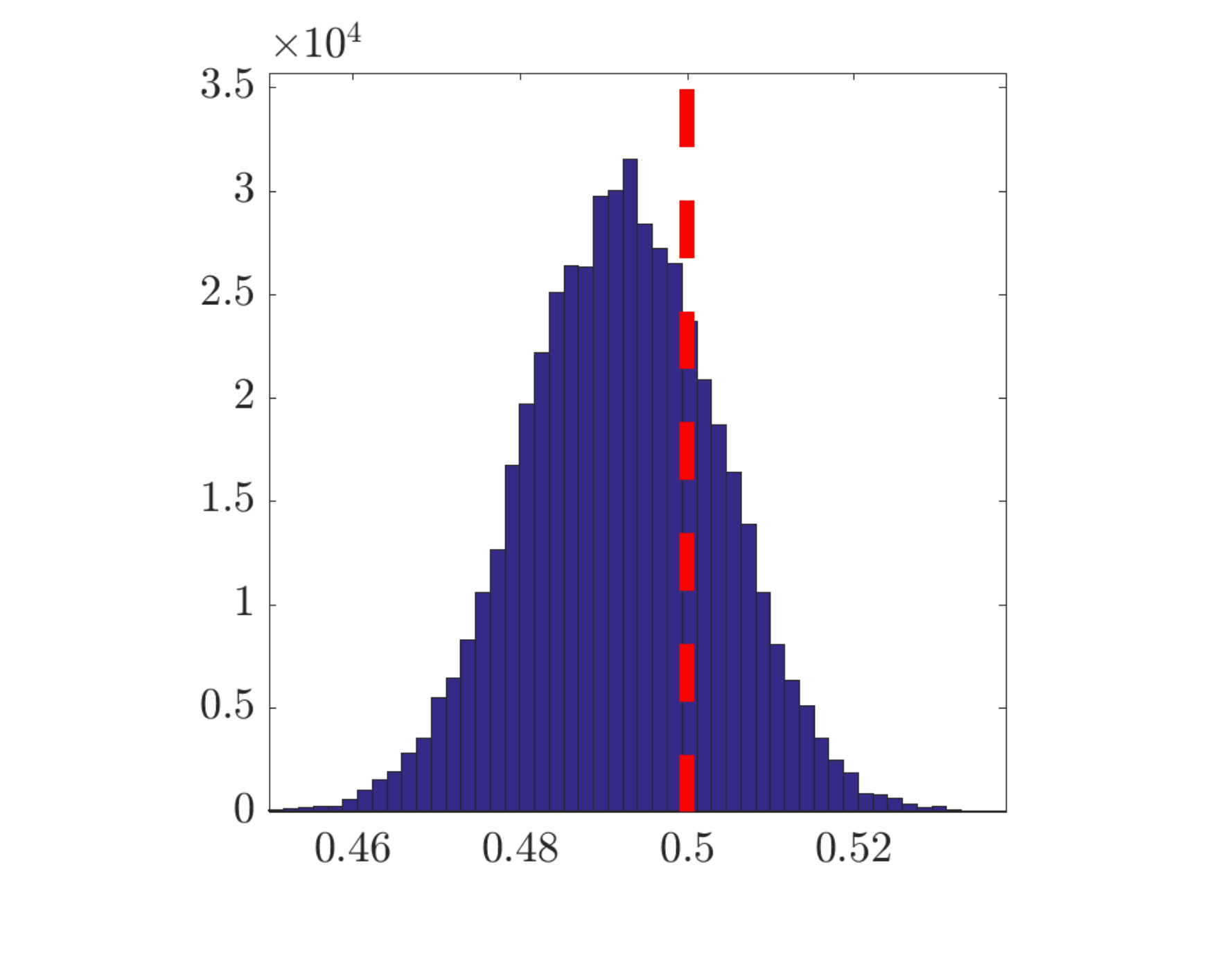}& & & &\\
  \includegraphics[width = .2\textwidth]{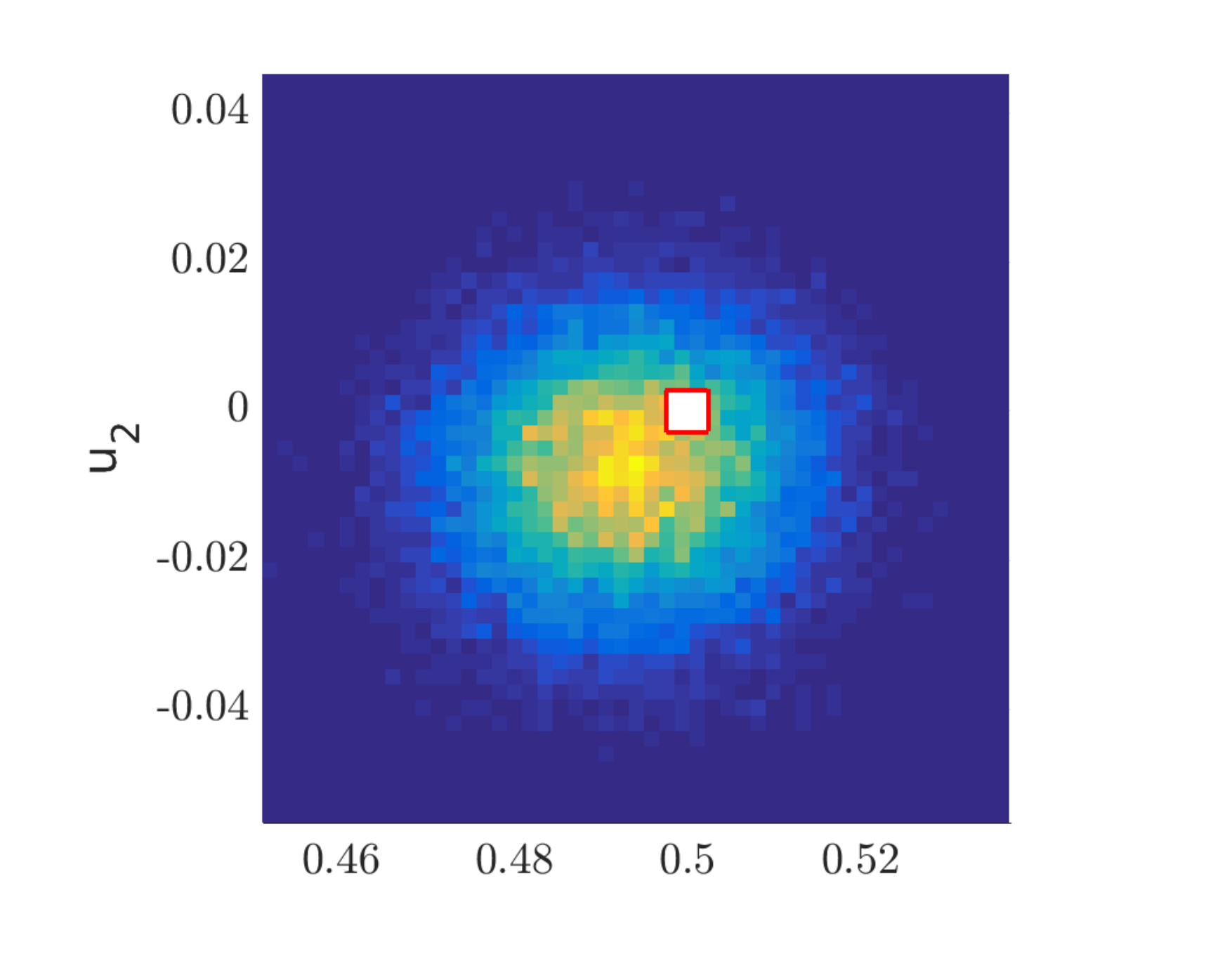}& 
  \includegraphics[width = .2\textwidth]{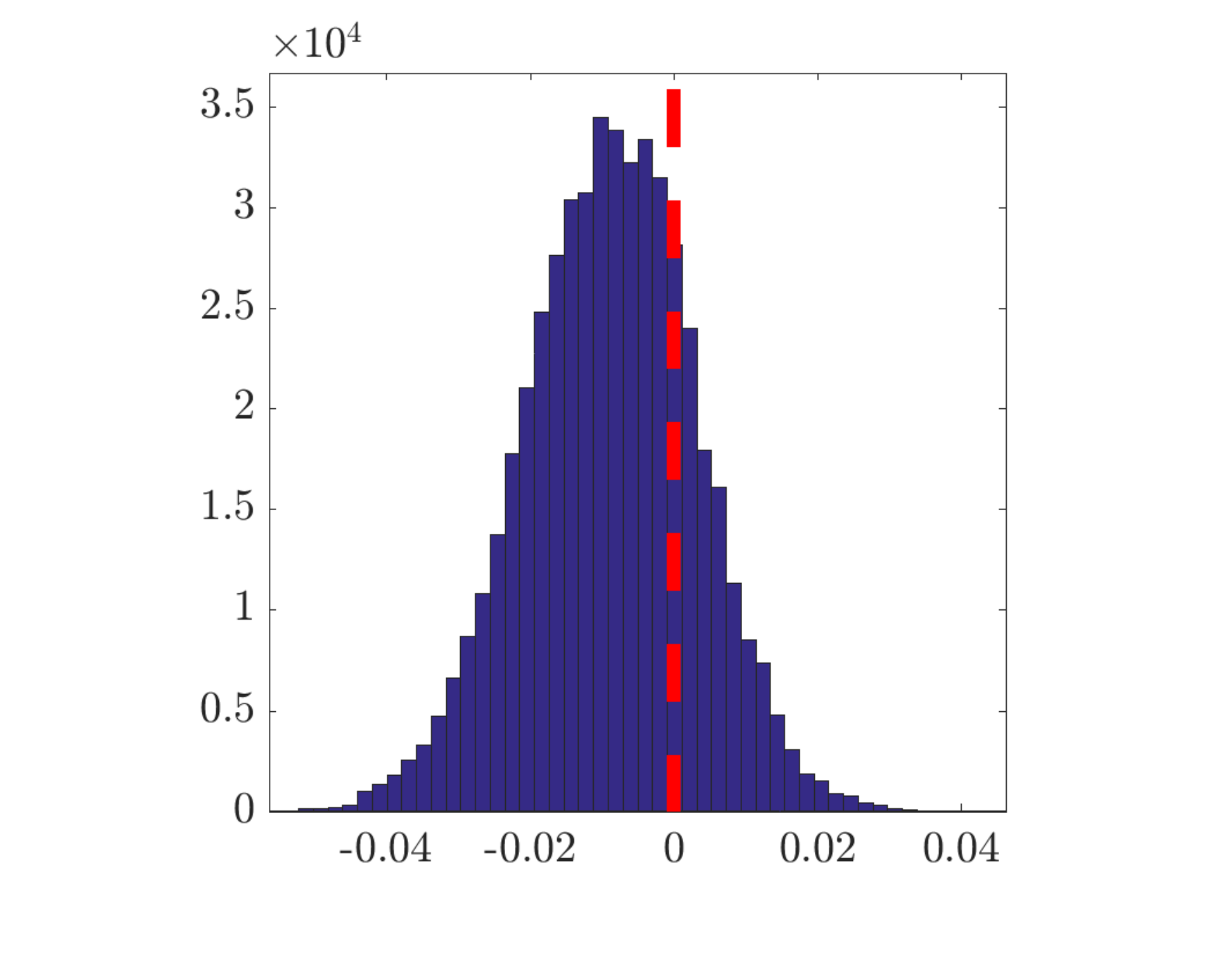}& & &\\ 
  \includegraphics[width = .2\textwidth]{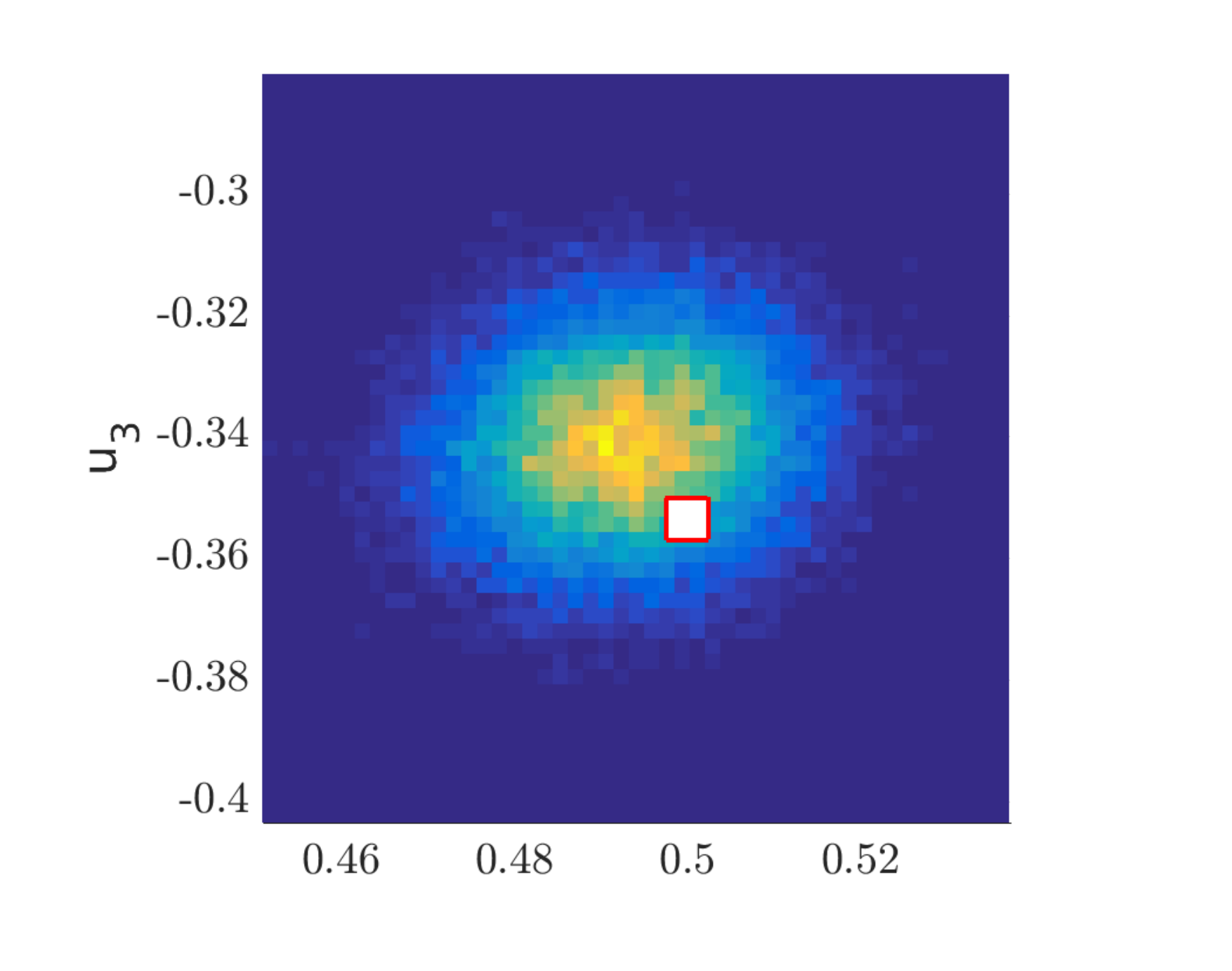}&
   \includegraphics[width = .2\textwidth]{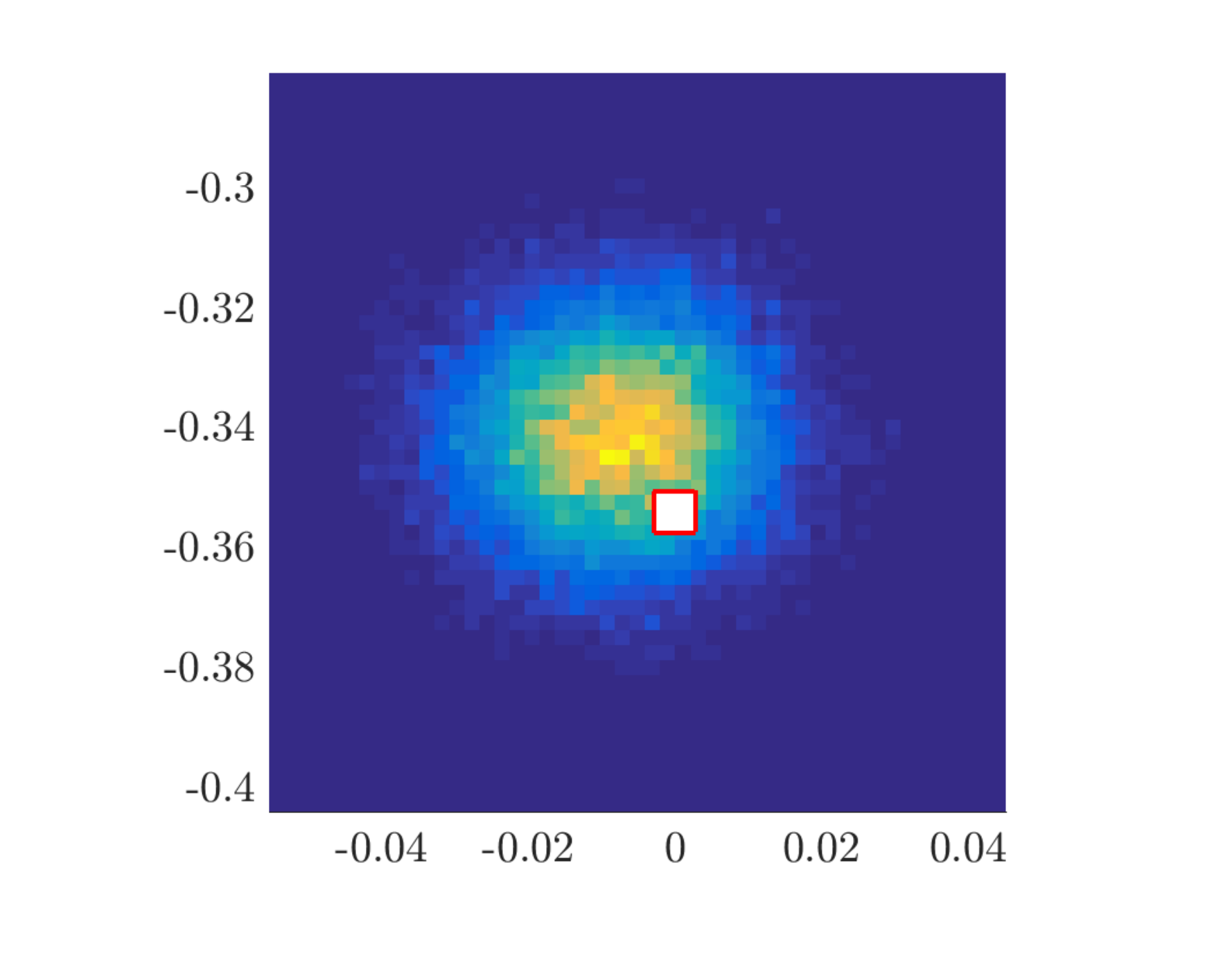}&
   \includegraphics[width = .2\textwidth]{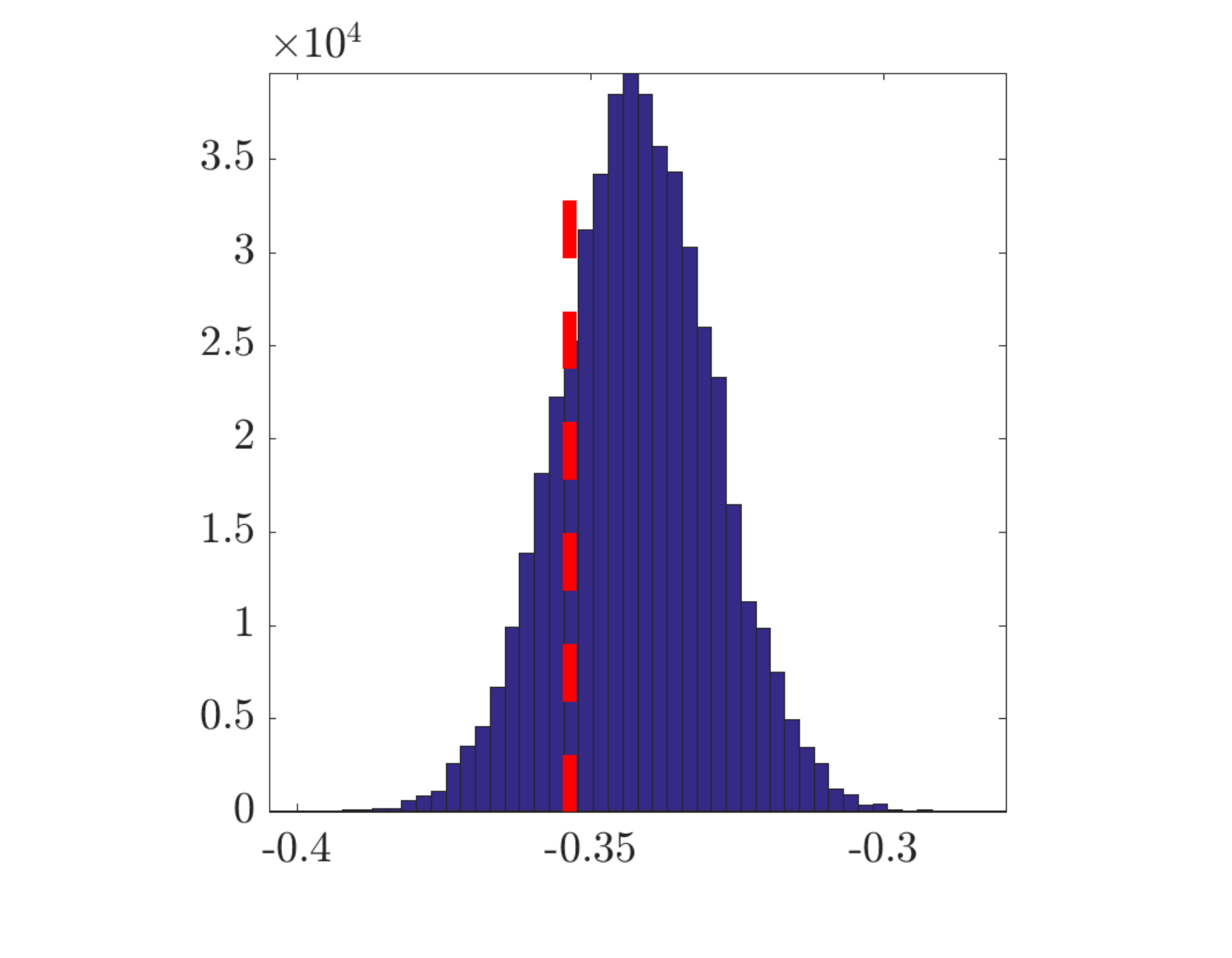}& &\\     
  \includegraphics[width = .2\textwidth]{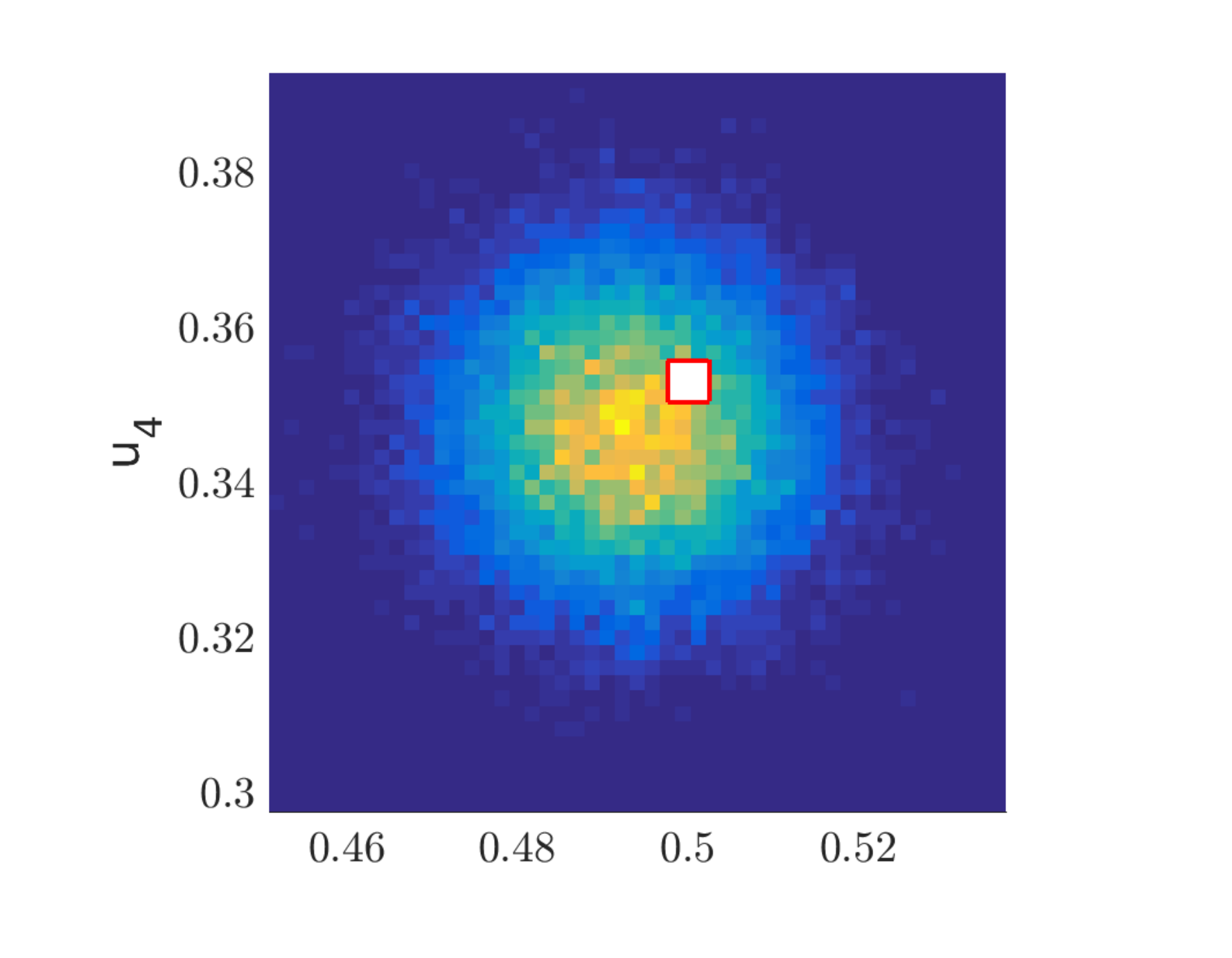}&
   \includegraphics[width = .2\textwidth]{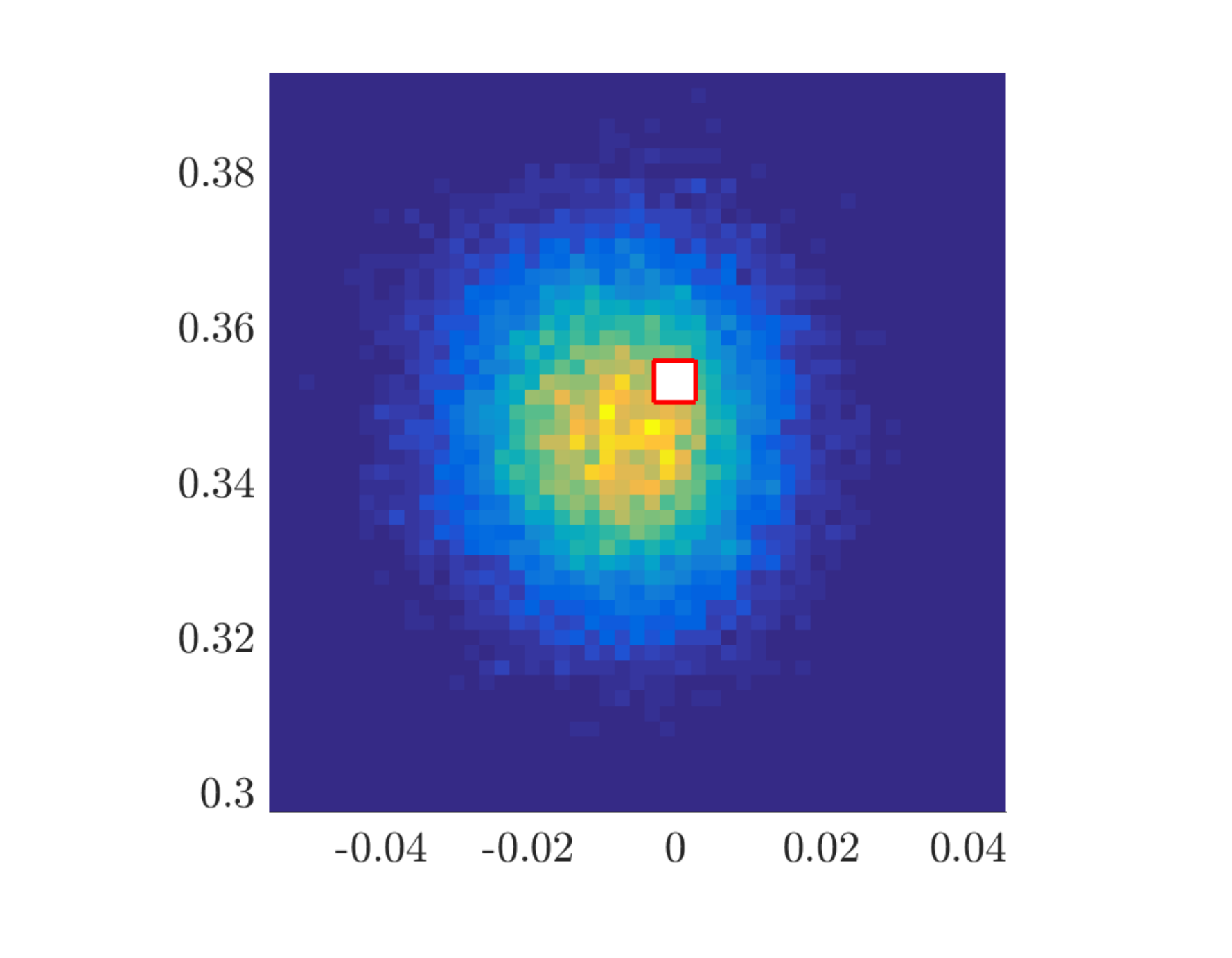}&
   \includegraphics[width = .2\textwidth]{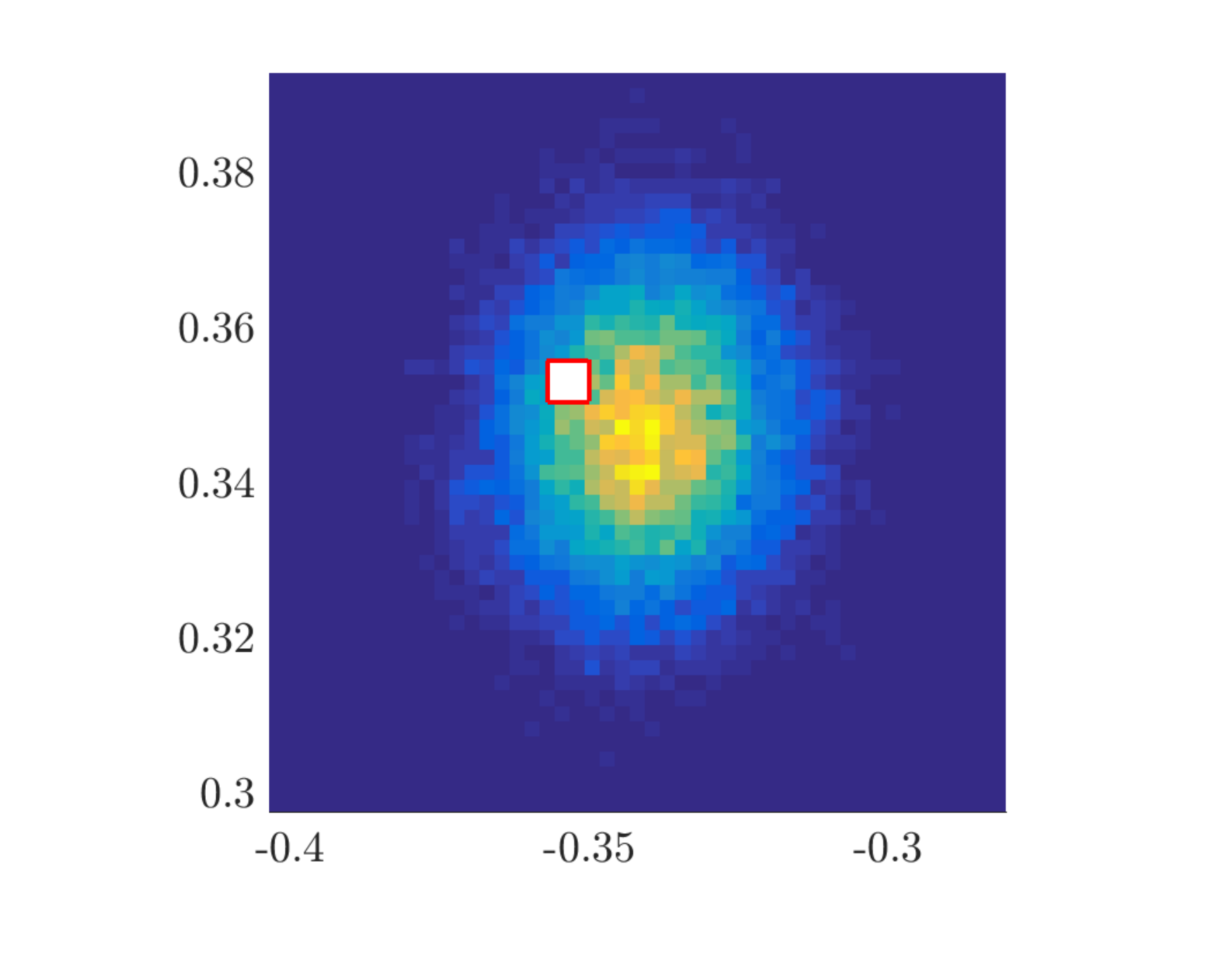}&
   \includegraphics[width = .2\textwidth]{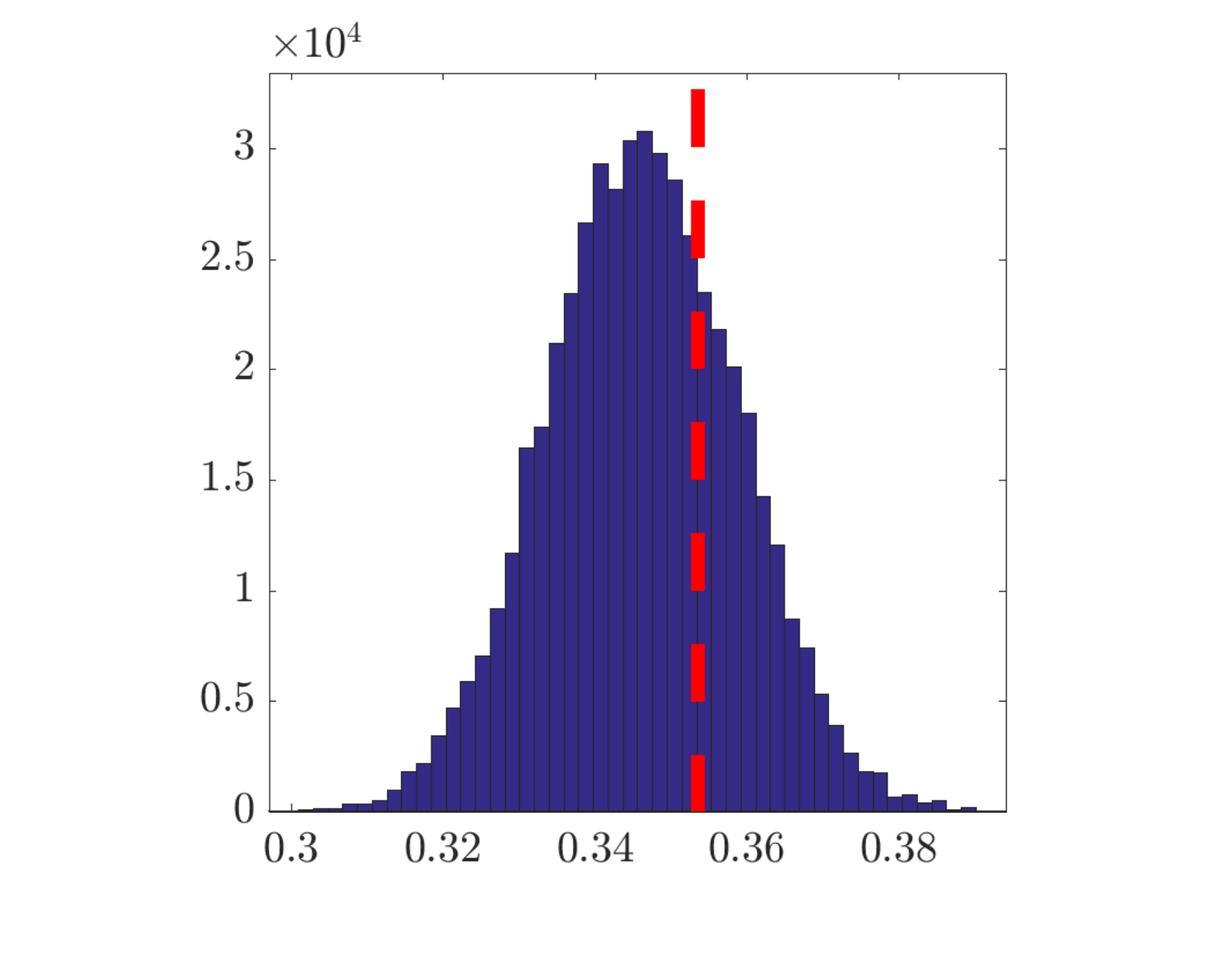}&\\      
   \includegraphics[width = .2\textwidth]{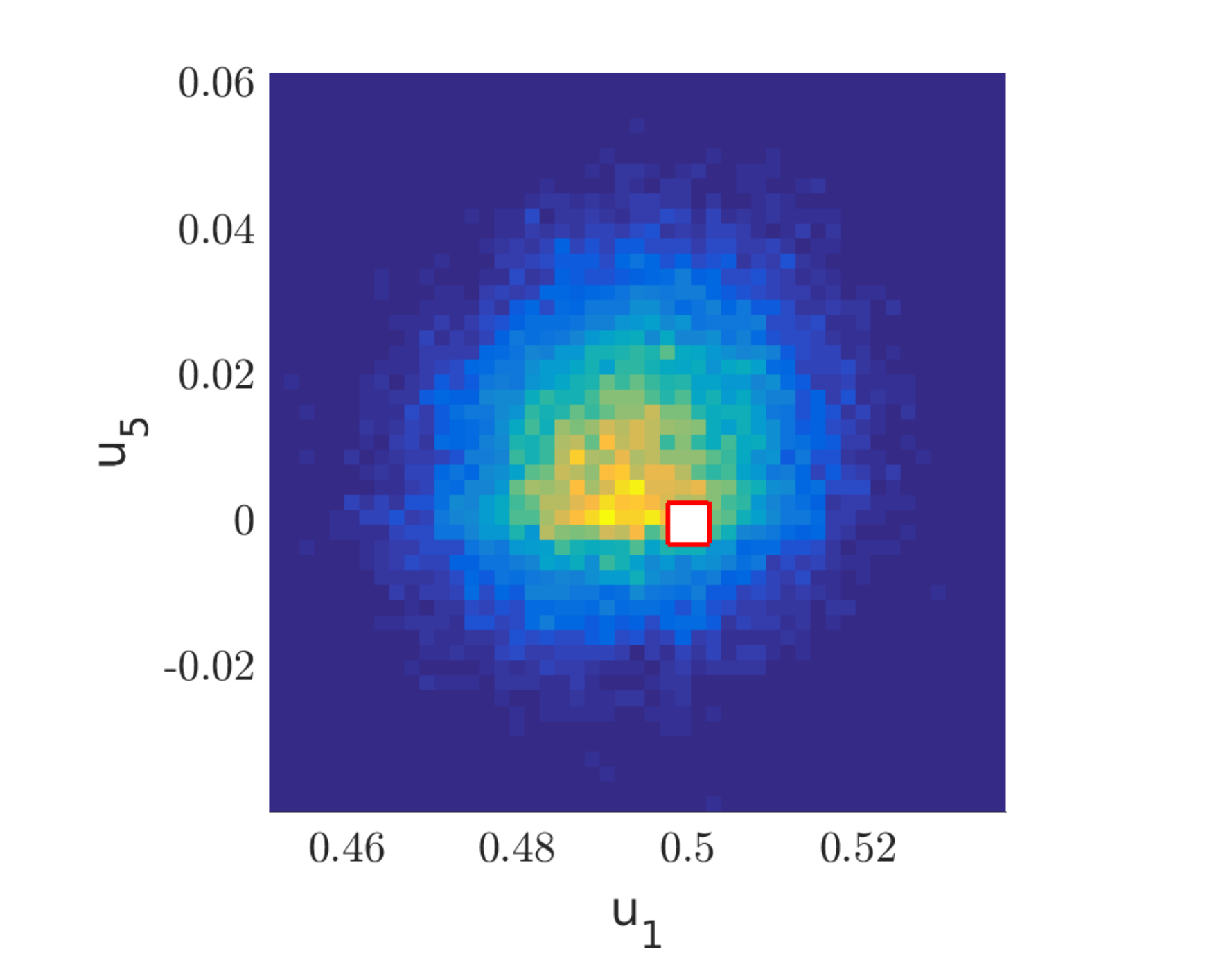}&
   \includegraphics[width = .2\textwidth]{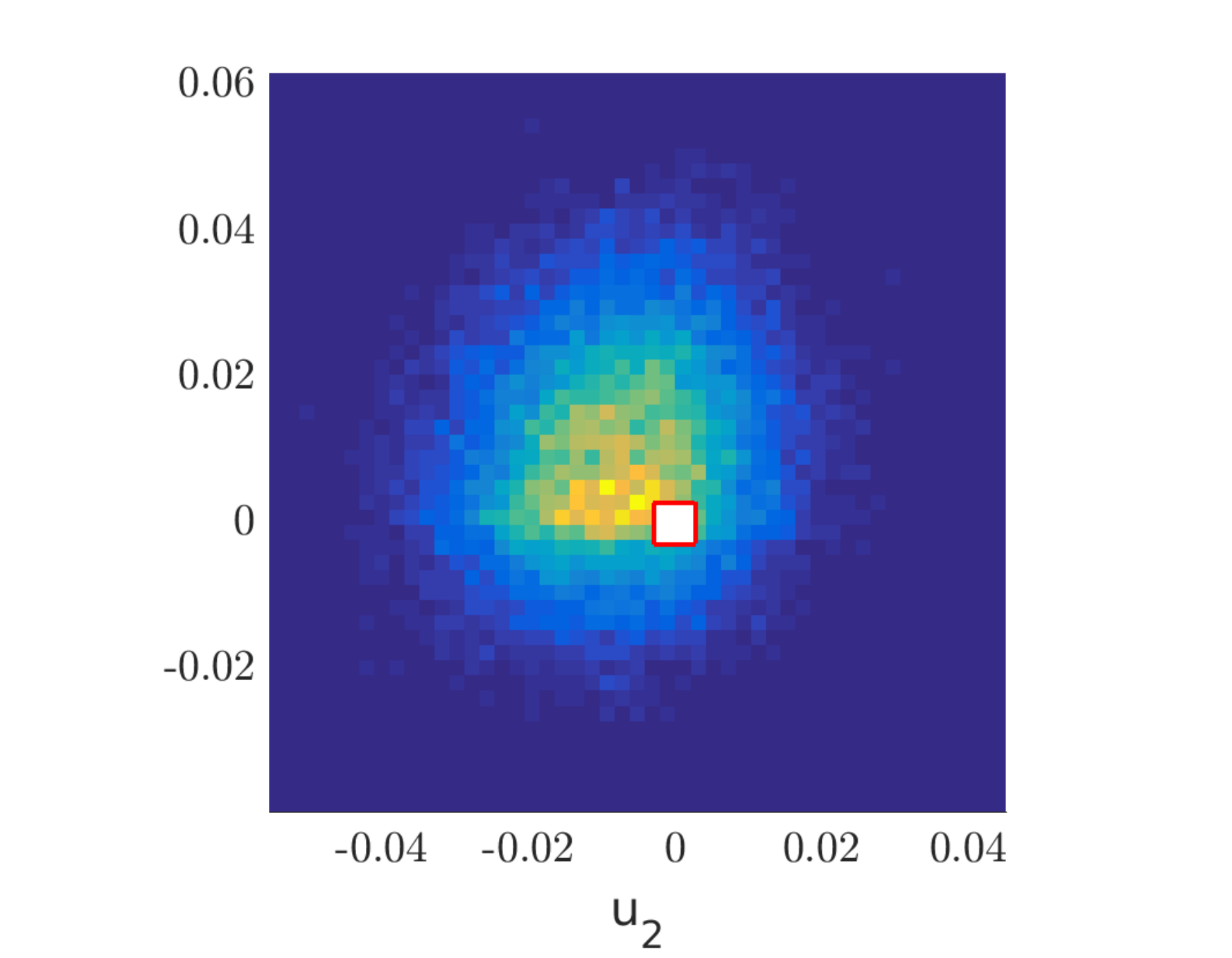}&
   \includegraphics[width = .2\textwidth]{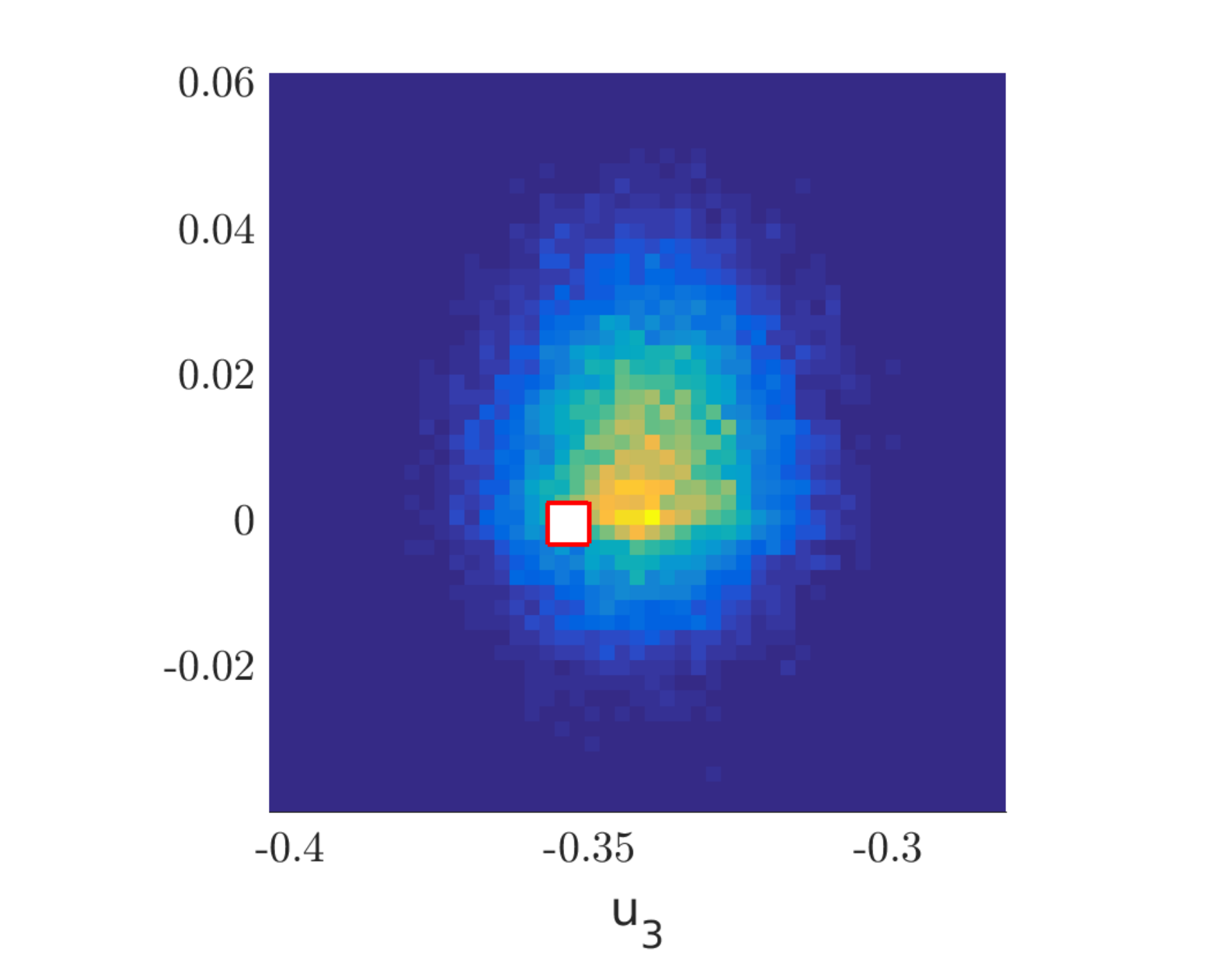}&
   \includegraphics[width = .2\textwidth]{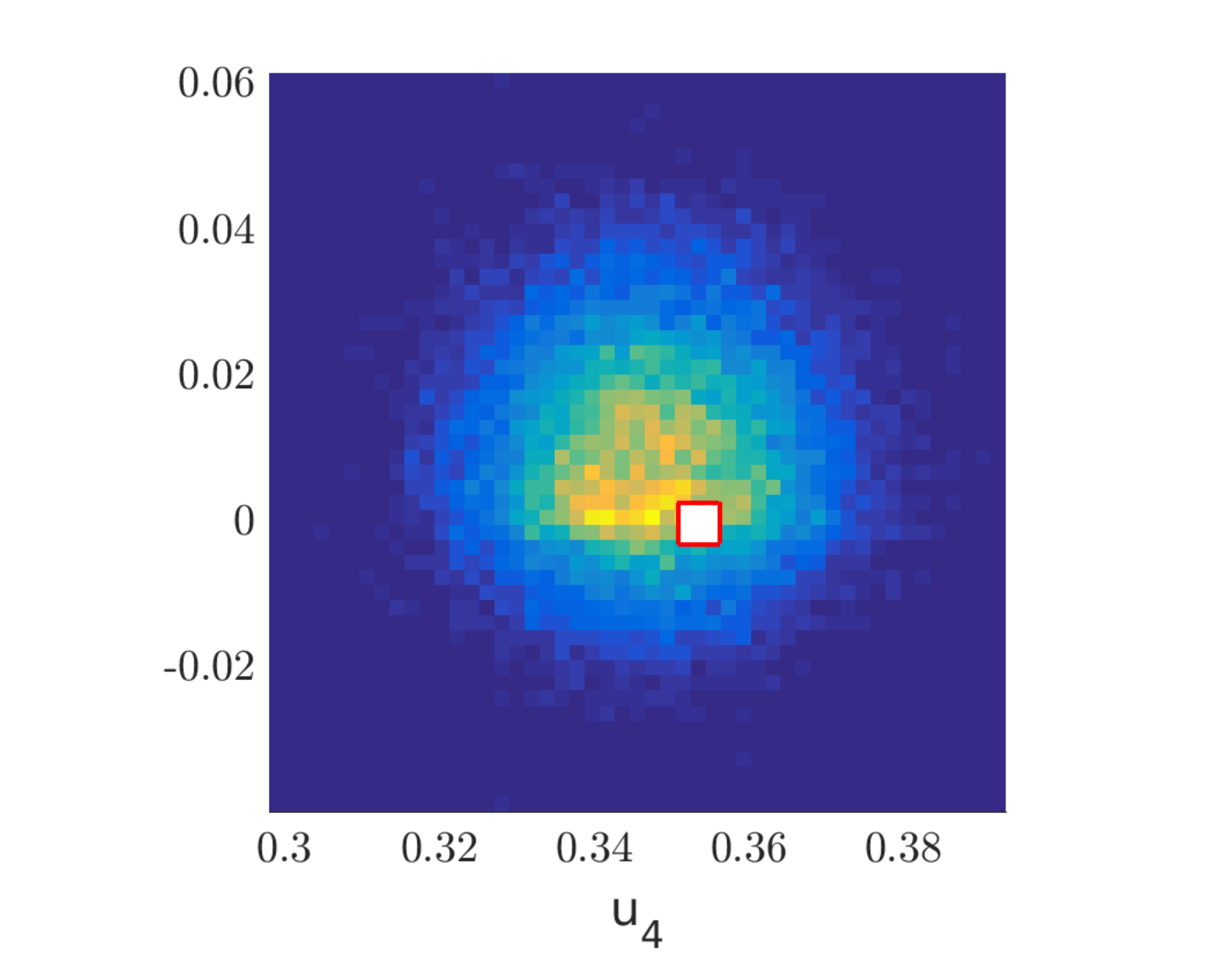}&
   \includegraphics[width = .2\textwidth]{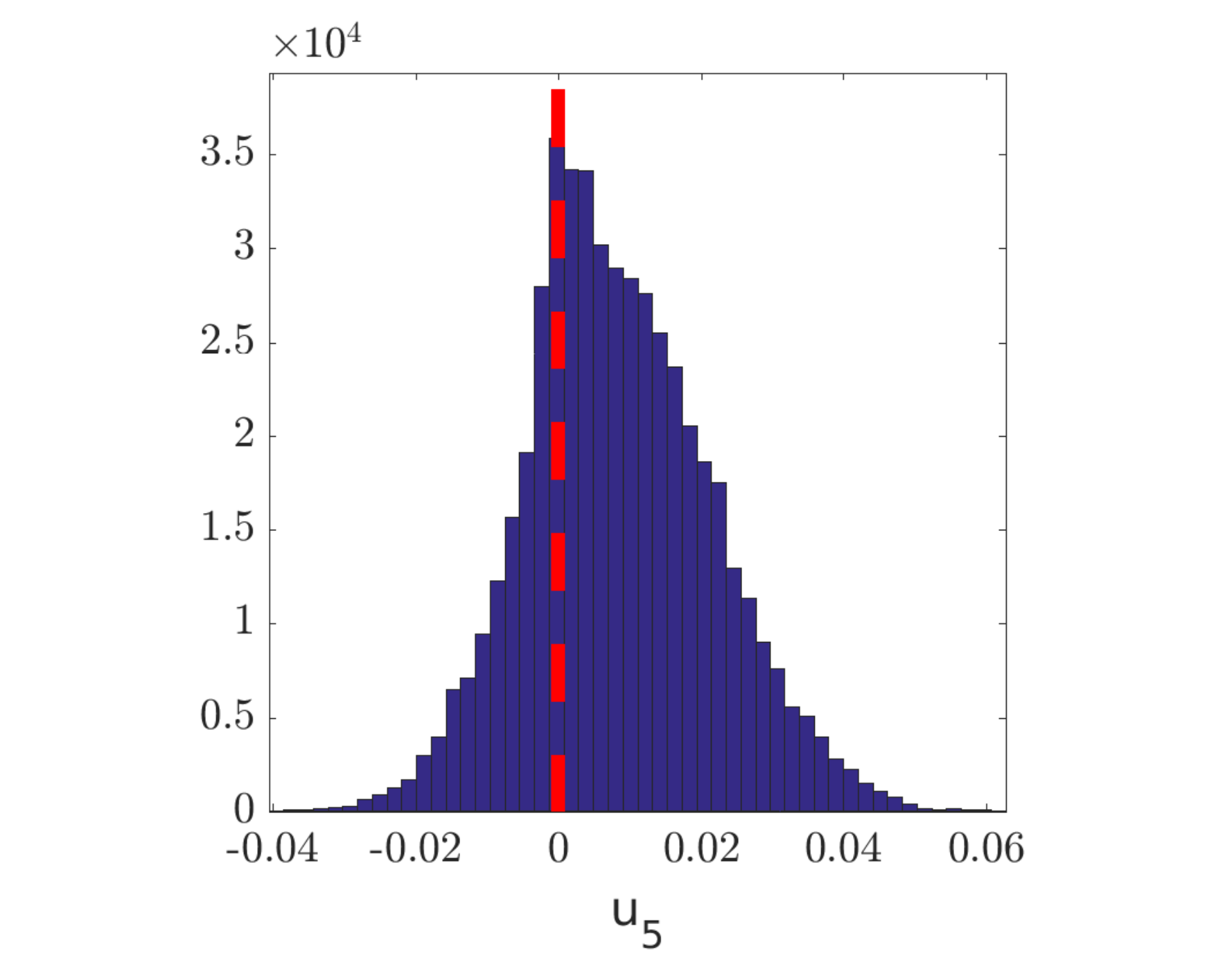}
\\                                                   
\end{tabular}  
\caption{Two dimensional histograms of the MCMC samples between the
  first five wavelet modes in the deconvolution problem of Example 3.
  These projections are regarded as two dimensional projections of the posterior measure.
  Components of the true solution $u_0$ are marked using the dashed red line on
  the 1D histograms and the white 
square in the 2D histograms.
}
  \label{fig:deconvolution-2D-posterior-hist}
\end{figure}

\subsubsection{Algorithm performance}\label{sec:RCAR-deconvolution-performance}

We now turn our attention to the performance of the lifted RCAR algorithm. In Figure~\ref{tab:chain-ESS-deconvolution}(a) we show statistics on the ESS of the different components of the Markov chain. Based on these 
results, in the case where $N=128$,  an independent sample was obtained  roughly every 
$500$ steps. While the ESS  deteriorated initially; as $N$ becomes larger
the ESS values  appeared to settle down for all components. The minimum, mean and the
maximum ESS values did not change significantly for $N \ge 16$. 

\begin{figure}[htp]
\centering
\raisebox{.22\textwidth}{a)}
  \raisebox{.13\textwidth}{\begin{tabular}{ |c | c | c| c |}
              \hline
              $N$ & $\min$ ESS & mean ESS & $\max$ ESS \\ \hline 
                 8 & 75&  98  & 171  \\ \hline 
                 16 &10 & 39  & 85\\ \hline 
                 32 & 17 & 41 & 116\\ \hline 
                 64 &  14& 39 & 118\\ \hline 
                 128 & 18& 41 & 91\\ \hline 
            \end{tabular}}
\qquad  \qquad
\raisebox{.22\textwidth}{b)}
\includegraphics[height=.225\textwidth]{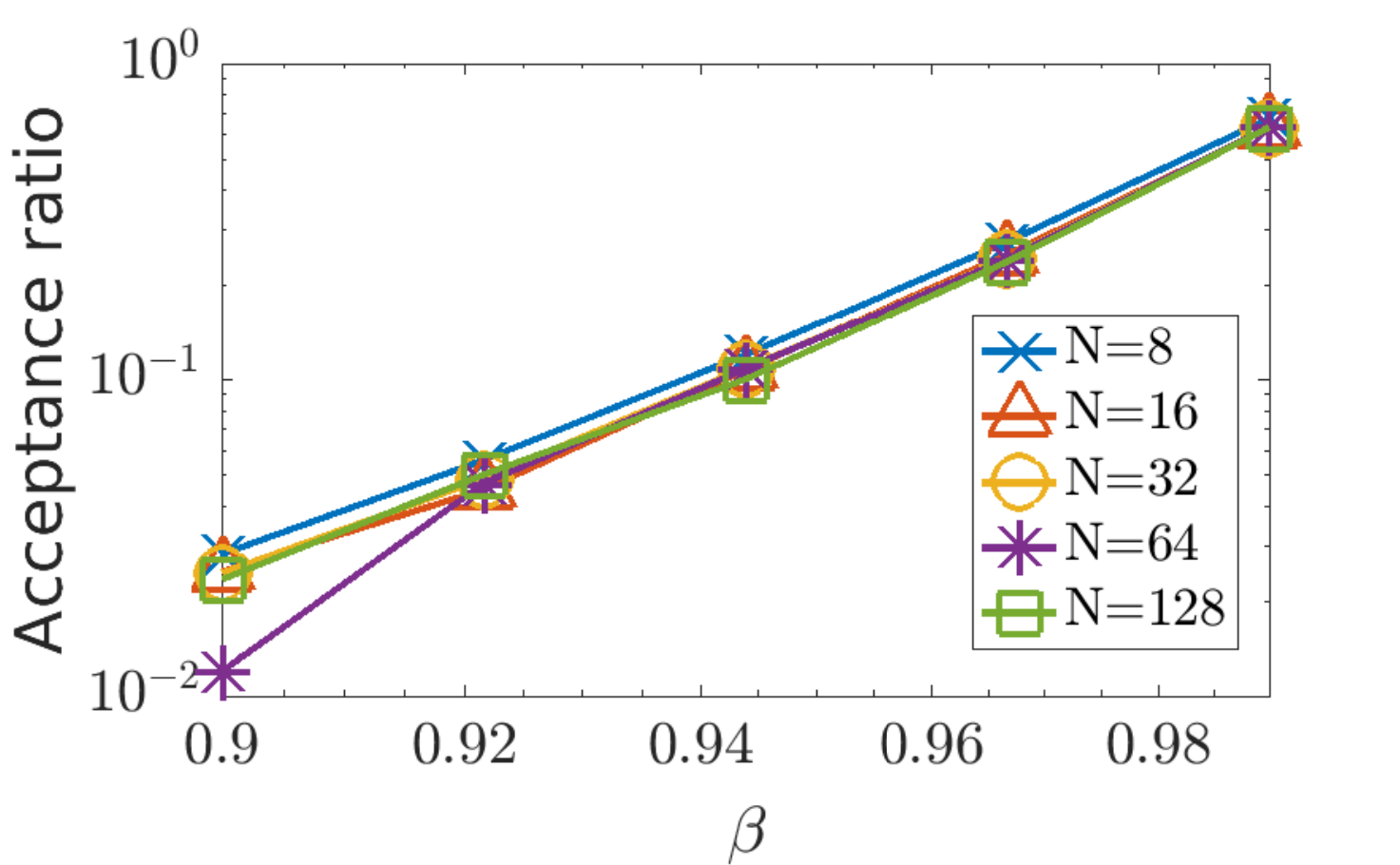}
   \caption{Performance indicators of the lifted RCAR algorithm for the deconvolution problem of 
Example 7 with $p=2/3$.
a) Approximated ESS of the chain components  for different values of $N$  and 
$\beta = 0.97$ per $10^4$ MCMC steps. b) Average acceptance ratio  for 
different values of $N$ and $\beta$ computed over $2 \times 10^5$
iterations after burnin and averaged over five restarts of the chain.}
\label{tab:chain-ESS-deconvolution}
\end{figure}

Since  the lifted RCAR algorithm is reversible in infinite dimensions we expect the 
acceptance ratio to remain bounded away from zero as $N$ becomes large.
In Figure~\ref{tab:chain-ESS-deconvolution}(b) we plot the average  acceptance ratio of
lifted RCAR for different values of  $\beta$. As before we fixed $p =2/3$ and 
$\lambda = 1$. We used a burnin  of $5 \times 10^4$ steps and 
computed the average acceptance rates over $2 \times 10^5$ steps with five restarts and averaged the acceptance ratios over the five trials. The acceptance ratios 
remained more or less consistent as a function of $N$ which is
in line with the reversibility of the infinite-dimensional limit of
the algorithm.

\subsubsection{Effect of  hyperparameters $p$ and $\lambda$}
We now study the effect of the hyperparameters $p$ and $\lambda$
on the posterior as well as performance of the lifted RCAR algorithm.
In all of the examples below we fixed $N = 32$ and $\beta = 0.97$.

First,
we considered different values of $p= 1, 4/5, 3/5, 2/5, 1/5$. Figure~\ref{fig:deconvolution-p-study-wavelet} depicts the posterior mean and standard deviation of the 
first eight wavelet modes of the Markov chain for different $p$.
In Figure~\ref{fig:deconvolution-p-study-wavelet}(a) we observe that for all values of $p$ the 
posterior mean was consistent and varied only slightly. While the posterior mean
appeared to be insensitive to choice of $p$ the posterior standard deviation 
seems to be quite sensitive to  $p$. 
This is evident in Figure~\ref{fig:deconvolution-p-study-wavelet}(b)
where  the standard deviation of the higher modes reduced with $p$. 

Next, we considered the effect of $p$ on RCAR's performance. Figure~\ref{tab:deconvolution-p-study}(b) shows that the ESS dropped when $p$ was too small or
too big. We also see that for fixed $\beta$ the acceptance ratio dropped when $p$ becomes larger.
Similarly, the acceptance ratio increased as $p$ was reduced. Overall, we conclude that
the optimal choice of the step size  $\beta$ is sensitive to the choice of $p$.

\begin{figure}[htp]
  \centering
  \includegraphics[width= .48 \textwidth]{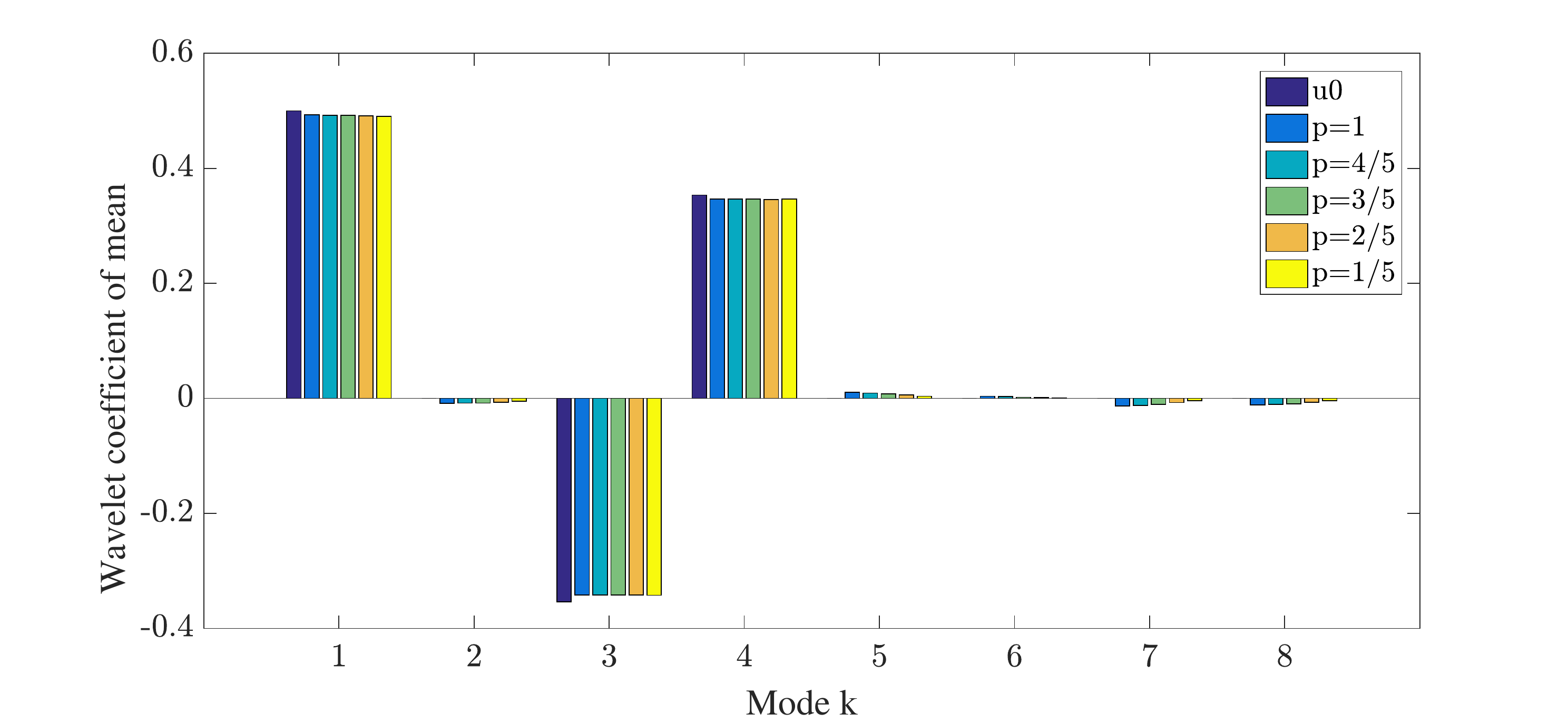}
\quad 
  \includegraphics[width= .48 \textwidth]{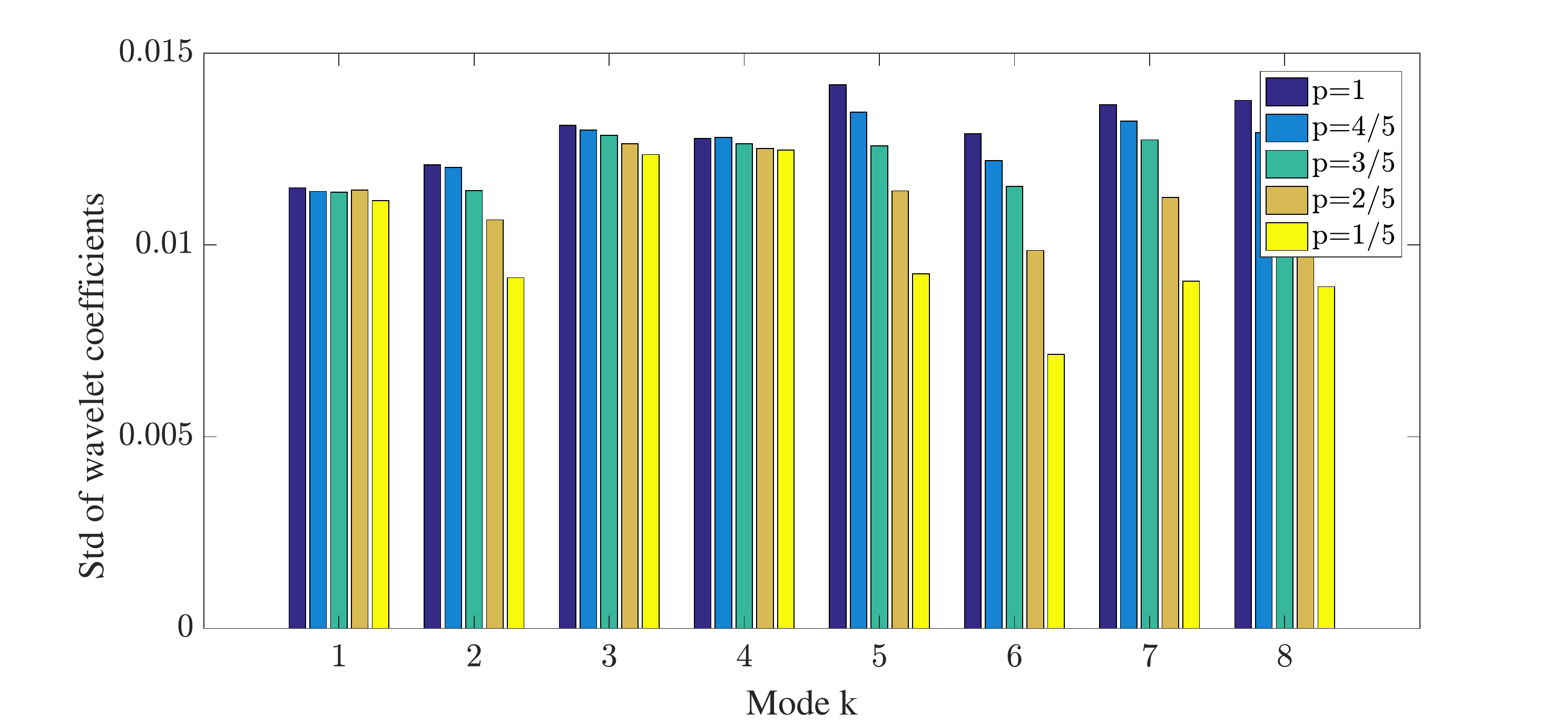}
  \caption{Posterior mean and standard deviation  of  the first eight wavelet modes in the deconvolution problem for different values of the parameter $p$.}
  \label{fig:deconvolution-p-study-wavelet}
\end{figure}

\begin{figure}[htp]
\centering
\raisebox{.22\textwidth}{a)}
\includegraphics[height=.25\textwidth]{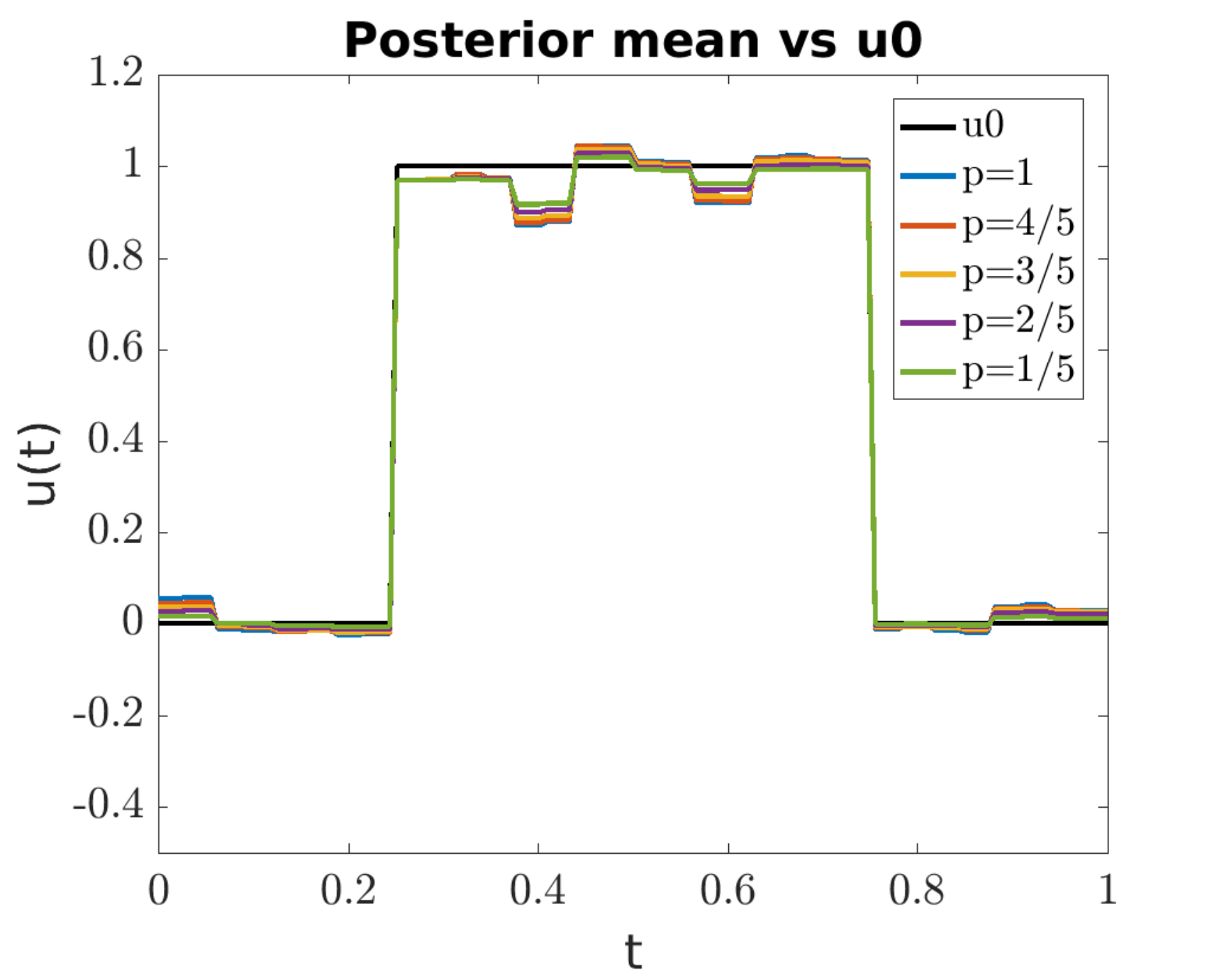}
\raisebox{.22\textwidth}{b)}
  \raisebox{.13\textwidth}{\begin{tabular}{ |c | c | c| c | c| }
              \hline
              $p$ & $\min$ ESS & mean ESS & $\max$ ESS & mean $a(\cdot,\cdot)$ \\ \hline 
                 1 & 14 & 34   & 84  & 0.15\\ \hline 
                 4/5 & 15 & 44  & 103 & 0.22\\ \hline 
                 3/5 & 20 & 54 & 107 & 0.31 \\\hline 
                 2/5 &  16& 68 & 113 & 0.45\\ \hline 
                 1/5 & 13 & 79 & 120& 0.62\\ \hline 
            \end{tabular}}
\qquad  \qquad
   \caption{a) Posterior mean of the deconvolution problem of Example 7 for different choices of 
the hyperparameter $p$. b) Certain statistics on the 
ESS and average acceptance ratio of the components of the Markov chains for different values of $p$ as 
well as the average acceptance ratio computed over $5 \times 10^5$ samples
and for $\beta= 0.97$. Smaller values of $p$ resulted in higher acceptance rate.
Minimum ESS  deteriorated as $p$ became too large or too small indicating that
$\beta$ should be tuned for different values of $p$ to achieve
optimal performance.}
\label{tab:deconvolution-p-study}

\end{figure}

Finally we consider the  $\lambda$ parameter. We fixed $p=2/3$ and varied 
$\lambda$ between $1/4$ and $4$. Recall that $\lambda$ controls the global variance of the 
wavelet modes of the solution.
 Figure~\ref{fig:deconvolution-lambda-study-wavelet}(a)
 shows empirical posterior mean and standard deviation of the first eight wavelet modes for
 different values of $\lambda$. The posterior mean was
somewhat sensitive to the choice of $\lambda$ specially for
the higher wavelet modes. This sensitivity to $\lambda$ is  evident in
Figure~\ref{tab:deconvolution-lambda-study}(a) where the posterior mean seems to have
higher variation for larger $\lambda$. We can also see the effect of
$\lambda$ in the posterior standard deviations. 
Figure~\ref{fig:deconvolution-lambda-study-wavelet}(b) shows that increasing $\lambda$ resulted
in increased  posterior variance which is expected considering that 
posterior variance is closely related to that of the prior and the measurement noise. Finally, 
in Figure~\ref{tab:deconvolution-lambda-study}(b) we present the average acceptance ratio 
and ESS of the Markov chains for different values of $\lambda$. All of these simulations 
shared the same value of $\beta = 0.97$. We observe that both 
the average acceptance ratio and ESS dropped as $\lambda$ was increased.

\begin{figure}[htp]
\centering
\raisebox{.22\textwidth}{a)}
\includegraphics[height=.25\textwidth]{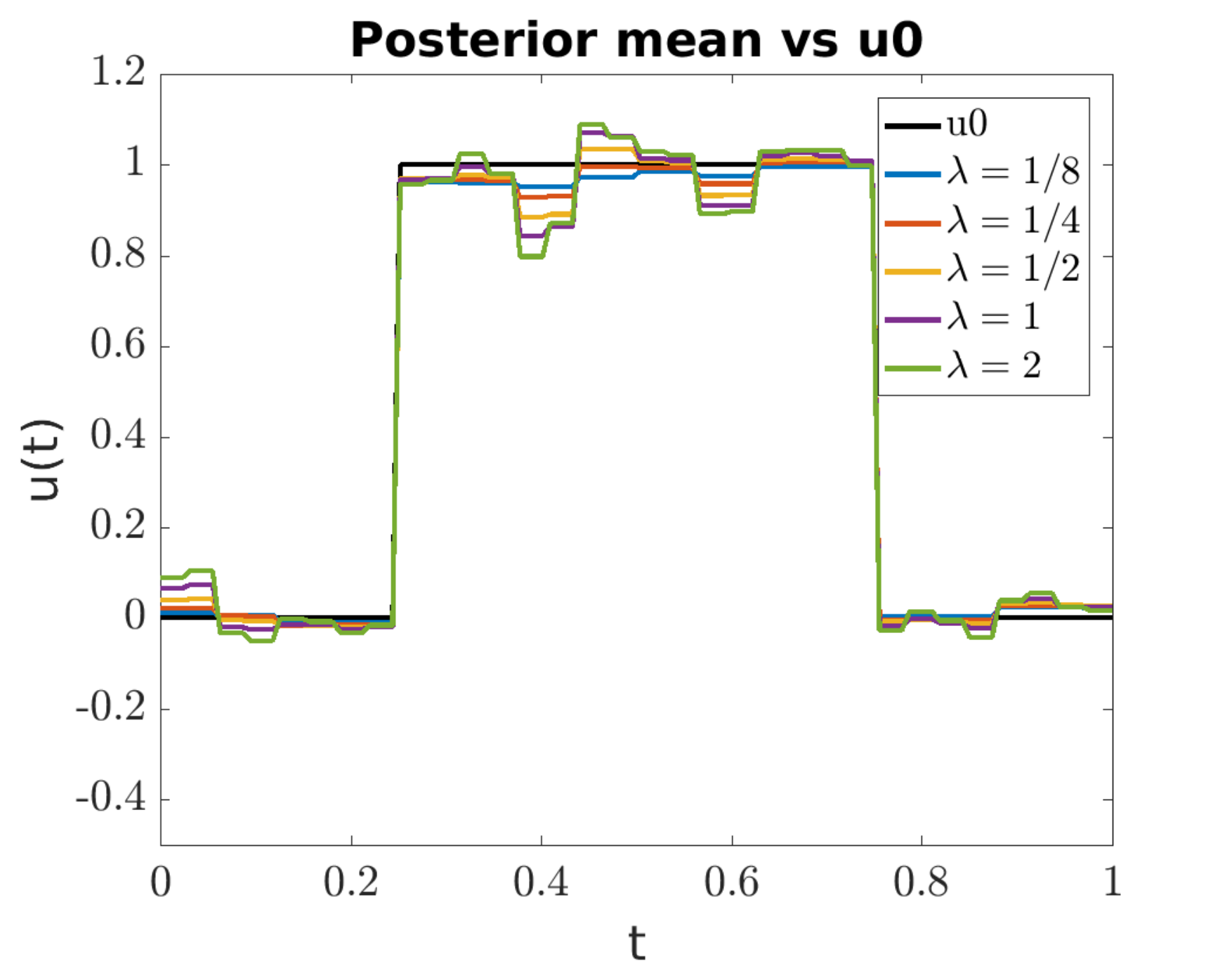}
\raisebox{.22\textwidth}{b)}
  \raisebox{.13\textwidth}{\begin{tabular}{ |c | c | c| c | c| }
              \hline
              $\lambda$ & $\min$ ESS & mean ESS & $\max$ ESS &mean $a(\cdot,\cdot)$ \\ \hline  
                 1/4 & 29 & 79  & 183& 0.47\\ \hline 
                 1/2 & 16 & 65 & 147& 0.37 \\\hline 
                 1 &  20& 50 & 104 & 0.27\\ \hline 
                 2 & 19 & 35 & 69 & 0.18\\ \hline
                 4 & 10      & 20     & 38     & 0.12 \\ \hline
            \end{tabular}}
\qquad  \qquad
   \caption{a) Posterior mean of the deconvolution problem of Example 7 for different choices of 
the hyperparameter $\lambda$. For larger values of $\lambda$ the posterior mean tends to 
be less regular. b) Certain statistics on the 
ESS of the components of the Markov chain for different values of $\lambda$ as 
well as the average acceptance ratio computed over $5 \times 10^5$ samples.
Both acceptance ratio and ESS  deteriorated with larger values of $\lambda$.}
\label{tab:deconvolution-lambda-study}

\end{figure}

\begin{figure}[htp]
  \centering
  \includegraphics[width= .48 \textwidth]{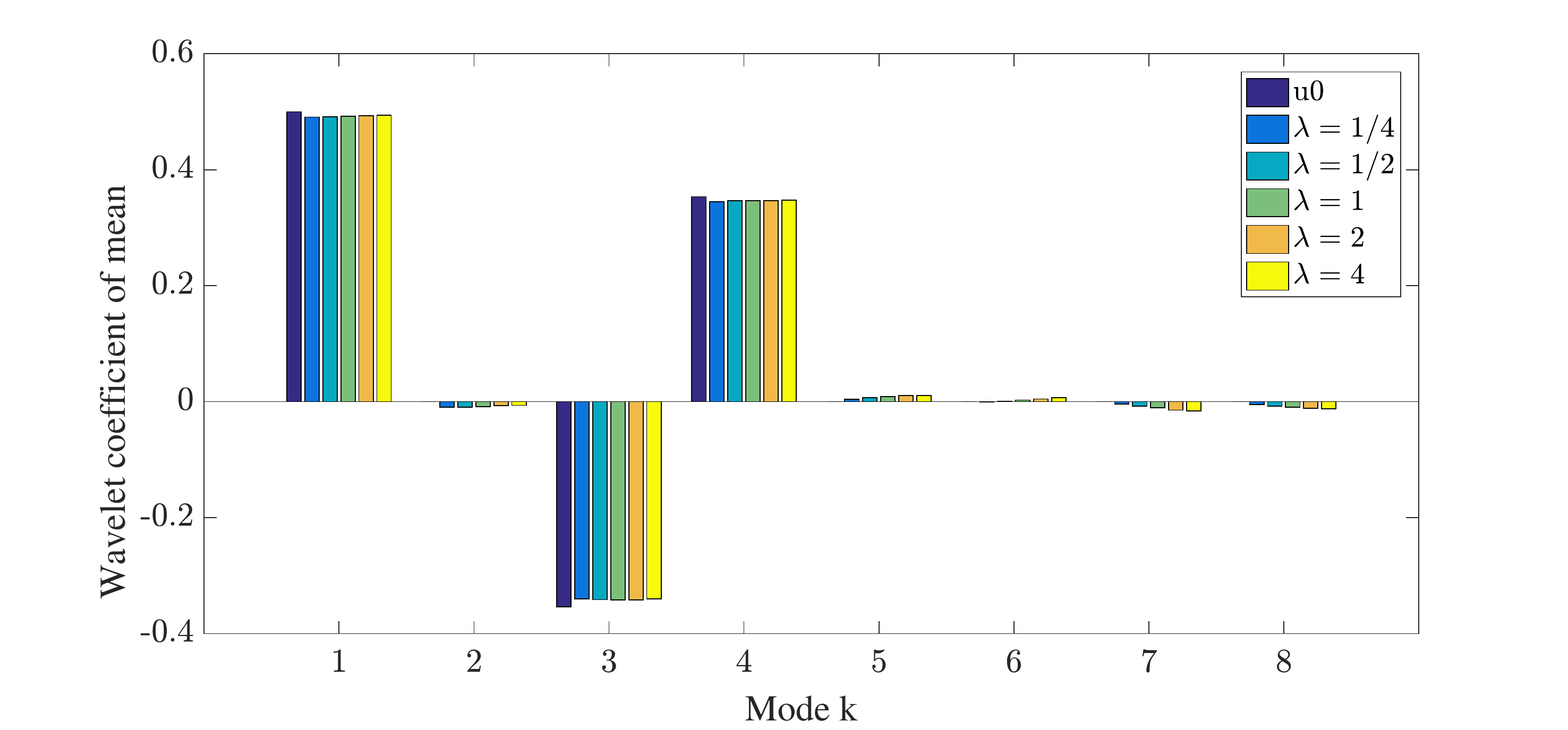}
  \quad
  \includegraphics[width= .48 \textwidth]{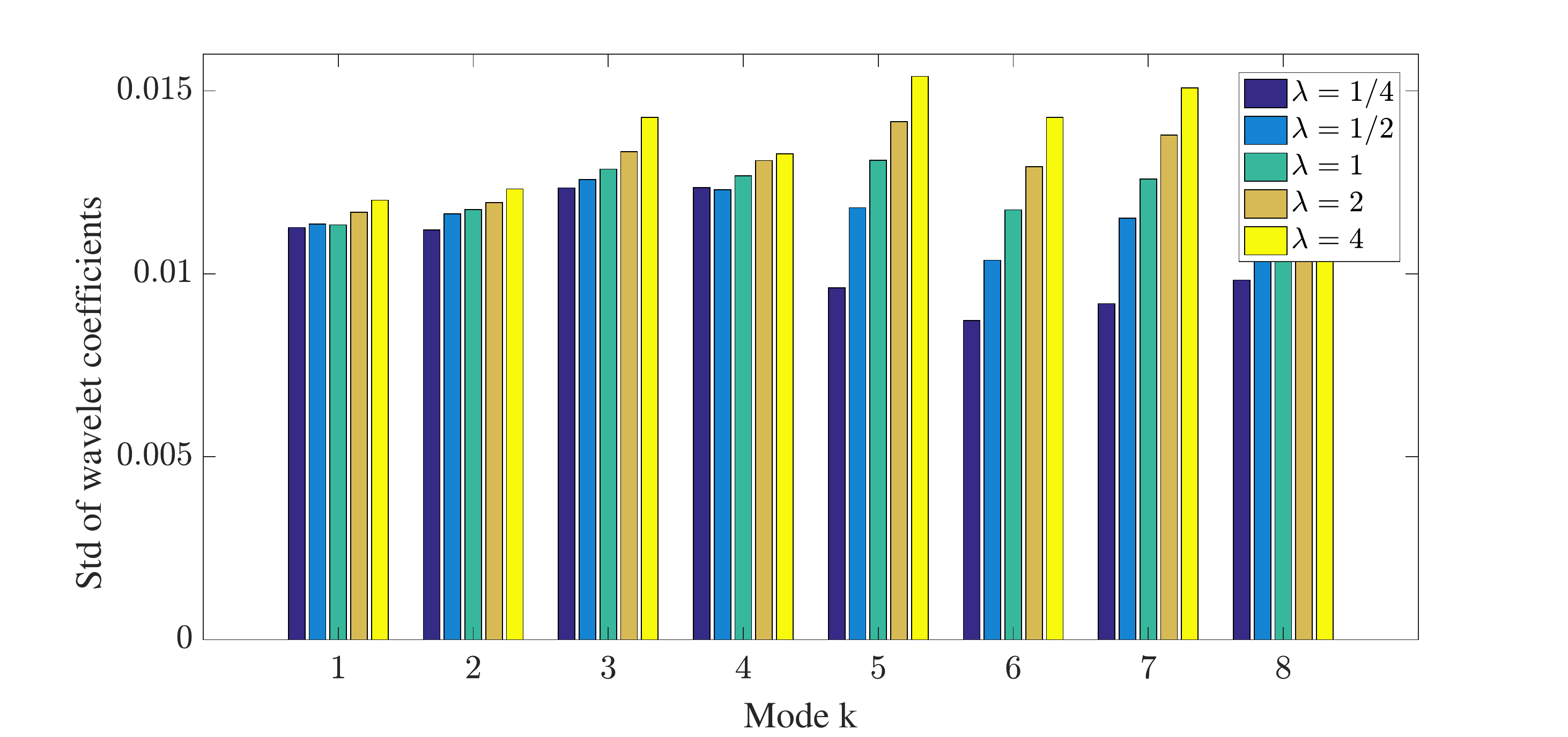}
  \caption{Posterior mean and standard deviation  of  the first eight wavelet modes in the deconvolution problem for different values of the parameter $\lambda$.}
  \label{fig:deconvolution-lambda-study-wavelet}
\end{figure}

\subsubsection{Effect of kernel width $\varepsilon$ and correlations in the data}\label{sec:kernel-width-study}

For our final set of experiments we considered the effect of the kernel width $\varepsilon$
on the quality of the posterior mean as a pointwise approximation to $u_0$.
Intuitively, larger values of $\varepsilon$ result in more smoothing
that in turn results in more correlated measurements  and overall less
information in the data
as evident in Figure~\ref{fig:deconvolution-data-kernel-width-study}. Throughout these experiments we used $N=32$ and $\beta = 0.97$.
 In Figure~\ref{tab:deconvolution-kw-study-mean-ESS}(a)
 we show the posterior means for different $\varepsilon$ and compare them to $u_0$.
 We observe more deviation from $u_0$
for larger values of $\varepsilon$ which is expected following our intuition that
$y$ is less informative when the forward map is more smoothing. This is also
evident in the average acceptance ratios reported in Figure~\ref{tab:deconvolution-kw-study-mean-ESS}(b). Since $y$ was less informative for larger $\varepsilon$, the posterior was more
dominated by the
prior rather than the likelihood. Since the lifted RCAR algorithm is prior reversible
 we expect it to perform better with larger $\varepsilon$.
More evidence of this phenomenon can be seen in Figure~\ref{fig:deconvolution-kw-study-wavelet}(b)
where 
larger values of $\varepsilon$  resulted in larger posterior uncertainty
in the wavelet modes indicating that the likelihood is less dominant.
We highlight that regardless of the effect of $\varepsilon$ on the posterior
uncertainties the posterior mean remained a qualitatively good approximator of $u_0$.

\begin{figure}[htp]
  \centering
  \includegraphics[height=.24\textwidth]{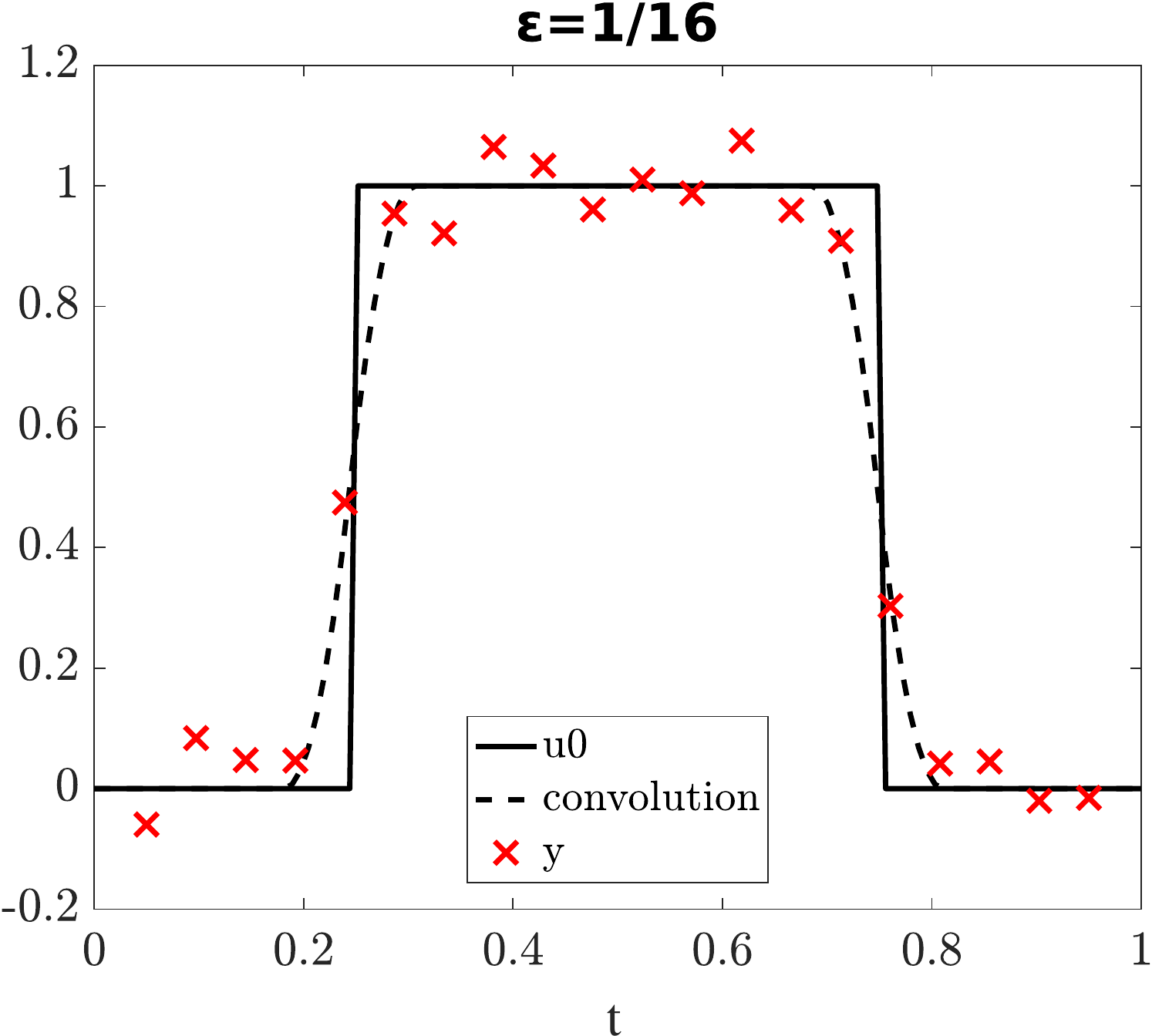}\qquad
  \includegraphics[height=.24\textwidth]{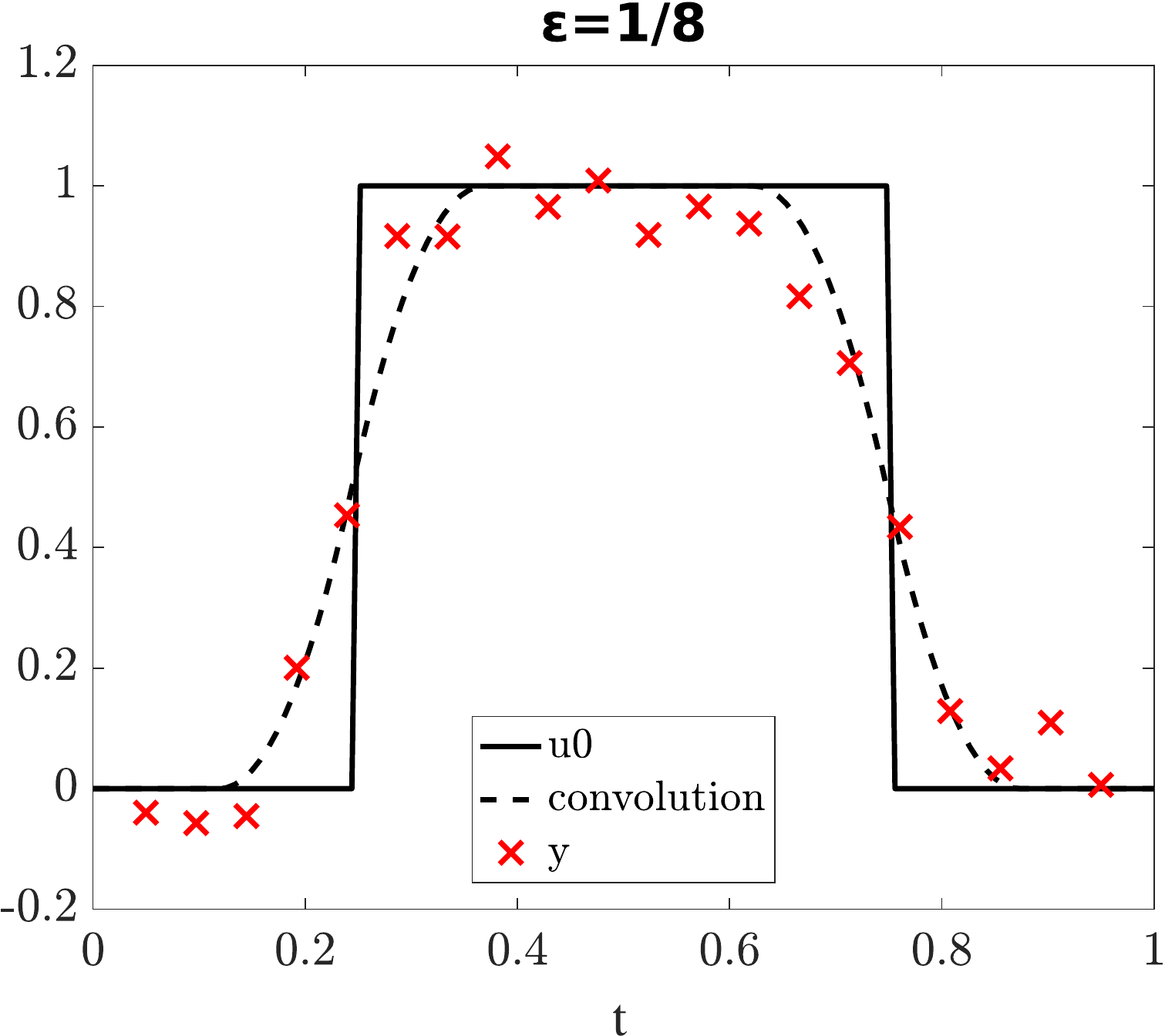}\qquad
  \includegraphics[height=.24\textwidth]{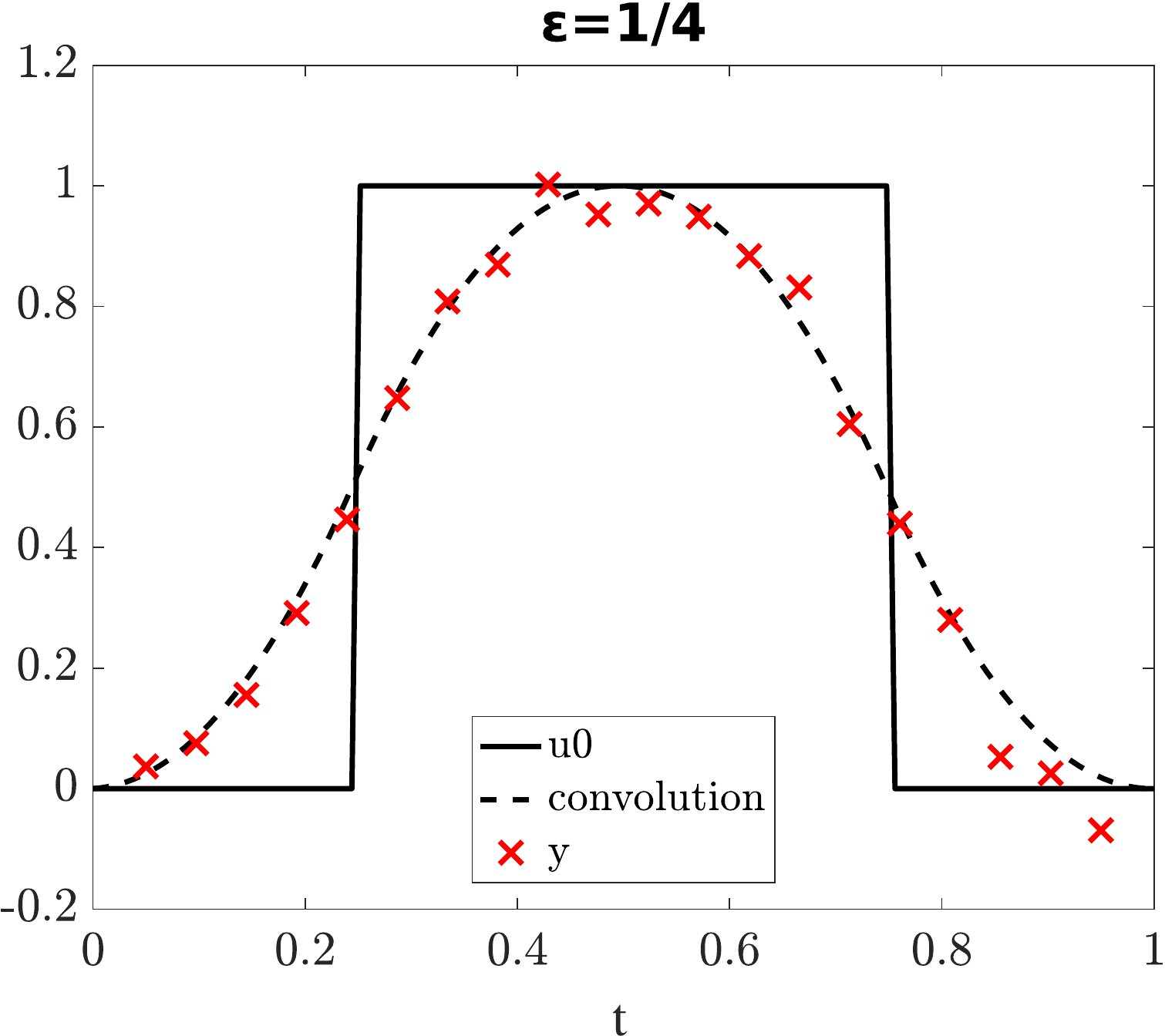}
  \caption{True function $u_0$ along with its convolved version with
    various kernel widths $\varepsilon = 1/16, 1/8, 1/4$ and a realization of
    the noisy pointwise measurements $y$
    for each kernel width. Larger values of $\varepsilon$ result in more smoothing
  which in turn results in more correlation in $y$.}
  \label{fig:deconvolution-data-kernel-width-study}
\end{figure}

\begin{figure}[htp]
\centering
\raisebox{.22\textwidth}{a)}
\includegraphics[height=.24\textwidth]{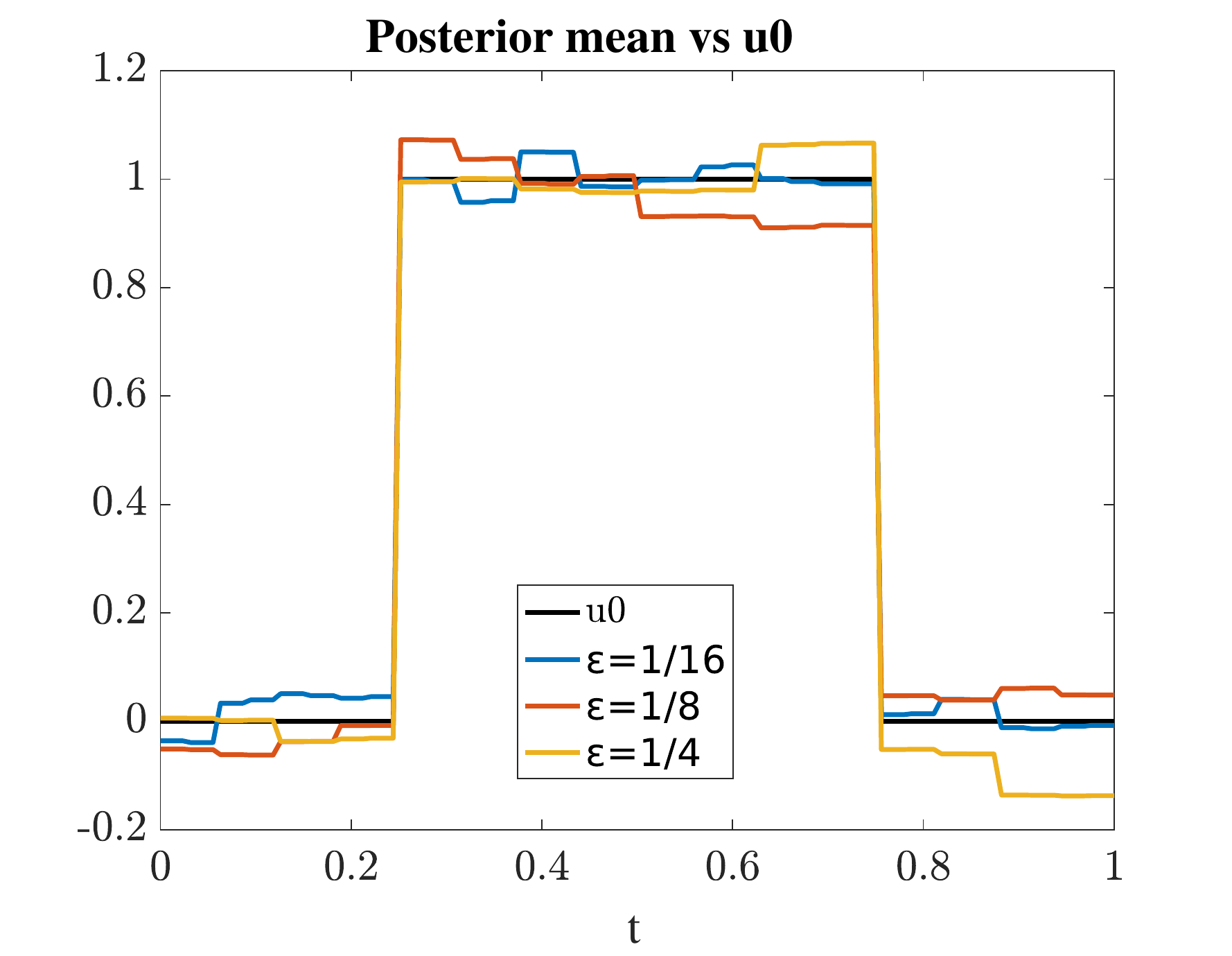}
\raisebox{.22\textwidth}{b)}
  \raisebox{.13\textwidth}{\begin{tabular}{ |c | c | c| c | c| }
              \hline
              $\varepsilon$ & $\min$ ESS & mean ESS & $\max$ ESS &mean $a(\cdot,\cdot)$ \\ \hline  
                 1/16 & 15 & 43  & 123& 0.27\\ \hline 
                 1/8 & 30 & 48 & 102& 0.31 \\\hline 
                 1/4 &  24& 47 & 87 & 0.37\\ \hline 
            \end{tabular}}
\qquad  \qquad
   \caption{a) Posterior mean of the deconvolution problem of Example 7 for different choices of 
     the kernel width $\varepsilon$. We observe more deviation from the
     true solution $u_0$ for larger values of $\varepsilon$. b) Certain statistics on the 
ESS of the components of the Markov chain for different values of $\varepsilon$ as 
well as the average acceptance ratio computed over $5 \times 10^5$ samples. All simulations
are performed with $N=32$ and $\beta = 0.97$.}
\label{tab:deconvolution-kw-study-mean-ESS}

\end{figure}

\begin{figure}[htp]
  \centering
  \includegraphics[width= .48 \textwidth]{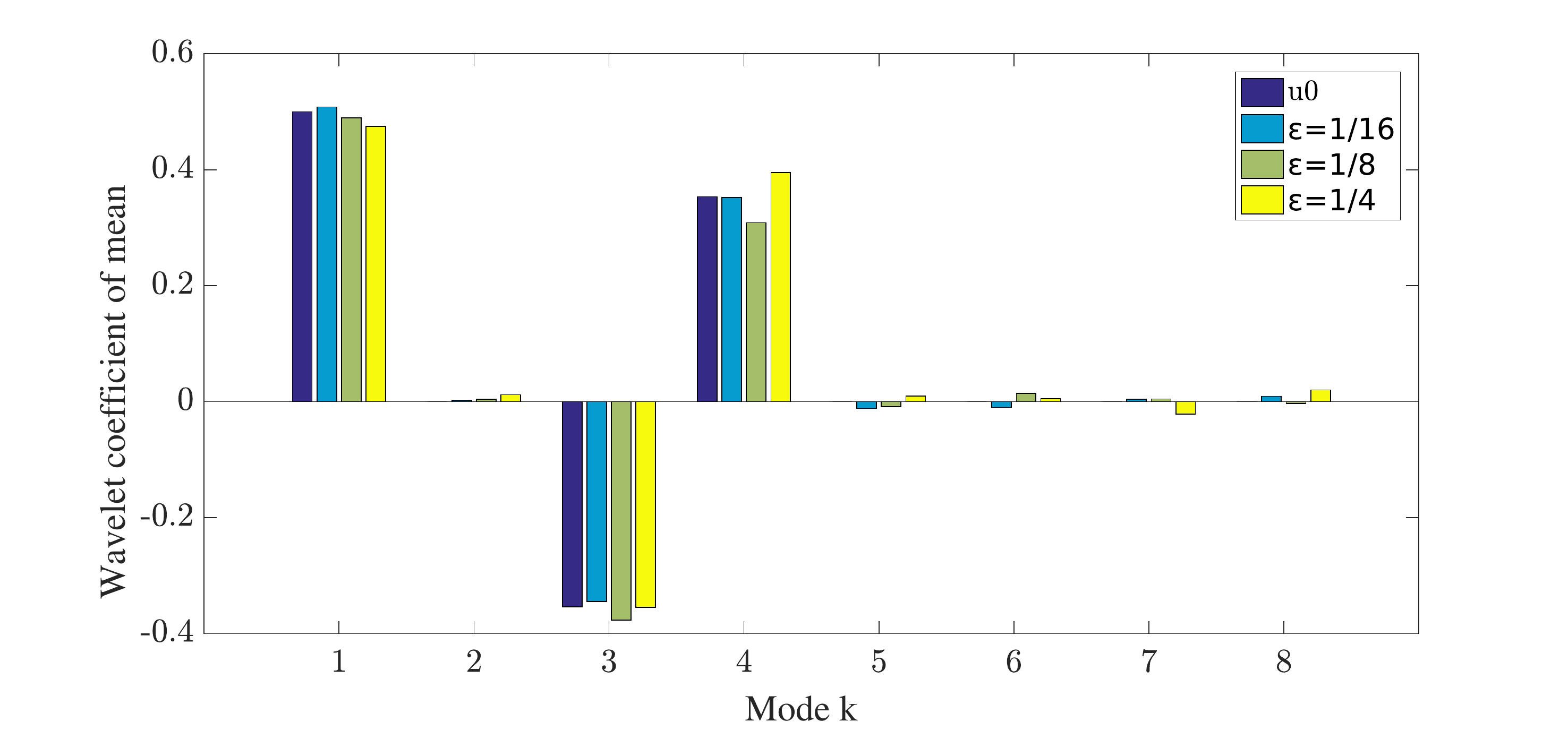}
\quad 
  \includegraphics[width= .48 \textwidth]{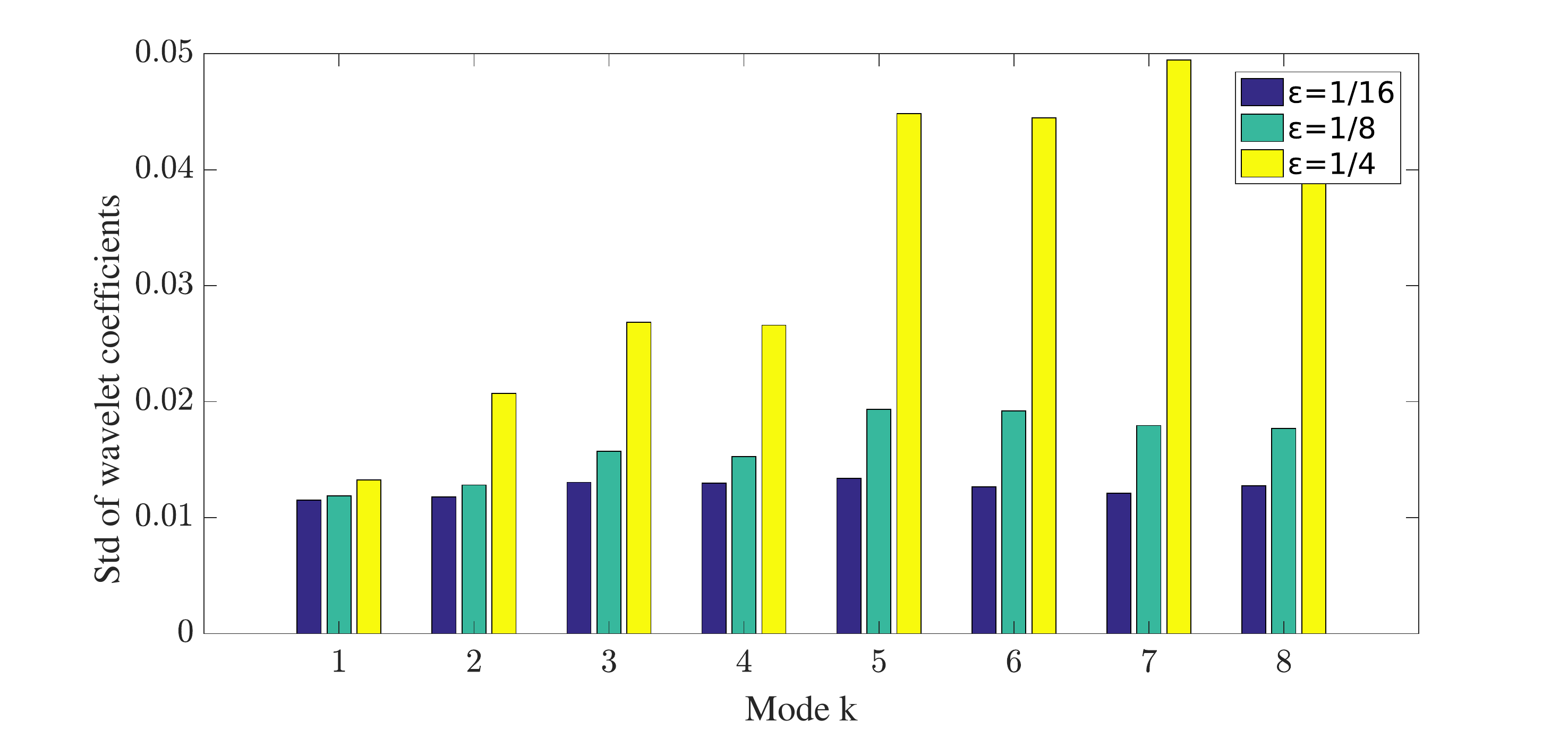}
  \caption{Posterior mean and standard deviation  of  the first eight wavelet modes in the deconvolution problem of Example 7 for different values of the kernel width $\varepsilon$.
    Posterior standard deviations appear to grow with $\varepsilon$. This is in line
    with the intuition that more smoothing in the forward map results in more
  uncertainty in the posterior due to less information in the data $y$.}
  \label{fig:deconvolution-kw-study-wavelet}
\end{figure}

\section{Closing remarks}
In the beginning of this article we set out to design algorithms for sampling  
measures $\nu$ that are absolutely continuous with respect to an underlying
non-Gaussian prior measure $\mu$. We  focused on the class of MH algorithms that utilize
$\mu$-reversible proposal kernels and showed that such MH algorithms are reversible
with respect to $\nu$ under mild conditions.
We then introduced two classes of algorithms called RCAR and SARSD that use
autoregressive type proposals that are $\mu$-reversible for certain
priors $\mu$. The RCAR algorithm is applicable to  gamma-type distributions
while SARSD can be applied when $\mu$ is SD. While SARSD is in principle
more widely applicable, it is often difficult to implement due to
issues with computing the time-reversal of AR(1) proposal kernels. 

Afterwards we introduced the Bessel-K priors as a concrete example of a 
non-Gaussian prior for which the RCAR and SARSD algorithms are applicable.
We further motivated  the Bessel-K  priors as interesting  
 candidates for modelling sparse or compressible parameters. We then 
 derived different versions of RCAR and SARSD  for the Bessel-K priors
 on Hilbert spaces and  studied the performance of both algorithms 
in  various numerical examples. 

\subsection{Future directions and open problems}
Our exposition  is a step towards the design of MCMC algorithms 
that are tailored to  
highly non-Gaussian priors. Our overall approach is that by analyzing the 
prior measure one can design proposals that result in  well-defined  and
efficient algorithms
in general 
state spaces.
This is in contrast with existing approaches in the \myhl{literature \cite{chen2018robust,
  marzouk-randomize-optimize} 
 that often introduce a non-linear mapping that transforms the prior to a Gaussian or another
 well-known measure
and modify the forward map or the likelihood to allow for application of } conventional sampling
algorithms with Gaussian priors. Our approach gives rise to many interesting 
questions that, to the knowledge of the author, have not been addressed in the 
literature. 

The first question is whether our approach can be extended to larger 
classes of prior measures. That is, whether it is possible to design RCAR or SARSD
algorithms for priors that belong to larger classes than 
the SD class or generalizations of the gamma distribution.
Good candidates here are the classes of infinitely-divisible 
 or convex priors that were discussed in 
 \cite{hosseini-convex, hosseini-sparse} or the stable priors of \cite{sullivan}.
 Another good candidate is  limit distributions of
 random coefficient autoregressive processes that are time-reversible.

 Another promising avenue of research is the design of likelihood aware
 proposal kernels for non-Gaussian priors in contrast with the 
prior preserving kernels of this work. The RCAR and SARSD algorithms 
perform well when the likelihood  is not 
dominant and the posterior  is close to the prior. Intuitively, this 
is due to the fact that the proposal kernel of RCAR and SARSD 
depends only on the prior and not the likelihood. Then a natural question 
is whether information regarding the gradient of the likelihood can be 
incorporated into the proposals to construct more efficient 
algorithms  similarly to MALA 
 or HMC. 

Finally, we note that the design and implementation of the SARSD algorithm relies on identifying   
 the innovation of the underlying 
 SD prior and the associated reverse kernel of the forward AR(1) proposal. A large portion of the
 existing literature on SD measures focuses on identifying the
 innovation  measures and so  a lot can be said about the forward kernel of
 SARSD. However, results on the reverse kernels are
 scarce and this is the main hurdle in implementing SARSD
 with most SD prior measures besides the exponential distribution and its
 extensions.

 \section*{Acknowledgements}
 The author is thankful to Profs. Derek Bingham, David Campbell, Nilima Nigam and Andrew M. Stuart 
 as well as Dr. James E. Johndrow and Sam Powers for interesting discussions and comments.
 We also owe a debt of gratitude to the anonymous reviewers whose careful
 comments and questions helped us improve an earlier version of this article.
 The author is also supported by a PDF fellowship granted by
 the Natural Sciences and Engineering Research Council of Canada.


\section*{References}
\bibliographystyle{abbrv}
       %
\bibliography{ref}


\appendix
\section*{Appendix}
\section{Self-decomposable measures}\label{SD-measures}

 We say 
 $\mu \in P(X)$ is SD if for every choice of $\beta$ there exists a  $\mu_\beta \in P(X)$ 
so that \cite{kumar-self} 
\begin{equation}\label{SD-characteristic}
  \widehat{\mu}(\varrho) = \widehat{\mu}(\beta \varrho) \widehat{\mu}_\beta(\varrho)
  \qquad \forall \varrho \in X^\ast.
\end{equation}
Here $\widehat{\mu}$ is the characteristic function of $\mu$  and $X^\ast$ is the dual of $X$.
The measure $\mu_\beta$ is referred to as the {\it innovation} of $\mu$.
{In other words, a random variable 
  $\xi$ is SD if, for every $\beta \in (0,1)$ there exists an independent random variable $\xi_\beta$
  so that 
the law of $\xi_\beta$ coincides with the law of $\beta \xi + \xi_\beta$.} The SD random variables 
are a subclass of infinitely-divisible random variables \cite{steutel} that were 
first introduced by Paul L{\'e}vy (the SD class is also known as 
the class of L{\'e}vy L probability measures). 
The SD class includes well-known probability measures such as 
 Gaussian measures and certain generalizations of the Laplace and gamma distributions. 
For the most part, the theory of SD measures was developed in the 1960s
 in connection to the theory of  L{\'e}vy processes. A detailed study of the 
 SD class can be found in the works of Kumar and Schreiber~\cite{kumar-self} and Urbanik~\cite{urbanik-self} as well as Jurek~\cite{jurek-selfdecomposability-exception-or-rule,jurek-different-aspects, jurek-integral} and
 Barndorff-Nielsen~\cite{barndor-multivariate, barndorf-selfdecomposable-free-probability}.
 We refer the reader to \cite[Ch.~V]{steutel} 
for an accessible introduction to real valued SD random variables.

\section{Some random variables on the real line}\label{app:standard-rvs}
Here we gather the definition of some standard random variables that are
used throughout the article in order to specify the  parameterizations
used in our article.

 \begin{definition}[Gamma random variable]\label{generalized-gamma}
  A positive random variable $\xi$ is distributed according to a 
  gamma distribution $G(p, \sigma)$ with  shape parameter $p > 0$
   and scale parameter $\sigma >0$ 
if its law has Lebesgue density 
\begin{equation}\label{generalized-gamma-density}
\Gamm(p, \sigma; t) = \frac{1}{\sigma \Gamma (p)} \left(
  \frac{t}{\sigma}\right)^{p-1} \exp \left( - \frac{t}{\sigma} 
\right)  \mb{1}_{(0, \infty)}(t) \qquad \text{for} \qquad t \in \reals.
\end{equation}
\end{definition}

 \begin{definition}[Exponential random variable]\label{exponential-def}
  A positive random variable $\xi$ is distributed according to a 
  exponential distribution $\Exp(\sigma)$ with parameter $\sigma > 0$ 
if its law has Lebesgue density 
\begin{equation}\label{exponential-density}
  \Exp( \sigma; t) = \frac{1}{\sigma}  \exp \left( - \frac{t}{\sigma} 
\right)  \mb{1}_{(0, \infty)}(t) \qquad \text{for} \qquad t \in \reals.
\end{equation}
\end{definition}

 \begin{definition}[Laplace random variable]\label{Laplace-def}
  A real valued random variable $\xi$ is distributed according to a 
  Laplace distribution $\Lap(\sigma)$ with parameter $\sigma > 0$ 
if its law has Lebesgue density 
\begin{equation}\label{Laplace-density}
  \Lap( \sigma; t) = \frac{1}{2\sigma}  \exp \left( - \left|\frac{t}{\sigma} \right| 
\right) \text{for} \qquad t \in \reals.
\end{equation}
\end{definition}

\begin{definition}[Beta random variable]\label{beta-def}
  A random variable $\xi \in (0,1)$ is distributed according to a 
  beta distribution $\Beta(p, q)$ with parameters $p,q > 0$ 
if its law has Lebesgue density 
\begin{equation}\label{beta-density}
  \Beta(p, q; t) = \frac{\Gamma(p+q)}{\Gamma(p)\Gamma(q)} t^{p-1}(1 - t)^{q-1} \mb{1}_{(0,1)}(t)
  \qquad \text{for} \qquad t \in \reals.
\end{equation}
\end{definition}

\begin{definition}[Bernoulli random variable]\label{bernoulli-def}
  A random variable $\xi \in \{0, 1\}$ is distributed according to a 
  Bernoulli distribution $\Bern(p)$ with parameter $p \in (0,1)$ 
if 
\begin{equation}\label{bernoulli-density}
  \mbb P(\xi = 0) = 1 - p, \qquad\mbb P(\xi = 1) = p.
\end{equation}
\end{definition}

 \begin{definition}[Poisson random variable]\label{poisson-def}
   A random variable $\xi \in \mbb N$ is distributed according to a Poisson
   distribution $\Pois(c)$ with rate $c > 0$ if
   \begin{equation*}
     \mbb{P}(\xi = k) = \frac{c^k \exp(-c)}{k!}, \qquad \forall k \in \mbb N.
   \end{equation*}
\end{definition}

\section{Properties of gamma and exponential random variables}
Here we collect some results on exponential and gamma distributions that are used
in the design of the RCAR and SARSD algorithm for 1D Bessel-K priors in Section~\ref{sec:bessel-k}.

\subsection{Properties of the gamma distribution}\label{app:gamma-dist}

Observe that
$\Gamm( 1, \sigma)$ coincides with the law of an 
exponential random variable $\Exp(\sigma)$. A straightforward calculation shows that
the 
gamma distribution has bounded raw moments of all orders, in fact
\begin{equation}\label{gamma-raw-moments}
\int_0^\infty t^k\Gamm(p,\sigma;t) \dd \Lambda(t) = \frac{\sigma^{k}\Gamma\left( {k + p}\right)}{\Gamma( p)}.
\end{equation}

\subsubsection{The thinned gamma process}\label{sec:thinned-gamma-process}
It is well-known \cite[Ch.~17.6]{kotz-univariate-v1} that given independent
random variables $\xi \sim \Gamm(p_1, \sigma)$ and $\xi' \sim \Gamm(p_2, \sigma)$
then $\xi + \xi' \sim \Gamm(p_1 + p_2, \sigma)$
and $\frac{\xi}{\xi + \xi'} \sim \Beta(p_1, p_2)$ and furthermore $\xi + \xi'$ and $\frac{\xi}{\xi + \xi'}$ are independent. In light of this fact we consider an RCAR(1)
process of the form
\begin{equation*}
  u^{(n)} = \zeta^{(n)} u^{(n-1)} + w^{(n)}, \qquad \zeta^{(n)} \stackrel{iid}{\sim} \Beta(p\beta, p(1-\beta)),\quad w^{(n)} \stackrel{iid}{\sim} \Gamm(p(1-\beta), \sigma),
  \qquad 
\end{equation*}
for fixed parameter $\beta \in (0,1)$. The limit distribution of this RCAR(1) process
is precisely $\Gamm(p, \sigma)$ \cite{lewis1989gamma}. Moreover, this process is
time-reversible and so the transition kernel
\begin{equation}\label{RCAR-proposal-kernel-for-gamma-distribution}
  \mcl{Q}(u, \dd v) = \Law \{v = \zeta u + w, \qquad \zeta \sim \Beta(p\beta, p (1-\beta)),
  w \sim \Gamm(p(1- \beta), \sigma)\},
\end{equation}
 satisfies detailed balance with respect to $\Gamm(p, \sigma)$.

 \subsubsection{Self-decomposability}
The gamma distribution is  SD. We 
demonstrate this using the characteristic function of $\Gamm(p, \sigma)$.
Recall
\begin{equation}\label{gamma-characteristic-function}
\widehat{\Gamm}( p,\sigma; s) 
= ( 1-  is\sigma)^{-p}.
\end{equation}
Choose $\beta \in (0,1)$ and
define 
\begin{equation}\label{gamma-innovation-characteristic-function}
\widehat{\Gamm}_\beta( p, \sigma; s) := \frac{\widehat{\Gamm}(
  p,\sigma;s)}{\widehat{\Gamm}(p, \sigma; \beta s)} = \left(\frac{1 -
  i \beta s\sigma}{ 1 - is\sigma } \right)^{p} =  \left( \beta +   
\frac{1 - \beta}{ 1  - is \sigma} \right)^p.
\end{equation}
It is  straightforward to check that
$$
\widehat{\Gamm}(p, \sigma; s) = \widehat{\Gamm}(p, \sigma;\beta s) \widehat{\Gamm}_\beta(p,\sigma;s).
$$
But $\widehat{\Gamm}(p, \sigma; \beta s)$ is simply the characteristic function of $\beta \xi$.
Furthermore, 
$\widehat{\Gamm}_\beta( p,\sigma; 0) = 1$ and $|\widehat{\Gamm}_\beta(p,\sigma; s)|
\le 1$. In fact, $\widehat{\Gamm}_\beta$ is continuous and differentiable at 0 and so it is the characteristic function of a
random variable on $\reals$. We denote this variable by $\xi_\beta$ and
its law by ${\Gamm}_\beta(p, \sigma)$. On this account, we have the following decomposition of gamma 
random variables:
\begin{equation}\label{gamma-self-decomposable}
\xi  \dequal \beta \xi' + \xi_\beta,
\end{equation}
where $\xi, \xi' \sim \Gamm( p,\sigma)$ and $\xi_\beta \sim \Gamm_\beta(p,\sigma)$ are independent 
of each other.  Lawrance
\cite{lawrance} showed that $\xi_\beta$ is a compound Poisson
random variable of the form
\begin{equation}\label{gamma-innovation}
\xi_\beta \dequal \sum_{k=1}^\tau \beta^{\eta_k} \theta_k \qquad \text{where}
\qquad \tau \sim \text{Pois}(p \log(1/\beta)), \quad \eta_k \sim U(0,1), \quad \theta_k \sim \Exp(\sigma),
\end{equation}
{where $\tau$, $\eta_k$ and $\theta_k$ are all independent and the latter sequences are
  identically distributed and $U(0,1)$ denotes the uniform distribution
  on $[0,1]$.}
Note that with the above expression we can simulate $\xi_\beta$ exactly. This is not possible for general SD random variables. We also note that a more efficient recipe for simulating $\xi_\beta$
was discovered by Walker \cite{walker2000note}.

\subsection{Self-decomposability of the exponential distribution}\label{app:exponential-dist}
Using the fact that $\Exp(\sigma)  = \Gamm(1, \sigma)$ we have that
the characteristic function of an $\Exp(\sigma)$ random variable is given by
\begin{equation*}
  \widehat{\Exp}(\sigma; s) = (1 - i s\sigma )^{-1}.
\end{equation*}
Using the factorization
\begin{equation*}
  \frac{1}{1 - i s}  = \left( \frac{1}{1 - i\beta s}  \right) \left( \beta +
    \frac{1- \beta}{1 - i t}\right)
\end{equation*}
for $\beta \in (0,1)$
we infer that $\Exp(1)$ is SD and
\begin{equation*}
  \Exp_\beta(1)  = \Law\{ \zeta w, \qquad \zeta \sim \Bern(1- \beta), w \sim \Exp(1) \}.
\end{equation*}
We can then easily extend this to other exponential distributions to get 
\begin{equation}\label{exponential-innovation}
  \Exp_\beta(\sigma) = \Law\{ \zeta w, \qquad \zeta \sim \Bern(1-\beta), w \sim \Exp(\sigma) \}.
\end{equation}
Thus the transition kernel
\begin{equation}\label{exp-forward-kernel}
  \mcl Q(u, \dd v) = \Law\{v= \beta u + \zeta w, \qquad \zeta \sim \Bern(1-\beta), w \sim \Exp(\sigma) \}
\end{equation}
preserves the distribution $\Exp(\sigma)$ but it does not satisfy detailed balance.
However, we can directly compute the reverse kernel $\mcl Q^\ast$.
Consider, $u, w \stackrel{iid}{\sim} \Exp(1)$ and $\zeta \sim \Bern(1- \beta)$ then 
\begin{equation*}
  v =  \beta u + \zeta w,
\end{equation*}
and $v \sim \Exp(1)$.
We wish to find an expression for $u | v$ but this is not trivial since $v$ is not
independent of $\zeta$ and $w$. Let $f(v|u)$ denote the distribution of $v$ given $u$, i.e., 
\begin{equation*}
  f(v|u) = \beta \delta_{\beta u}(v) + (1- \beta) \exp( - (t - \beta u)) \mb{1}_{[\beta u, \infty)}(v).
\end{equation*}
Using Bayes' rule we now have for $v >0$
\begin{equation*}
  \begin{aligned}
    f(u|v)   & =  \frac{f(v|u) f(u)}{f(v)} \\
    & = \frac{1}{\exp( -v)} \left(\beta \delta_{\beta u}(v) + (1- \beta) \exp( - (v - \beta u)) \mb{1}_{[\beta u, \infty)}(v)\right) \exp( - u) \mb{1}_{[0, \infty)}(u) \dd u   \\
    & =  \beta \delta_{\beta u}(v) \exp( - u + v) \dd u  + (1 - \beta) \exp(  - (1 - \beta) u ) \mb 1_{[\beta u, \infty)}(v)  \mb 1_{[0, \infty)}(u) \dd u   \\
    & = \beta \delta_{v/\beta}(u) \exp \left( - \frac{1 - \beta }{\beta}  v \right) \dd u  + (1 - \beta) \exp(  - (1 - \beta) u ) \mb 1_{[\beta u, \infty)}(v) \mb 1_{[0, \infty)}(u)  \dd u \\
    & = \frac{\beta}{1 - \beta} \delta_{v/\beta}\left( \frac{\zeta}{1 - \beta} \right) \exp\left( -  \frac{1 - \beta}{\beta}  \right) \dd \zeta  +  \mb 1_{[0, v/\beta]}\left( \frac{\zeta}{1- \beta} \right)  \exp( - \zeta )  \dd \zeta \\
    & = \mbb P \left( \frac{\zeta}{1 - \beta} \ge \frac{v}{\beta} \right) \delta_{v/\beta}\left(\frac{\zeta}{1- \beta} \right) +  \mb 1_{[0, v/\beta]}\left( \frac{\zeta}{1- \beta} \right)  \exp( - \zeta )  \dd \zeta \\
  \end{aligned}
\end{equation*}
where we used the change of variables $\zeta = (1 - \beta) u$. 
We now observe that the above expression is precisely the law of $\min \{ v/\beta, \zeta/(1 - \beta) \}$
where $\zeta \sim \Exp(1)$. Thus the reverse kernel $\mcl Q^\ast$  associated to \eqref{exp-forward-kernel}
can be identified as
\begin{equation}\label{exp-reverse-kernel}
  \mcl Q^{\ast}( v, \dd u) = \Law \left\{ u = \min \{ v/\beta, \zeta/(1 - \beta) \}, \qquad \zeta \sim \Exp(1) \right\}.
\end{equation}
\end{document}